\documentclass[12pt]{article}
\usepackage{graphicx}					
\usepackage{mathrsfs}
\usepackage{verbatimbox}
\usepackage{listofitems}
\usepackage{amsmath, amssymb}
\usepackage{wasysym} 
\usepackage{arydshln}
\usepackage{amsfonts, ifthen, amsthm, vmargin}
\usepackage[usenames]{color}
\usepackage{ulem}
\usepackage[colorlinks=true, linkcolor=blue3, 
urlcolor=gray2, citecolor=blue3, pagebackref=true]{hyperref} 
\newcommand\adj{{\displaystyle\ast}}
\newcommand\conj{{\scriptstyle\triangle}}

\newcommand{\bcdot}{\boldsymbol{\cdot}}
\newcommand{\sref}[2]{\hyperref[#2]{#1 \ref*{#2}}}
\usepackage[T1]{fontenc}
\usepackage{lmodern}
\newcommand{\id}[5]{{\frak{#1}}_{#2}^{{#3}{\tiny{\mbox{#4}}}\,{#5}}\!}
\newcommand{\idn}[5]{{\mathbf{#1}}_{#2}^{{#3}{\tiny{\mbox{#4}}}\,{#5}}\!}

\newcommand{\keywords}[1]{{\textbf{Key words and phrases.}} #1}
\newcommand{\msc}[1]{{\textbf{MSC2020 subject classifications.}} #1}
\newcommand{\dci}[1]{{\textbf{Declaration of competing interest.}} #1}
\newcommand{\aiuse}[1]{{\textbf{Declaration of AI use.}} #1}

\newcommand{\rbx}[2]{\raisebox{#1}{#2}}
\newcommand{\scbox}[2]{\scalebox{#1}{#2}}
\newcommand{\tup}{\textup}
\newcommand{\tp}[6]{{#1} \,{}_{{\rbx{#2}{#3}}}\underline{\otimes}{}_{{\rbx{#4}{#5}}}\, {#6}}

\renewcommand\labelenumi{\textup{(\roman{enumi})}}
\renewcommand\theenumi\labelenumi
\DeclareFontFamily{OT1}{rsfs}{} \DeclareFontShape{OT1}{rsfs}{m}{n}{
<-7> rsfs5 <7-10> rsfs7 <10-> rsfs10}{}
\DeclareMathAlphabet\mathcurl{OT1}{rsfs}{m}{n}

\newtheorem{theorem}{Theorem}[section]
\newtheorem{lemma}[theorem]{Lemma}
\newtheorem{corollary}[theorem]{Corollary}
\newtheorem{proposition}[theorem]{Proposition}

\theoremstyle{definition}
\newtheorem{example}[theorem]{Example}

\theoremstyle{definition}
\newtheorem{definition}[theorem]{Definition}
\theoremstyle{definition}
\newtheorem{remark}[theorem]{Remark}

\theoremstyle{remark}
\newtheorem*{cp*}{Commented Proof}

\theoremstyle{plain}
\newtheorem*{assump*}{Assumption}

\theoremstyle{plain}
\newtheorem*{lemma*}{Lemma}

\theoremstyle{plain}
\newtheorem*{cor*}{Corollary}

\theoremstyle{plain}
\newtheorem*{gt*}{The Grothendieck Inequality (Lindenstrauss-Pe{\l}czy\'{n}ski style)}

\theoremstyle{plain}
\newtheorem*{thmGI*}{The Grothendieck equality}

\theoremstyle{remark}
\newtheorem{problem}{\textup{\textbf{\normalsize P\scriptsize ROBLEM}}}
\newenvironment{problem*env}
  {\begin{problem}\label{op:test}}
  {\end{problem}}

\theoremstyle{remark}
\newtheorem*{noc*}{\textup{\textbf{\normalsize N\scriptsize O-CLONING THEOREM}}}

\theoremstyle{plain}
\newtheorem*{thmHI*}{The Haagerup equality}

\theoremstyle{plain}
\newtheorem*{obsv*}{Observation}

\theoremstyle{plain}
\newtheorem*{fact*}{Fact}

\theoremstyle{definition}
\newtheorem*{df*}{Definition}

\theoremstyle{definition}
\newtheorem*{ex*}{Example}

\theoremstyle{definition}
\newtheorem*{ack*}{Acknowledgments}

\theoremstyle{plain}
\newtheorem*{exkrv*}{Example (Krivine's constant)}

\theoremstyle{plain}
\newtheorem*{gi*}{Matrix Version of the Grothendieck equality}

\theoremstyle{plain}
\newtheorem*{guidex*}{Guiding Example}

\theoremstyle{plain}
\newtheorem*{prp*}{Proposition}

\theoremstyle{plain}
\newtheorem*{cprp*}{Proposition (correlation matrix version of GT and little GT)}

\theoremstyle{plain}
\newtheorem*{thm*}{Theorem}

\theoremstyle{remark}
\newtheorem*{thmTS*}{\textup{\textbf{\normalsize T\scriptsize HEOREM  \normalsize 
(T\scriptsize SIREL'SON, \normalsize 1980)}}}

\theoremstyle{plain}
\newtheorem*{thmgn*}{Theorem (Grothendieck, 1953; Niemi, 1983)}

\theoremstyle{plain}
\newtheorem*{c*}{Conjecture}

\theoremstyle{definition}
\newtheorem*{rem*}{Observation}

\theoremstyle{remark}
\newtheorem*{rem2*}{Remark}

\theoremstyle{plain}
\newtheorem*{arcsinex*}{Example (Stieltjes, 1889)}

\theoremstyle{plain}

\theoremstyle{plain}

\theoremstyle{plain}

\newcount\colveccount
\newcommand*\colvec[1]{
        \global\colveccount#1
        \begin{pmatrix}
        \colvecnext
}
\def\colvecnext#1{
        #1
        \global\advance\colveccount-1
        \ifnum\colveccount>0
                \\
                \expandafter\colvecnext
        \else
                \end{pmatrix}
        \fi
}

\definecolor{amaranth}{rgb}{0.9, 0.17, 0.31}
\definecolor{blue3}{rgb}{0.32, 0.09, 0.98}
\definecolor{blue4}{rgb}{0.13, 0.67, 0.8}
\definecolor{tangerine}{rgb}{1.0, 0.8, 0.0}
\definecolor{taupegray}{rgb}{0.55, 0.52, 0.54}
\definecolor{coolblack}{rgb}{0.0, 0.18, 0.39}
\definecolor{gray2}{rgb}{0.5, 0.5, 0.5}
\definecolor{cadetblue}{rgb}{0.37, 0.62, 0.63}
\definecolor{green2}{rgb}{0.0, 0.5, 0.0}
\definecolor{myred}{rgb}{0.8, 0.0, 0.0}
\definecolor{taupe}{rgb}{0.10, 0.50, 0.2} 

\definecolor{blue2}{rgb}{0.08, 0.38, 0.74}
\definecolor{mypink}{rgb}{0.858, 0.188, 0.478}
\definecolor{cite}{rgb}{1.0, 0.13, 0.32}

\newcommand\lb{\left(}
\newcommand\rb{\right)}

\newcommand\ignore[1]{}
\newcommand\A{\mathbb A}

\newcommand\C{\mathbb C}

\newcommand\R{\mathbb R}
\newcommand\T{\mathbb T}
\newcommand\M{\mathbb M}
\newcommand\N{\mathbb N}
\newcommand\Z{\mathbb Z}
\newcommand{\F}{\mathbb{F}}

\newcommand\E[2][\P]{\mathbb E_{#1}\left[ #2\right]}
\renewcommand\E{\mathbb E}
\renewcommand\P{\mathbb P}
\renewcommand\H{\mathbb H}

\renewcommand*{\Re}{\operatorname{Re}}
\renewcommand*{\Im}{\operatorname{Im}}

\newcommand\unit{1\hspace{-1.7mm}1}

\newcommand\emphas[1]{{\color{blue}#1}}

\makeatletter
\renewcommand{\@seccntformat}[1]{\csname the#1\endcsname.\hspace{1em}}%

\makeatother

\makeatletter
  \renewcommand*\env@matrix[1][*\c@MaxMatrixCols c]{%
    \hskip -\arraycolsep
    \let\@ifnextchar\new@ifnextchar
  \array{#1}}
\makeatother

\makeatletter
\def\widebreve{\mathpalette\wide@breve}
\def\wide@breve#1#2{\sbox\z@{$#1#2$}%
     \mathop{\vbox{\m@th\ialign{##\crcr
\kern0.08em\brevefill#1{0.8\wd\z@}\crcr\noalign{\nointerlineskip}%
                    $\hss#1#2\hss$\crcr}}}\limits}
\def\brevefill#1#2{$\m@th\sbox\tw@{$#1($}%
  \hss\resizebox{#2}{\wd\tw@}{\rotatebox[origin=c]{90}{\upshape(}}\hss$}
\makeatletter

\begin{document}



\title{{Beyond trace-class and Hilbert-Schmidt - Interaction between operator 
ideals and von Neumann algebras in quantum physics}}


\author{Frank Oertel  \\
	Philosophy, Logic \& Scientific Method\\
	Centre for Philosophy of Natural and Social Sciences (CPNSS)\\
	London School of Economics and Political Science\\
	Houghton Street, London WC2A 2AE, UK\\
	}

\date{} 

\maketitle

\hypersetup{linkcolor=coolblack} 
\tableofcontents
\cleardoublepage 
\pagenumbering{arabic} 
\begin{abstract}
\noindent 
Starting from a thorough analysis of the conjugate $\overline{H}$ of a complex 
Hilbert space $H$, including its significant importance regarding a 
representation of the tensor product of two complex Hilbert spaces and its 
impact to the theorem of Fr\'{e}chet-Riesz over to a revisit of applications of 
nuclear and absolutely $p$-summing operators in algebraic quantum field theory 
(AQFT) in the sense of Araki, Haag and Kastler ($p=2$) and more recently in the 
framework of general probabilistic spaces ($p=1$), we will outline that Banach 
operator ideals in the sense of Pietsch, or equivalently tensor products of 
Banach spaces in the sense of Grothendieck are even lurking in the foundations 
and philosophy of quantum physics and quantum information theory. In particular, 
we concentrate on their importance in AQFT 
(\sref{Theorem}{thm:split_property_and_2_summing_modular_op}). In doing so, we 
revisit the role of trace-class operators in quantum theory and construct 
the enveloping $\tup{C}^\adj$-algebra, corresponding to an arbitrarily given 
normed operator ideal 
(\sref{Proposition}{prop:Banach_op_ideal_components_on_H_are_Banach_ast_algebras} 
and \sref{Theorem}{thm:enveloping_C_star_alg_of_normed_op_ideal}).   
 
Applications are presented, including a purely linear algebraic description of 
the quantum teleportation process, thereby showing a link to quantum 
information theory, also due to the emergence of the Hadamard-Walsh transform 
and the controlled NOT gate (\sref{Example}{ex:quantum_teleportation_matrix}).

All Hilbert spaces discussed in this paper may be nonseparable (and hence 
infinite-dimensional). 
%
\end{abstract}
, tensor product of Hilbert spaces, Hilbert-Schmidt 
operator, {quasinormed operator ideal}, absolutely 2-summing operator, 
{$2$-nuclear operator}, no-cloning theorem, quantum teleportation, 
{algebraic quantum field theory, enveloping $\tup{C}^\adj$-algebra, von Neumann 
algebra, spatial tensor product, minimal tensor product, Tomita-Takesaki theory, 
general probabilistic theories}\\[0.20em]
\noindent
\msc{Primary 46C05, 47B10, {81T05}; 
secondary 46M05, {81P16}, 81P45.} 

\section{Introduction and brief summary}
The main theme of this work is to reveal quite surprising links between operator 
ideals in the sense of A. Pietsch, $\tup{C}^\adj$-algebras and subfields of 
quantum theory, primarily following from the very rich underlying structure of 
the tensor product of two complex Hilbert spaces, which originally emerges from 
its representation as a Hilbert-Schmidt operator between suitable Hilbert spaces. 
It is a well-known rule in quantum mechanics that the composition of two 
physical systems which are represented by Hilbert spaces $H$ and $K$, 
respectively, results in a physical system which is represented as the Hilbert 
space, modelled by the tensor product $H \otimes K$. An important example is 
given by the Fock space $\mathscr{F}(H_\pm)$ (symmetric and antisymmetric, 
respectively) that can be used to construct the Hilbert space of an unknown 
number of identical particles from a single particle Hilbert space $H_\pm$. By 
construction, it is given by the (Hilbert) direct sum of tensor products of 
copies of a complex single-particle Hilbert space $H_\pm$: $\mathscr{F}(H_\pm) 
: = \C \oplus \bigoplus\limits_{n=1}^\infty H_\pm^{\otimes n} \equiv \C \oplus 
H_\pm \oplus (H_\pm\otimes H_\pm) \oplus (H_\pm\otimes H_\pm \otimes H_\pm) 
\oplus\cdots$ (cf. \sref{Definition}{def:n_fold_tp}).
 
Possibly less well-known is the important meaning of the tensor product of 
Hilbert spaces over $\C$ as a building block of the spatial tensor product of 
von Neumann algebras and the minimal tensor product of $\tup{C}^\adj$-algebras, 
respectively, which play a significant role in algebraic quantum field theory 
in the sense of H. Araki, R. Haag and D. Kastler (cf. 
\sref{Remark}{rem:spatial_tp_and_minimal_tp}, \sref{Theorem}{thm:split_char} 
and \sref{Theorem}{thm:nuclearity_implies_split_property}). 

We reveal that the tensor product of Hilbert spaces $H \otimes K$ can 
actually be isometrically identified with the Hilbert space 
$(\id{P}{2}{}{}{}\,(\overline{H}, K), \idn{P}{2}{}{}{}\,(\cdot))$, 
where $(\id{P}{2}{}{}{}\,, \idn{P}{2}{}{}{}\,)$ denotes the 
injective, totally accessible and self-adjoint maximal Banach ideal of 
absolutely $2$-summing operators and $\overline{H}$ denotes the complex 
conjugate Hilbert space of $H$ (cf. \cite{DF1993, DJT1995, J1981, P1980}, 
\eqref{eq:structure_of_conj_HS}, \sref{Theorem}{thm:Riesz_linear_version} and 
\sref{Proposition}{prop:induced_tensor_norm_g2}). By implementing that isometric 
equality, canonical relations between separable and entangled composite quantum 
states and maps, which are known for quantum systems in finite-dimensional 
Hilbert spaces, can be generalised canonically to the more realistic situation 
of infinite-dimensional (and not necessarily separable) Hilbert spaces (cf. 
\cite[III.5.2.2]{B2006} and \cite{GKM2007}). 

From functional analysis and operator ideal theory in the sense of A. Pietsch 
it is known that absolutely $2$-summing operators - originally introduced by A. 
Grothendieck - play a fundamental role in the geometry of Banach spaces, 
implying that \sref{Proposition}{prop:induced_tensor_norm_g2} (and its natural 
extension to the modelling of a composition of multi-particle systems) 
illustrates a hidden, nontrivial geometric structure underlying the tensor 
product of Hilbert spaces and the spatial tensor product of von Neumann 
algebras in the complex case (cf. \sref{Theorem}{thm:N2_vs_P2_min}).

Firstly, we take a closer look at the conjugate Hilbert space 
$(\overline{H}, \langle\cdot, \cdot\rangle_{\overline{H}})$ of a given Hilbert 
space $(H, \langle\cdot, \cdot\rangle_H)$. By construction, the \textit{set} 
$\overline{H}$ coincides with the \textit{set} $H$, with the same definition of 
vector addition, yet with a modified scalar multiplication and a different 
inner product structure (cf. \eqref{eq:scalar_multipl_in_H_bar}, 
\eqref{eq:inner_product_on_the_conjugate_HS} and \eqref{eq:structure_of_conj_HS}). 

An implementation of the conjugate Hilbert space allows us to ``translate'' 
statements involving semilinear (or antilinear) operators, including 
antiunitary ones, into the much more familar language of \textit{linear} 
functional analysis and operator theory (including its entire structure, which 
results from the perspective of category theory). In particular, the famous 
duality of Fr\'{e}chet-Riesz, generated by a - semilinear - isometric 
isomorphism, can be \textit{canonically} transformed to a \text{linear} version, 
which therefore is independent of the choice of an underlying orthonormal basis 
(cf. \sref{Theorem}{thm:Riesz_linear_version}).

By making use of the linear version of the Fr\'{e}chet-Riesz duality 
$H^\prime \cong \overline{H}$,
we show next that $H \otimes K$ coincides \textit{linear-}isometrically 
with the Hilbert space $(\id{S}{2}{}{}{}\,(\overline{H}, K), \sigma_2)$, 
where $\id{S}{2}{}{}{}\,(\overline{H}, K)$ denotes the class of all Hilbert-Schmidt 
operators between the Hilbert spaces $\overline{H}$ and $K$ and $\sigma_2$ 
designates the Hilbert-Schmidt norm (cf. 
\sref{Theorem}{thm:Hilbert_Schmidt_and_tensor_product}, 
\sref{Corollary}{cor:char_of_tensor_product_of_HS_as_Hilbert_Schmidt} and 
\eqref{eq:HS_norm}).

With this crucial (linear) isometric identification as a starting point, an 
application of the isometric (linear) operator ideal identity 
\eqref{eq:HS_equals_P2} (cf. \cite[Proposition 11.6]{DJT1995}) instantly unfolds 
a fruitful application of the very powerful ``good old-school'' operator ideal 
machinery of Grothendieck and Pietsch\,! Equipped with a very few of their key 
results, enriched by somewhat hidden local structural accessibility properties 
of Banach operator ideals (cf. \cite{DF1993}), we recognise that the important 
role of trace-class operators (particularly that one in quantum physics - cf. 
\sref{Theorem}{thm:normal_states_on_vNAs} and \sref{Theorem}{thm:Gleason_Drisch}) 
can be linked with the much more general role of $p$-nuclear operators 
($1 \leq p < \infty)$, where the latter are defined between arbitrary Banach 
spaces- as opposed to the trace-class operators (cf. 
\sref{Theorem}{thm:a_corollary_of_Grothendieck}, 
\sref{Theorem}{thm:char_of_trace_class_ops}, 
\sref{Theorem}{thm:L2_circ_N_equals_P2_circ_P2} and
\sref{Corollary}{cor:L2_circ_N_equals_N2_circ_N2}). 

We will re-recognise that in particular $1$-nuclear operators and absolutely 
$2$-summing operators, mapping a von Neumann algebra into a Hilbert space over 
$\C$, play a significant role in algebraic quantum field theory (cf. 
\eqref{eq:nuclearity_condition_I}, \eqref{eq:nuclearity_condition_II},  
\sref{Theorem}{thm:nuclearity_implies_split_property} and 
\sref{Theorem}{thm:split_property_and_2_summing_modular_op}), and we 
sketch an implied open problem. 

In addition to an application of operator ideals in algebraic quantum field theory, 
we indicate how absolutely $1$-summing operators emerge in the investigation of 
compatibility of dichotomic measurements in general probabilistic theories 
(cf. \cite{BJN2022}). Moreover, we present a few applications of a concrete 
representation of the Hilbert space tensor product and list some of its algebraic 
properties, leading to a convenient, purely linear-algebraic description 
of quantum teleportation, which arises from a straightforward application of 
the Kronecker product of finite-dimensional Hilbert spaces and the associativity of the 
tensor product of Hilbert spaces (cf. 
\sref{Theorem}{thm:associativity_law_of_HS_tensorprod}, 
\eqref{eq:HS_tensor_product_vs_Kronecker_product} and 
\sref{Example}{ex:quantum_teleportation_matrix}).

Finally, while parts of the topics discussed in our paper are well-known to 
(some) researchers, we are not aware of any sources that thoroughly link them 
in this way, looking strongly through the window of operator ideal theory in 
the sense of Pietsch. Regarding the useful use of nonseparable Hilbert spaces 
in quantum physics (and its surprising implications), we highly recommend to 
study \cite{E2020}. Moreover, there are very recent studies regarding a 
development of a comprehensive framework for understanding the emergence of 
classical mechanics from quantum mechanics via so called 
\textit{coarse-grained measurements} (opening a new research field on its own). 
Within the scope of this research field, classicality has been shown to emerge 
in the mathematical limit of an infinite number of quantum systems, which can 
be described by a \textit{nonseparable} Hilbert space that is given by the 
tensor product of a Hilbert space family $\{H_n : n \in \N\}$; i.e., by the 
complex Hilbert space ${\bigotimes\limits_{n=1}^\infty}\,H_n$, developed in 
full detail in \cite[Chapter 2]{R2020}. A smart inductive limit construction 
of the infinite tensor product of complex Hilbert spaces is reflected in 
\cite[Exercise 11.5.29]{KR1986} (cf. also \cite{GV2023, E2020, G2025} and 
\eqref{eq:rep_of_the_n_fold_tens_prod}).

The remainder of our contribution is organised as follows. In 
\sref{Section}{sec:preliminaries_and_notation} we provide the necessary 
mathematical equipment, including a few basic facts from linear algebra and 
standard functional analysis. Already at this stage, we point to trace duality 
and the corresponding Frobenius matrix norm, where the latter coincides 
precisely with the Hilbert-Schmidt operator norm in the finite-dimensional 
case. In \sref{Section}{sec:overline_H_and_inear_version_of_Frechet_Riesz}, we 
revisit in greater depth the conjugate Hilbert space $\overline{H}$ (actually, 
a significant construction upon which all of our work is based), defined to be 
the same set as the given complex Hilbert space $H$, whereby the additive 
Abelian group structure of $H$ is maintained; in contrast to its scalar 
multiplication and its inner product. In analogy to the case of the Euclidean 
Hilbert space $\C^n$, it is possible to construct the inner product of 
$\overline{H}$ explicitly; even if the Hilbert space over $\C$ is 
infinite-dimensional and nonseparable. We will develop this construction in 
detail (cf. \eqref{eq:inner_prod_of_HS_conjugate_complex_elements}). Equipped 
with the conjugate Hilbert space, the composition of the corresponding 
antiunitary operator \eqref{eq:canonical_antiunitary} and the well-known 
semilinear Fr\'{e}chet-Riesz isomorphism $H \stackrel{\cong}{\longrightarrow} 
H^\prime, x \mapsto \langle \cdot, x \rangle_H$ between a complex Hilbert space 
$H$ and its (topological) dual $H^\prime$ naturally induces a - linear - 
canonical isometric isomorphism between the Hilbert space $H$ and the dual space 
of $\overline{H}$ (\sref{Theorem}{thm:Riesz_linear_version}). 
\sref{Section}{sec:HS_ops_and_tp_of_Hilbert_spaces} serves to provide a short 
and strongly simplified (yet still rigorous) construction of the $n$-fold 
tensor product $\bigotimes\limits_{i=1}^n H_i$ of Hilbert spaces, originating 
from an implementation of the structure of the conjugate Hilbert space. The 
inclusion of the latter namely implies that we may characterise the tensor 
product of two complex Hilbert spaces $H$ and $K$ precisely as the complex 
Hilbert space of all Hilbert-Schmidt operators between $\overline{H}$ and $K$ 
(\sref{Theorem}{thm:Hilbert_Schmidt_and_tensor_product} and 
\sref{Corollary}{cor:char_of_tensor_product_of_HS_as_Hilbert_Schmidt}). Already 
at this stage, an important link between operator ideals in the sense of 
Pietsch and the tensor product of Hilbert spaces emerges, where the former are 
encoded in the class of all Hilbert-Schmidt operators (examined in detail 
in \sref{Section}{sec:op_ideals_and_Grothendieck}). Implications of this 
crucial representation are given, including an infinite dimensional version of 
the Schmidt decomposition (\sref{Corollary}{cor:Schmidt_decomposition_revisited}) 
and an explicit description of elements $x \otimes y \in \C_2^n \otimes \C_2^m$ 
as the noncommutative Kronecker product of $x \in \C^n \cong \M_{n,1}(\C)$ and 
$y \in \C^m \cong \M_{m,1}(\C)$ (\sref{Example}{ex:Kronecker_product}). 
A straightforward application of the Kronecker product results in a purely 
linear algebraic description of quantum teleportation. Here, we characterise 
the corresponding unitary matrix in $\M_8(\C)$ and reveal its surprising 
composition structure, built by well-known factors from quantum information 
theory (\sref{Example}{ex:quantum_teleportation_matrix}). 
\sref{Section}{sec:op_ideals_and_Grothendieck} - the heart of our contribution - 
concentrates on creating new connections between operator ideal theory in the 
sense of Pietsch, the corresponding trace duality and the structure of 
$\tup{C}^\adj$-algebras (especially of that one of von Neumann algebras), with 
an application in algebraic quantum field theory (AQFT) in the sense of Araki, 
Haag and Kastler in mind 
(\sref{Theorem}{thm:enveloping_C_star_alg_of_normed_op_ideal}, 
\sref{Proposition}{prop:weak_star_functionals_on_max_BOI_components} and 
\sref{Theorem}{thm:split_property_and_2_summing_modular_op}). Actually, the 
latter theorem was coined by H. J. Borchers and R. Schumann, resulting in 
Schumann's possibly underrated doctoral thesis (cf. \cite{BS1991, S1994, S1996}). 
To improve the understanding of this result, the necessary tools are described, 
including the corresponding examples of operator ideals, generalising the 
class of Hilbert-Schmidt operators (cf. e.g., 
\sref{Proposition}{prop:Banach_op_ideal_components_on_H_are_Banach_ast_algebras}, 
\sref{Theorem}{thm:char_of_trace_class_ops}, 
\sref{Corollary}{cor:vector_state_version_of_Gleason} and 
\sref{Proposition}{prop:induced_tensor_norm_g2}), the spatial tensor product of 
two von Neumann algebras (whose construction is built on the tensor product of 
Hilbert spaces - \sref{Remark}{rem:spatial_tp_and_minimal_tp}), the split 
property of two commuting von Neumann algebras \eqref{eq:split_property_of_vNAs} 
and required key facts from Tomita-Takesaki modular theory 
(\sref{Theorem}{thm:Tomita_Takesaki_1967}). Finally, in 
\sref{Section}{sec:GPTs_and_P1}, we roughly sketch how the famous operator 
ideal of absolutely $1$-summing operators also emerge in the investigation of 
compatibility of dichotomic measurements in general probabilistic theories 
(cf. \cite{BJN2022}).             
\section{Preliminaries and notation}
\label{sec:preliminaries_and_notation}
As is usual, we denote the set of complex numbers by $\C$ and the set of real 
numbers by $\R$. $\Z$ represents the set of all integers and $\N$ stands for the 
subset of positive integers. $\N_0 : = \{0\} \cup \N$ denotes the set of all 
nonnegative integers (often also somewhat unhappily denoted as $\Z_+$). If 
$m \in \N$, we put $[m] : = \N \cap [1, m] = \{1, 2, \ldots, m\}$. We will use 
the symbol $\F$ to denote either the real field $\R$ or the complex field $\C$. 
If we wish to state a definition or a result that is satisfied for either real 
or complex numbers (i.e., if $\F \in \{\R, \C\}$), we simply will make use of 
the letter symbol $\F$.

Fix $m, n \in \N$. The set of all $m \times n$-matrices with entries in a given 
nonempty subset $S \subseteq \F$ is denoted by $\M_{m,n}(S)$. The matrix ring 
$\M_{n,n}(\F)$ is abbreviated as $\M_n(\F)$. As usual, $e_i \in \F^n$ denotes 
the column vector having a $1$ in the $i$th place and zeros elsewhere. If we 
wish to emphasize the dependence on the dimension $n$ of the vector space 
$\F^n$, then we speak of the set $\{e_1^{(n)}, e_2^{(n)},\dots, e_n^{(n)}\} 
\subseteq \F^n$. $I_n : = 
(e_1 \,\brokenvert\, e_2 \,\brokenvert\, \cdots \,\brokenvert\, e_n) \in 
\M_n(\F)$ describes the identity matrix. Initially, if not indicated 
otherwise, any vector $x \in \F^n$ in this manuscript is set as column vector, 
so that the allocated row vector is decribed by transposition ($x \mapsto 
x^\top$). Translated into Dirac's bra-ket language, it holds that $e_i = 
\vert i-1 \rangle$ and $e_i^\top = \langle i-1 \vert$ ($i \in [n]$). In 
particular, $\vert 0 \rangle = e_1$, $\vert 1 \rangle = e_2$ and $\vert n-1\rangle
\langle 1 \vert = e_n e_2^\top \in \M_{n,2}(\F)$. 

If $A \in \M_{m,n}(S)$ is given, it is sometimes very fruitful to represent the 
entries of $A$ as $A_{ij} : = e_i^\top A e_j = e_j^\top A^\top e_i = 
(A^\top)_{ji}$, so that $A = (a_{ij})$, where $a_{ij} : = A_{ij}$. 
$\overline{A} \in \M_{m,n}(\F)$ is defined as $\overline{A}_{ij} : = 
\overline{A_{ij}}$, implying that $A^\adj : = \overline{A}^\top = 
\overline{A^\top}$ and $x^\adj : = \overline{x}^\top = \overline{x^\top}$. 
Recall that the Euclidean norm is given by $\Vert x \Vert_2 : = 
\sqrt{x^\adj\,x} = \sqrt{\sum\limits_{i=1}^n \vert x_i \vert^2}$ for any $x = 
(x_1, \ldots, x_n)^\top \in \F^n$. If we equip the $n$-dimensional vector space 
$\F^n$ with the Euclidean inner product, we obtain the $n$-dimensional Hilbert 
space  $\F_2^n : = (\F^n, \langle \cdot, \cdot\rangle_2)$, where the inner 
product on $\F^n$ is given by $\langle x, y\rangle_2 : = \langle x, y\rangle_
{\F_2^n} : = y^\adj x = \sum\limits_{i=1}^n x_i \overline{y_i}$. If there is no 
risk of confusion regarding $\F$, we simply speak of the space $l_2^n$ (as 
usual). Throughout the manuscript, we also identify any linear operator 
$T : \F^n  \longrightarrow \F^m$ with its representing matrix with respect to 
the respective standard bases: $T \equiv (T_{ij})_{(i,j) \in [m] \times [n]}$. 
In particular, we have
\[
T_{ij} \equiv e_i^\top Te_j = \langle Te_j, e_i\rangle_{\F_2^m} = 
\langle e_j, T^\adj e_i\rangle_{\F_2^n} \equiv (T^\adj)_{ji} \text{ for all } 
i,j \in [m] \times [n]\,.
\]
As usual, $O(n)$ denotes the orthogonal group, consisting of all invertible 
matrices $A \in \M_n(\R)$ such that $A^{-1} = A^\top$. $U(n)$ describes the 
unitary group, consisting of all invertible matrices $A \in \M_n(\C)$ such 
that $A^{-1} = A^\adj$. $SO(n) : = \{A \in O(n) : \text{det}(A) = 1\}$ is the 
special orthogonal group, and $SU(n) : = \{A \in U(n) : \text{det}(A) = 1\}$ 
describes the special unitary group.

Throughout the manuscript all Hilbert spaces, unless otherwise specified, are 
Hilbert spaces over the field $\F$ (short: $\F$-Hilbert spaces). They may be 
nonseparable (and hence infinite-dimensional). To avoid a risk of confusion, 
we adopt the convention that is used in most mathematical contexts; namely that 
\textit{any} (well-defined) inner product $\langle \cdot, \cdot\rangle_V : V 
\times V \longrightarrow \C$ on a complex vector space $V$ is 
\textit{linear in the first argument} and 
\textit{conjugate linear in the second one}\,:
\[
\langle \lambda x, \mu y \rangle_V = \lambda \overline{\mu}\,
\langle x, y\rangle_V\,\text{ for all } \lambda, \mu \in \C \text{ and } 
x, y \in V.
\]
If $E$ is an arbitrary normed space, then $S_E : = \{x \in E : \Vert x \Vert_E = 1\}$ 
denotes its unit sphere and $B_E : = \{x \in E : \Vert x \Vert_E \leq 1\}$ its 
closed unit ball.

An important inner product on the $\F$-vector space $\M_{m,n}(\F)$, which turns 
$\M_{m,n}(\F)$ into an $mn$-dimensional Hilbert space, is the Frobenius inner 
product, which is defined as follows: if $A = (a_{ij}) \in \M_{m,n}(\F)$ and 
$B = (b_{ij}) \in \M_{m,n}(\F)$, then
\[
\langle A, B\rangle_F : = \text{tr}(A B^\adj) = \text{tr}(B^\adj A) = 
\sum\limits_{i=1}^m \sum\limits_{j=1}^n a_{ij}\,\overline{b_{ij}} = 
\overline{\langle B, A\rangle_F}\,,
\]
where 
\[
\text{tr}(C) : = \sum\limits_{i=1}^n \langle C e_i, e_i\rangle_2 = 
\sum\limits_{i=1}^n c_{ii} = \text{tr}(C^\top) = \overline{\text{tr}(C^\adj)} = 
\overline{\text{tr}(\overline{C})}
\] 
denotes the (nonnormalised) trace of a given (quadratic) matrix $C = (c_{ij}) 
\in \M_{n,n}(\F)$. Note that $\text{tr}(AD) = \sum\limits_{i=1}^m \sum\limits_
{j=1}^n a_{ij} d_{ji}$ for all $(A, D) \in \M_{m,n}(\F) \times \M_{n,m}(\F)$. 
 
We adopt the symbolic notation (commonly used by researchers in operator ideal 
theory and geometry of Banach spaces) 
to represent the set of all bounded linear operators between two normed spaces 
$(E, \Vert\cdot\Vert_E)$ and $(F, \Vert\cdot\Vert_F)$ by $\mathfrak{L}(E,F)$, 
and for the identity operator on $E$, we write $Id_E$. The collection of all 
finite rank (resp. approximable) operators from $E$ to $F$ is denoted by 
$\id{F}{}{}{}{}(E,F)$ (resp. $\overline{\id{F}{}{}{}{}}{}{}{}{}(E,F) : = 
\overline{\id{F}{}{}{}{}(E, F)}^{\Vert\cdot\Vert}$). 
The (topological) dual of a Banach space $E$ is denoted by $E^{\prime }$, and 
$E^{\prime \prime }$ denotes its bidual $(E^{\prime })^{\prime }$. If $J \in 
\mathfrak{L}(E,F)$ is an operator, we indicate that it is an isometry (or 
metric injection) by writing $J:E\stackrel{1}{\hookrightarrow }F$. If $Q \in 
\mathfrak{L}(E,F)$ is a quotient map (or metric surjection), we write 
$Q:E \stackrel{1}{\twoheadrightarrow} F$. Recall that any quotient map $Q\in 
\mathfrak{L}(E,F)$ is onto and maps the open unit ball of $E$ onto the open unit 
ball of $F$  and that $J \in \mathfrak{L}(E,F)$ is an isometry if and only if 
$J^\prime \in \mathfrak{L}(F^\prime,E^\prime)$ is a quotient map (cf., e.g., 
\cite[Chapter B.3.]{P1980}). 
Throughout the manuscript, we apply the common and very useful dual 
pairing notation 
\[
\langle x, a\rangle \equiv a(x)\,\hspace{1em}\,((x,a) \in E \times E^\prime)
\]
with angled brackets. (A confusion with inner products 
$\langle  \cdot, \cdot\rangle_H$ on a Hilbert space $H$ shouldn't arise.)
 
If $E$ is a Banach space, $M$ a finite dimensional subspace of $E$ and $K$ a 
finite codimensional subspace of $E$, then $J_M^E : M 
\stackrel{1}{\hookrightarrow }E$ denotes the canonical metric injection, 
$Q_K^E : E\stackrel{1}{\twoheadrightarrow } E\diagup K$ the canonical metric 
surjection and $j_E : E \stackrel{1}{\hookrightarrow } E^{\prime\prime}$ the 
canonical, yet very important isometric embedding from $E$ into its bidual 
$E^{\prime\prime}$, whose canonical mapping rule is given by 
\[
\langle a, j_E x\rangle : = \langle x, a\rangle
\] 
for all $(x, a) \in E \times E^\prime$. 
Finally, $T^{\prime } \in \mathfrak{L}(F^{\prime }, E^{\prime})$ denotes the 
dual operator of $T\in \mathfrak{L}(E,F)$. For the sake of completeness, we 
provide a quick and simple proof of the following fact.
\begin{proposition}\label{prop:isom_isomorphisms_between_Banach_spaces} 
Let $E, F$ be Banach spaces and $T \in \mathfrak{L}(E,F)$. Then the 
following statements are equivalent:
\begin{enumerate}
\item $T \in \mathfrak{L}(E,F)$ is an isometric isomorphism.
\item $T^\prime \in \mathfrak{L}(F^\prime,E^\prime)$ is an isometric isomorphism.
\end{enumerate}
\end{proposition}
\begin{proof}
(i) $\Rightarrow$ (ii): A straightforward calculation shows that $(T^\prime)^
{-1} = (T^{-1})^\prime$ is the well-defined inverse of $T^\prime$. In 
particular, $T^\prime$ is onto. 
Let $b \in S_{F^\prime}$. Since $T$ is an onto isometry, it follows that 
\[
\Vert b \Vert = \sup\limits_{y \in S_F}\vert\langle y, b\rangle\vert = 
\sup\limits_{x \in S_F}\vert\langle Tx, b\rangle\vert = \Vert T^\prime b\Vert\,.
\]
Thus, $T^\prime$ is an isometry.
\\[0.2em]
(ii) $\Rightarrow$ (i): According to what has just been proven, also 
$T^{\prime\prime}$ is an isometry. Consequently,
\[
\Vert Tx\Vert = \Vert j_FTx\Vert = \Vert T^{\prime\prime}(j_E x)\Vert =
\Vert j_E x\Vert = \Vert x\Vert \,\text{ for all }\, x \in E\,,
\]
implying that $T$ is an isometry. Since $T^\prime(F^\prime) = E^\prime$ is 
closed and $T^\prime$ has a zero kernel, a direct application of the closed 
range theorem, implies that $T$ in fact is onto (cf., e.g., 
\cite[Theorem IV.5.1]{W2011}).
\end{proof}
The ``$\mathfrak{B}$-community'', often encountered among researchers in the 
field of operator algebras, uses $\mathfrak{B}(E,F)$ instead, so that for 
example, $\mathfrak{B}(H) \equiv \mathfrak{L}(H) : = \mathfrak{L}(H,H)$, where 
$H$ is a given Hilbert space. Remember that every linear operator $T : E_0 
\longrightarrow F$ from a finite-dimensional normed space $E_0$ to an arbitrary 
normed space $F$ already is bounded. Recall that an operator $T \in 
\mathfrak{L}(H)$ is by definition \textit{positive}, if $\langle Tx, x\rangle 
\geq 0$ for all $x \in H$. 

Let us also recall that a map $T : V \longrightarrow W$ between two vector 
spaces $V$ and $W$ over $\C$ is \textit{semilinear} (also known as 
\textit{antilinear} or \textit{conjugate linear}) if $T(v+w) = T(v) + 
T(w)$ and $T(\lambda v) = \overline{\lambda}\,T(v)$ for all $(v,w) \in 
V \times W$ and $\lambda \in \C$. If $H$ and $K$ are Hilbert spaces over $\C$, 
then a semilinear operator $W : H \longrightarrow K$ is \textit{antiunitary} if 
the following conditions are satisfied:
\begin{enumerate}
\item $\langle W x, W y \rangle_K = \overline{\langle x, y \rangle_H} = 
\langle y, x\rangle_H$ for all $x, y \in H$.
\item $W$ is onto.
\end{enumerate}
In particular, $\Vert W x\Vert_K = \Vert x\Vert_H$ for 
all $x \in H$, implying that $W^{-1} : K \longrightarrow H$ is a well-defined 
antiunitary operator as well. Thus, if $L$ is a third Hilbert space over $\C$, 
and if $X : K \longrightarrow L$ is a further antiunitary operator, it is 
trivial that the (linear) composed operator $X W: H 
\stackrel{\cong}{\longrightarrow} L$ is a unitary operator (i.e., an isometric 
isomorphism). In fact, we will recognise that an application of the complex 
conjugate Hilbert space $\overline{L}$ (cf. 
\eqref{eq:structure_of_conj_HS}), together with the corresponding canonical 
antiunitary operator \eqref{eq:canonical_antiunitary} reveals that conversely 
every unitary operator $U : H \stackrel{\cong}{\longrightarrow} L$ can be 
written as a composition of two antiunitary operators, independently of the 
choice of the orthonormal basis of $H$ (namely as $U = J_{\overline{L}}
(J_L U)$). If in addition, $H = K$, and if $W$ is an \textit{involution} 
(i.e., if $W^2 = Id_H$), then the antiunitary operator $W : H \longrightarrow 
H$ is called a \textit{conjugation} (cf. \cite[Definition 2.1]{GPP2014}). It is 
well-known that the class of conjugations plays a fundamental role in the study 
of symmetries in quantum mechanics, such as time reversal (including Wigner's 
Theorem and Uhlhorn's Theorem - \cite{P2020, R2022}).
    
Throughout the manuscript, we apply the following characterisation of positive 
semidefinite (respectively positive definite) matrices, which does not depend 
on the choice of the field $\F \in \{\R, \C\}$ (cf. \cite[Corollary 1.1]{Oe2024}): 
\begin{lemma}\label{lem:uniform_psd_char}
Let $\F \in \{\R, \C\}, n \in \N$ and $A \in \M_n(\F)$. $A$ is 
positive semidefinite $($resp. positive definite$)$ in $\M_n(\F)$ if the 
following two conditions are satisfied:
\begin{enumerate}
\item $A = A^\adj$.
\item $z^\adj A z \geq 0$ $($resp. $z^\adj A z > 0$$)$ for all 
$z \in {\F^n}\setminus\{0\}$.
\end{enumerate}
In particular, $A \in \M_n(\R)$ is positive semidefinite in $\M_n(\R)$, if and 
only if $A \in \M_n(\R) \subseteq \M_n(\C)$ is positive semidefinite in $M_n(\C)$.
\end{lemma}
\begin{remark}
Based on our identification of (bounded) linear operators $A \in {\mathfrak{L}}
(\F^n, \F^n)$ and matrices $A \in M_n(\F)$, it follows that $A \in {\mathfrak{L}}
(\F^n, \F^n)$ is positive semidefinite if and only if $A$ is a positive 
self-adjoint operator. Since any positive operator $A \in {\mathfrak{L}}
(\C^n, \C^n)$ already is self-adjoint, positivity coincides with positive 
semidefiniteness on ${\mathfrak{L}}(\C^n, \C^n) \equiv \M_n(\C)$, in contrast 
to positivity on ${\mathfrak{L}}(\R^n, \R^n) \equiv \M_n(\R)$; i.e., there are 
positive nonsymmetric operators $B \in {\mathfrak{L}}(\R^n, \R^n)$, such as 
$B : = \begin{pmatrix}
0 & 1\\
-1 & 0
\end{pmatrix}$, implying that these operators cannot be positive semidefinite 
in $M_n(\C)$.  
\end{remark}
In general, if $M$ is a subset of an ordered vector space $(N, \geq)$, then 
$M^+ : = M \cap N^+$, where $N^+ : = \{x \in N : x \geq 0\}$ denotes the 
positive cone of $N$, such as 
\[
M_n(\C)^+ = \{A \in M_n(\C) : z^\ast A z \geq 0 \text{ for all } z \in \C^n\}
\]
(the PSD cone of $M_n(\C)$) and
\[
\M_n(\R)^+ = \{A \in M_n(\R) : A = A^\top \text{ and } x^\top A x \geq 0 
\text{ for all } x \in \R^n\}
\]
(the PSD cone of $M_n(\R)$). 

At this point, it is worth noting that actually $\langle A, B\rangle_F = 
\langle A, B\rangle_{\id{S}{2}{}{}{}\,}$ coincides with the Hilbert-Schmidt 
inner product, defined on the Hilbert space $\mathfrak{L}(\F_2^n, \F_2^m)$ (cf. 
\eqref{eq:HS_inner_product}). One can easily verify the well-known fact that the 
set of all elementary matrices $\{e_i\,e_j^\top : (i,j) \in [m] \times [n]\}$ is 
an orthonormal basis in the $mn$-dimensional $\F$-Hilbert space 
$(\M_{m,n}(\F), \Vert \cdot \Vert_F)$ (since $(e_i\,e_j^\top)_{\alpha\beta} = 
\delta_{i\alpha}\,\delta_{j\beta}$ for all $(\alpha,\beta) \in [m] \times [n]$ 
and $\text{tr}(xy^\top) = y^\top x$ for all $x, y \in \F^n$). Let us also note 
the easy-to-prove fact that
\[
\Pi_{m,n}: (\M_{m,n}(\F), \Vert \cdot \Vert_F)  \stackrel{\cong}{\longrightarrow} 
(\M_{m,n}(\F), \Vert \cdot \Vert_F)^\prime, A \mapsto (B \mapsto \text{tr}(B A^\top))
\]
is an isometric isomorphism, whose inverse is given by
\[
\M_{m,n}(\F) \ni \Pi_{m,n}^{-1}(t) = (t(e_i e_j^\top))_{i,j} \text{ for all } t \in 
(\M_{m,n}(\F), \Vert \cdot \Vert_F)^\prime\,.
\]
Also the dual operator of $(\M_{n,m}(\F), \Vert \cdot \Vert_F) 
\stackrel{\cong}{\longrightarrow} (\M_{m,n}(\F), \Vert \cdot \Vert_F), 
M \mapsto M^\top$
\[
\Theta_{m,n} : (\M_{m,n}(\F), \Vert \cdot \Vert_F)^\prime \stackrel{\cong}{\longrightarrow}
(\M_{n,m}(\F), \Vert \cdot \Vert_F)^\prime, t \mapsto (M \mapsto \langle M^\top, 
t \rangle) 
\]
is an isometric isomorphism (cf., e.g., 
\sref{Proposition}{prop:isom_isomorphisms_between_Banach_spaces}), and the 
classical Fr\'{e}chet-Riesz theorem, characterising the dual of a Hilbert space 
(cf., e.g., \sref{Theorem}{thm:Riesz_linear_version}) implies 
that the composition $\Theta_{m,n}\Pi_{m,n}$ actually reflects the 
finite-dimensional version of \textit{trace duality}\,:
\[
\Theta_{m,n}\Pi_{m,n} : (\M_{m,n}(\F), \Vert \cdot \Vert_F)  
\stackrel{\cong}{\longrightarrow} (\M_{n,m}(\F), \Vert \cdot \Vert_F)^\prime, 
A \mapsto (B \mapsto \tup{tr}(A B))
\]
is an isometric isomorphism, whence 
\begin{align*}
\idn{P}{2}{}{}{}\,(A) &= \idn{S}{2}{}{}{}\,(A) = \Vert A \Vert_F = 
\Vert\Theta_{m,n}\Pi_{m,n}(A)\Vert\\ 
&= \sup\big\{\vert\tup{tr}(A B)\vert : B \in \M_{n,m}(\F)\,\text{ and }\,
\Vert B\Vert_F \leq 1\big\}\\
&= \max\big\{\vert\tup{tr}(A B)\vert : B \in \M_{n,m}(\F)\,\text{ and }\,
\Vert B\Vert_F \leq 1\big\}
\end{align*}
for all $A \in \M_{m,n}(\F)$ (cf. \cite[Theorem 5.30 and Theorem 6.4]{DJT1995}, 
\sref{Theorem}{thm:trace_duality_in_the_Hilbert_space_case} and 
\eqref{eq:HS_equals_P2}). 
\section{The conjugate of a complex Hilbert space and the theorem of 
Fr\'{e}chet-Riesz}
\label{sec:overline_H_and_inear_version_of_Frechet_Riesz}
Fix an arbitrary Hilbert space $(H, +, \cdot)$ over $\C$, with inner product 
$\langle \cdot, \cdot\rangle_H$. Recall from \cite[Chapter 2.6]{KR1983} 
(or \cite[Chapter 0.2]{AS2017}, resp. \cite{K1995}) the - all-important - 
construction of the conjugate Hilbert space $\overline{H}$. 
$\overline{H}$ is defined to be the \textit{same set} $H$, whereby the additive 
Abelian group structure of $H$ is maintained; in contrast to its scalar 
multiplication and its inner product. Instead, $H$ is equipped with the 
alternative scalar multiplication 
\begin{align}\label{eq:scalar_multipl_in_H_bar}
\C \times H \ni (\lambda, x) \mapsto \overline{\lambda}\,x = : \lambda \ast x\,.
\end{align}
Consequently, also
\begin{align}\label{eq:inner_product_on_the_conjugate_HS}
H \times H \ni (x,y) \mapsto 
\langle x, y\rangle_{\overline{H}} : = 
\overline{\langle x, y\rangle_H} = \langle y, x\rangle_H
\end{align}
is a well-defined inner product on $H$ (with respect to the 
scalar multiplication $\ast$, of course), which clearly satisfies 
$\Vert x \Vert_H = \Vert x \Vert_{\overline{H}}$ for all $x \in H$.
In summary, the inner product space $\overline{H}$ is the 
\textit{same set} $H$, equipped with the algebraic structure and the inner 
product, defined by the following mappings:
\begin{align}\label{eq:structure_of_conj_HS}
\begin{split}
H \times H & \longrightarrow H, (x,y) \mapsto x+y\,,\\ 
\C \times H & \longrightarrow H, (\lambda, x) \mapsto \lambda \ast x : = 
\overline{\lambda}\,x\,,\\
H \times H & \longrightarrow \C, (x,y) \mapsto \langle x, y \rangle_{\overline{H}} : = 
\langle y, x \rangle_H\,.
\end{split}
\end{align} 
$\overline{H}$ is a complex Hilbert space, too, of course. Obviously, 
$\{e_i : i \in I\}$ is an orthonormal basis of $H$ if and only if 
$\{e_i : i \in I\}$ is an orthonormal basis of $\overline{H}$. Observe also 
the simple, yet important fact that the ``original'' multiplication 
$\C \times H \ni (\alpha, x) \mapsto \alpha x$ satisfies
\[
\langle \lambda x, \mu y\rangle_{\overline{H}} = 
\langle \overline{\lambda}\ast x, \overline{\mu}\ast y\rangle_{\overline{H}} =
\overline{\lambda}\mu \langle x, y\rangle_{\overline{H}}\,\text{ for all } 
\lambda, \mu \in \C \text{ and } x, y \in \overline{H}.
\]
Note also that by construction the dual space of the Hilbert space $\overline{H}$ 
coincides with the so called ``antidual (Hilbert) space'', introduced in 
(cf. \cite[p. 116]{T1967}); an important fact which should be taken into 
account, when considering \sref{Theorem}{thm:Riesz_linear_version}.    

$\langle \cdot, \cdot\rangle_{\overline{H}}$ actually could be viewed as 
an inner product on the complex Hilbert space $H$ itself, which is linear 
in the second argument and conjugate linear in the first one. That observation 
leads directly to a further representation of inner products $\langle \cdot 
\mid \cdot\rangle_H \not= \langle \cdot, \cdot \rangle_H$, commonly used in 
physics. 

To put it another way: in general, if $(H, \langle \cdot, \cdot\rangle{_H})$ is 
a given complex Hilbert space, many physicists are actually working with the 
inner product on $\overline{H}$, where the structure of the latter is the one 
which is commonly used in mathematics:
\begin{align}\label{eq:bra_ket_vs_H_bar}
H \times H \ni (x,y) \mapsto 
\langle x \!\mid\! y \rangle_H : = \langle x, y\rangle_{\overline{H}} = 
\overline{\langle x, y\rangle_H} = \langle y , x \rangle_H\,.
\end{align}
The construction also implies that the complex conjugate of the Hilbert 
space $(\overline{H}, \langle \cdot, \cdot \rangle_{\overline{H}})$ precisely 
coincides with the Hilbert space $H$ itself. {Namely, since 
$\overline{\overline{H}} = \overline{L}$, where $L : =\overline{H}$, it follows 
that
\begin{align}\label{eq:outer_multiplication_in_H_bar_bar}
\lambda \ast_{\overline{L}} x = \overline{\lambda} \ast_L x = \lambda x
\end{align}
for all $(\lambda, x) \in \C \times H$. Moreover, $\langle \cdot, 
\cdot \rangle_{\overline{L}} = \langle \cdot, \cdot \rangle_H$.} It is also 
obvious that in the case of real Hilbert spaces $H^\R$, the whole structure of 
the conjugate Hilbert space $\overline{H^\R}$ is well-defined as well. However, 
in the real case $\overline{H^\R}$ can be completely ignored (since 
$\lambda \ast x = \lambda x$ for all $(\lambda, x) \in \R \times H^\R$ and 
$\langle \cdot, \cdot \rangle_{\overline{H^\R}} = \langle \cdot, \cdot \rangle_
{H^\R}$).
The definition of $\overline{H}$ clearly implies the trivial, yet crucial fact, 
that
\begin{align}\label{eq:canonical_antiunitary}
Id_H \not= J_H : H \longrightarrow \overline{H}, x \mapsto x
\end{align}
defines a canonical antiunitary operator, whose antiunitary inverse is given by 
$J_H^{-1} = J_{\overline{H}}$. Despite the trivially appearing mapping rule of 
$J_H$, the structure of this antiunitary operator is very important, since the  
utilisation of $J_H$ serves to clarify, whether we have to apply the inner 
product structure of $H$ or the different one of the conjugate Hilbert space 
$\overline{H}$ (cf., e.g, \eqref{eq:conjugation_operator})\,!    

In analogy to the case of the Hilbert space $\C^n$, it is possible to construct 
the inner product of $\overline{H}$ explicitly; even if the Hilbert space over 
$\C$ is infinite-dimensional and nonseparable. Roughly formulated, the inner 
product of elements $x \in \overline{H}$ and $y \in 
\overline{H}$ with respect to the Hilbert space $\overline{H}$ coincides with 
the inner product of the conjugate elements $x^\adj \in H$ and $y^\adj \in H$ 
of the original Hilbert space $H$. So, we have to construct these conjugate 
elements properly.      

To begin with, let $\mathscr{B} \equiv \{e_i : i \in I\}$ and $\mathscr{C} 
\equiv \{f_j : j \in J\}$ be two orthonormal bases of $H$. Fix $x, y \in H$. 
Then $x = \sum\limits_{i\in I}\langle x, e_i \rangle_H\,e_i = 
\sum\limits_{j\in J}\langle x, f_j \rangle_H\,f_j$. A description of the 
mathematics and the technical subtleties underlying such ``summable families'' 
in general Banach spaces (which won't be considered in this article) can be 
found e.g. in \cite{G2022, H2011, Oe2006}. To put it very briefly (and 
nonrigorously): working with these summable families in the Hilbert space case 
is - in general - comparable to working with finite sums. 
This leads us to the following important constructions:
\begin{align}\label{eq:Re_part_and_Im_part_in_H}
\Re_{\mathscr{B}}(x) : = \sum\limits_{i\in I}\Re(\langle x, e_i \rangle_H)\,e_i
\in H\,,\, \Im_{\mathscr{B}}(x): = -\Re_{\mathscr{B}}(i x) = \sum\limits_{i\in I}
\Im(\langle x, e_i \rangle_H)\,e_i \in H\,,
\end{align}
and 
\begin{align}\label{eq:conjugate_complex_element_in_H}
x^\adj_{\mathscr{B}} : = \Re_{\mathscr{B}}(x) - i \Im_{\mathscr{B}}(x) = 
\sum\limits_{i\in I}\overline{\langle x, e_i \rangle_H}\,e_i \in H\,.
\end{align}
Parseval's identity implies that $\Re_{\mathscr{B}}(x)$, $\Im_{\mathscr{B}}(x)$ 
and $x^\adj_{\mathscr{B}}$ are well-defined. However, it should be emphasized 
that these elements all depend on the choice of the orthonormal basis of $H$. 
It is trivial (yet important) that $\Re_{\mathscr{B}}(e_i) = e_i$ and 
$\Im_{\mathscr{B}}(e_i) = 0$ for all $i \in I$. Clearly, the construction also 
implies that $\Re_{\mathscr{B}}(x) + i \Im_{\mathscr{B}}(x) = x = 
\Re_{\mathscr{C}}(x) + i \Im_{\mathscr{C}}(x)$, $\Re_{\mathscr{B}}(x) = 
(x + x^\adj_{\mathscr{B}})/2$ and $\Im_{\mathscr{B}}(x) = (x - 
x^\adj_{\mathscr{B}})/{2i}$. It is also obvious that $(\alpha x)^\adj_
{\mathscr{B}} = \overline{\alpha}\,x^\adj_{\mathscr{B}}$ for all $(\alpha, x) 
\in \C \times H$ and that $(e_i)^\adj_{\mathscr{B}} = e_i$ for all $i \in I$. 
A further trivial, yet important implication of the construction is given by 
\[
\langle \Re_{\mathscr{B}}(x), e_i\rangle_H = \Re(\langle x, e_i\rangle_H)\,
\text{ and }\,\langle \Im_{\mathscr{B}}(x), e_i\rangle_H = 
\Im(\langle x, e_i\rangle_H)\,,
\]
implying that $\Re_{\mathscr{B}}(\Re_{\mathscr{B}}(x)) = \Re_{\mathscr{B}}(x)$, 
$\Re_{\mathscr{B}}(\Im_{\mathscr{B}}(x)) = \Im_{\mathscr{B}}(x)$ and 
$\Im_{\mathscr{B}}(\Im_{\mathscr{B}}(x)) = \Im_{\mathscr{B}}
(\Re_{\mathscr{B}}(x)) = 0$ for all $x \in H$. If $v \in 
\{\Re_{\mathscr{B}}(x), \Im_{\mathscr{B}}(x)\}$ and $w \in 
\{\Re_{\mathscr{B}}(y), \Im_{\mathscr{B}}(y)\}$, where $x, y \in H$, then 
\begin{align}\label{eq:inner_product_of_Re_and_Im_elements}
\langle w, v \rangle_H = \langle v,w \rangle_H \in \R\,. 
\end{align}
Hence, $x^{\adj\adj} \equiv (x^\adj_{\mathscr{B}})^\adj_{\mathscr{B}} = x$ for 
all $x \in H$, and
\begin{align}\label{eq:inner_prod_of_HS_conjugate_complex_elements}
\langle x^\adj_{\mathscr{B}}, y^\adj_{\mathscr{B}} \rangle_H 
\stackrel{\eqref{eq:inner_product_of_Re_and_Im_elements}}{=} 
\langle y, x \rangle_H = \overline{\langle x, y\rangle_H} = 
\langle x, y\rangle_{\overline{H}} \,.
\end{align}
It follows that $\langle x^\adj_{\mathscr{B}}, y^\adj_{\mathscr{B}} \rangle_H$ 
does not depend on the choice of the orthonormal basis $\mathscr{B}$ of $H$ (as 
opposed to $\Re_{\mathscr{B}}(x) \in H$ and $\Im_{\mathscr{B}}(x) \in H$), and
\begin{align}\label{eq:norm_inequalities}
\max\{\Vert\Re_{\mathscr{B}}(x)\Vert_H^2, \Vert\Im_{\mathscr{B}}(x)\Vert_H^2\} 
\leq \Vert x^\adj_{\mathscr{B}}\Vert_H^2 = \Vert x\Vert_H^2 
\stackrel{\eqref{eq:inner_product_of_Re_and_Im_elements}}{=} 
\Vert\Re_{\mathscr{B}}(x)\Vert_H^2 + \Vert\Im_{\mathscr{B}}(x)\Vert_H^2
\end{align}     
for all $i \in I$ and $x \in H$. Consequently, if we put
\[
H_{\mathscr{B}}^\R : = \{x \in H : \Im_{\mathscr{B}}(x) = 0\} 
\stackrel{\eqref{eq:Re_part_and_Im_part_in_H}}{=} 
\{\Re_{\mathscr{B}}(x) : x \in H\} \stackrel{\eqref{eq:conjugate_complex_element_in_H}}{=} 
\{x \in H : x^\adj_{\mathscr{B}} = x\}\,,
\]
it clearly follows that $H_{\mathscr{B}}^\R$ is a vector space over $\R$. Since 
$\Re_{\mathscr{B}}(x) = \Im_{\mathscr{B}}(ix)$ for all $x \in H$, it also 
follows that $H_{\mathscr{B}}^\R = \{\Im_{\mathscr{B}}(y) : y \in H\}$. It is 
obvious that 
\[
H_{\mathscr{B}}^\R \times H_{\mathscr{B}}^\R \ni (v, w) \mapsto 
\langle v, w \rangle_{H_\R} : = \langle v, w \rangle_H  
\]
is a well-defined inner product on $H_{\mathscr{B}}^\R$. Consequently, if we 
take \eqref{eq:norm_inequalities} into account, it follows that 
$(H_{\mathscr{B}}^\R, \langle \cdot, \cdot \rangle_{H_{\mathscr{B}}^\R})$ is a 
Hilbert space over $\R$. \eqref{eq:norm_inequalities} also implies that the 
complex Hilbert space $H$ coincides isometrically with the complexification of 
the real Hilbert space $H_{\mathscr{B}}^\R$\,:
\[
H_{\mathscr{B}}^\R + i H_{\mathscr{B}}^\R = H \cong  (H_{\mathscr{B}}^\R)_\C
\,,\, x \mapsto (\Re_{\mathscr{B}}(x), \Im_{\mathscr{B}}(x))\,.
\]
Finally, it should be noted the quite remarkable fact that the conjugate Hilbert 
space $\overline{H}$ emerges naturally, when we consider the complexification 
of the real Hilbert space direct sum $H_{\mathscr{B}}^{\R, 2} : = 
H_{\mathscr{B}}^\R\,\oplus\,H_{\mathscr{B}}^\R \cong 
H_{\mathscr{B}}^\R \otimes \R^2$ (cf. 
\sref{Proposition}{prop:tensor_product_as_direct_sum}). A straightforward 
calculation, including the construction of the complexification of a real 
Hilbert space namely reveals that for any $a, b, c, d \in H_{\mathscr{B}}^\R$\,,
\[
U : (H_{\mathscr{B}}^{\R, 2})_\C  
\stackrel{\cong}{\longrightarrow} H \oplus \overline{H}, 
((a,b), (c,d)) \mapsto (a+ic, \overline{b-id}) 
\]
is a well-defined (linear!) isometric isomorphism; i.e., a unitary operator. 
Equivalently, the mapping rule of $U$ can be defined as
\[
U\big((\Re_{\mathscr{B}}(x), \Re_{\mathscr{B}}(y)), (\Im_{\mathscr{B}}(x), 
\Im_{\mathscr{B}}(y))\big) : = (x, y^\adj_{\mathscr{B}}) 
\hspace{2em} (x, y \in H)\,.
\]    
\begin{example}\label{eq:conjugate_HS_in_the_separable_case}
Let $H$ be a separable complex Hilbert space. It is well-known that either 
$H \cong l_2(\C)$ (if $H$ is infinite-dimensional), or $H \cong \C_2^n$ (if 
$\textup{dim}(H) = n < \infty)$ (cf., e.g., 
\cite[Korollar V.4.10 and Korollar V.4.13]{W2011}). Then
\[
\overline{\C_2^n} \stackrel{\cong}{\longrightarrow} \C_2^n, 
(z_1, \ldots, z_n)^\top \mapsto 
(\overline{z_1}, \ldots, \overline{z_n})^\top
\]
as well as   
\[
\overline{l_2} \stackrel{\cong}{\longrightarrow} l_2, x = 
(x_i)_{i \in \N} \mapsto (\overline{x_i})_{i \in \N}
\]
are linear(!) isometric isomorphisms that, however depend on the choice of the 
orthonormal basis in $H$. Again, we recognise that in the separable Hilbert 
space case, there is an isometric isomorphism from $H$ onto its dual $H^\prime$, 
that is not canonical, though. We do not know, whether a - linear - isometric 
isomorphism from $H$ onto $H^\prime$ (or equivalently from $H$ onto 
$\overline{H}$) exists, which does \textit{not} depend on the choice of the 
orthonormal basis of $H$.   
\end{example}
If the choice of the Hilbert space $H$ is understood, we put 
$\langle x \!\mid\! y \rangle \equiv \langle x \!\mid\! y \rangle_H$. Let $H$ 
and $K$ be two Hilbert spaces. Let $\mathscr{B} \equiv \{e_i : i \in I\}$ be an 
orthonormal basis of $H$. Consider the linear(!) operator 
$C_H^{\mathscr{B}} \in \mathfrak{L}(H, \overline H)$, defined as (cf. 
\eqref{eq:conjugate_complex_element_in_H}):
\begin{align}\label{eq:conjugation_operator}
C_H^\mathscr{B} : H \longrightarrow \overline{H}, x \mapsto 
J_H x^\adj_{\mathscr{B}} = 
J_H\lb\sum_{i \in I} \overline{\langle x, e_i\rangle_H}\,e_i\rb =
\sum_{i \in I} \langle x, e_i\rangle_H\ast e_i\,.
\end{align}
In particular, $C_H^{\mathscr{B}}(\lambda x) = 
\lambda \ast C_H^{\mathscr{B}} x$ for all $(\lambda, x) \in \C \times H$. It is 
readily verifiable that $C_H^{\mathscr{B}}$ is a (linear) isometric isomorphism 
which satisfies 
$(C_H^{\mathscr{B}})^\adj = C_{\overline{H}}^{\mathscr{B}} \,
{= (C_H^{\mathscr{B}})^{-1}}$. Observe that $C_H^{\mathscr{B}}$ is not a 
canonical mapping. It namely depends on the choice of the orthonormal basis 
$\mathscr{B}$ of $H$. It should be also noted that N. J. Kalton gave a 
concrete example of a Banach space which is not a Hilbert space and which is 
not isomorphic to its complex conjugate (cf. \cite{K1995}). If the choice of 
the orthonormal basis of $H$ is understood, we denote the conjugation mapping 
on $H$ simply as $C_H$. 

Let $\mathscr{B}$ be a given orthonormal basis of $H$. The composition of the 
antiunitary operator $J_{\overline{H}} : \overline{H} 
\stackrel{\cong}{\longrightarrow} H$ and the unitary operator 
$C_H^{\mathscr{B}} : H \stackrel{\cong}{\longrightarrow} \overline{H}$ 
(cf. \eqref{eq:conjugation_operator}) instantly leads to an explicit 
construction of a conjugation which transfers the semilinear complex conjugation 
$\C_2^n \ni z \equiv (z_1, \ldots, z_n)^\top \mapsto 
(\overline{z_1}, \ldots, \overline{z_n})^\top \equiv \overline{z} \in \C_2^n$ 
to the infinite-dimensional complex Hilbert space case:
\begin{proposition}\label{prop:construction_of_a_semilin_conjugation}
Let $H$ be a complex Hilbert space and $\mathscr{B}$ be an orthonormal basis of 
$H$. Consider the semilinear invertible operator
\[
J_H^{\mathscr{B}} : = J_{\overline{H}}\,C_H^{\mathscr{B}} : H \longrightarrow H\,.  
\]
Then $J_H^{\mathscr{B}}$ satisfies the following properties:
\begin{enumerate}
\item $J_H^{\mathscr{B}}$ is an involution; i.e., 
\[
\big(J_H^{\mathscr{B}}\big)^2 = Id_H\,.
\]
\item 
\[
\langle J_H^{\mathscr{B}} x, J_H^{\mathscr{B}} y \rangle_H = 
\overline{\langle x, y \rangle_H} = \langle y, x \rangle_H\,\text{ for all } 
x, y \in H\,.
\]
\end{enumerate}
In particular, if $\mathscr{B}_n$ denotes the standard orthonormal basis in 
$\C_2^n$, then $J_{\C_2^n}^{\mathscr{B}_n} z = \overline{z}$ for all 
$z \in H_n$.
\end{proposition}
\begin{proof}
(i) Let $x \in H$. Recall that $J_{\overline{H}} = J_H^{-1}$. It follows that
\[
J_H^{\mathscr{B}} x = J_{\overline{H}}\,C_H^{\mathscr{B}}x 
\stackrel{\eqref{eq:conjugation_operator}}{=} 
J_{\overline{H}}\lb J_H\,x_{\mathscr{B}}^\adj\rb = x_{\mathscr{B}}^\adj\,,
\]
whence
\[
\lb J_H^{\mathscr{B}}\rb^2 x = J_H^{\mathscr{B}}\,x_{\mathscr{B}}^\adj = 
J_{\overline{H}}\lb C_H^{\mathscr{B}}\,x_{\mathscr{B}}^\adj\rb 
\stackrel{\eqref{eq:conjugation_operator}}{=}
J_{\overline{H}}\lb J_H (x^\adj_{\mathscr{B}})^\adj_{\mathscr{B}}\rb = x\,,
\]
and (i) follows.\\[0.2em] 
(ii) follows instantly from the definition of the operator $J_H^{\mathscr{B}}$ 
and the corresponding properties of its factors. 
\end{proof}  
Equipped with the Hilbert space $\overline{H}$ (being the same set $H$ with the 
same definition of vector addition), the proof of \cite[Theorem V.3.6]{W2011} 
reveals that the classical duality theorem of Fr\'{e}chet-Riesz can be 
instantly ``linearised'' as follows (cf. also \cite[Chapter 2.4, p. 54]{M1990} 
and \cite[Theorem 12.2]{T1967}). Here, we strongly agree with Tr\`{e}ves, 
``\textit{that it is quite incorrect by saying that a Hilbert space is its own 
dual}'' (cf. \cite[Remark prior to Theorem 12.1 on p. 117]{T1967}). Let $H$ be 
a complex Hilbert space, and let
\[
\Delta_H : H \stackrel{\cong}{\longrightarrow} H^\prime, 
x \mapsto \langle \cdot, x \rangle_H
\]
denote the canonical isometric Fr\'{e}chet-Riesz isomorphism, \textit{which is 
semilinear}. If we also take \eqref{eq:canonical_antiunitary} into account, we 
immediately obtain 
\begin{theorem}[\textbf{Linear version of Fr\'{e}chet-Riesz}]
\label{thm:Riesz_linear_version}
Let $H$ be a Hilbert space. Then the mapping
\begin{align}\label{eq:linear_Riesz}
\Phi_H : = \Delta_{\overline{H}}\,J_H : H \stackrel{\cong}{\longrightarrow} 
\overline{H}^{\,\prime}, x \mapsto \langle \cdot, J_H x \rangle_{\overline{H}}  
\end{align}
is a - linear - isometric isomorphism, which does not depend on the choice of 
the orthonormal basis of $H$. In particular,
\[
\Phi_{\overline{H}} = \Delta_H\,\Phi_H^{-1}\,\Delta_{\overline{H}}\,.
\] 
If $(x, y, a) \in H \times H \times H^\prime$, 
then the following conditions are satisfied:
\begin{align}\label{eq:char_prop_of_Phi_overline_H}
\langle J_H y, \Phi_H x\rangle = \langle J_H y, J_H x \rangle_{\overline{H}} = 
\langle x, y\rangle_H = \langle x, \Phi_{\overline{H}}\,y\rangle
\end{align}
and 
\begin{align}\label{eq:char_prop_of_Phi_H}
\langle x, a \rangle = 
\big\langle J_H x, \Phi_{\overline{H}}^{-1}\,a\big\rangle_{\overline{H}}\,. 
\end{align}
The inner product on the Hilbert space $\overline{H}^{\,\prime}$ is given as
\begin{align}\label{eq:inner_product_on_the_dual_of_overline_H}
\langle a, b \rangle_{\overline{H}^{\,\prime}} = 
\langle \Phi_H^{-1}a , \Phi_H^{-1}b \rangle_H \,\text{ for all }\, a,b \in 
\overline{H}^{\,\prime}\,.
\end{align}
Moreover, $\Phi_{\overline{H}}\,C_H^{\mathscr{B}} : H 
\stackrel{\cong}{\longrightarrow} H^\prime, x \mapsto 
\langle \cdot, x^\adj_{\mathscr{B}}\rangle_H$ is a - linear - isometric 
isomorphism, which however depends on the choice of the orthonormal basis 
$\mathscr{B}$ of $H$.   
\end{theorem}
\begin{remark}\label{rem:linear_Riesz_vs_semilinear_Riesz}
It is well-known (and simple to prove) that for any $R \in \id{L}{}{}{}{}
(H, K)$ the dual operator of $R$ can be completely characterised by the adjoint 
of $R$, and conversely:
\vspace{-0.2em}
\begin{center}
\includegraphics[width=6cm]{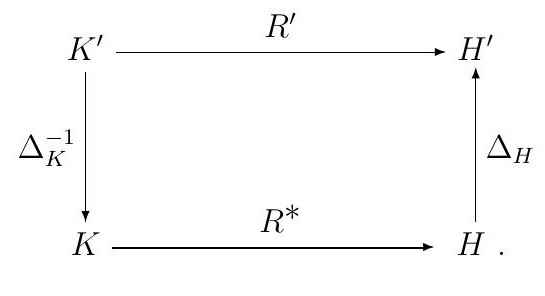} 
\end{center}
\vspace{-1em}
Similarly to \eqref{eq:inner_product_on_the_dual_of_overline_H}, it follows 
that $H^\prime$ is also a Hilbert space, and its inner product is given by 
\begin{align}\label{eq:inner_product_on_the_dual_of_H}
\langle a, b \rangle_{H^\prime} = 
\langle \Delta_H^{-1}b , \Delta_H^{-1}a \rangle_H
\end{align}
for all $a,b \in H^\prime$. Recall also that $S = j_K^{\,-1}\,
S^{\,\prime\prime}\,j_{\overline{H}}$ for all $S \in 
\id{L}{}{}{}{}(\overline{H}, K)$:
\vspace{-0.2em}
\begin{center}
\includegraphics[width=6cm]{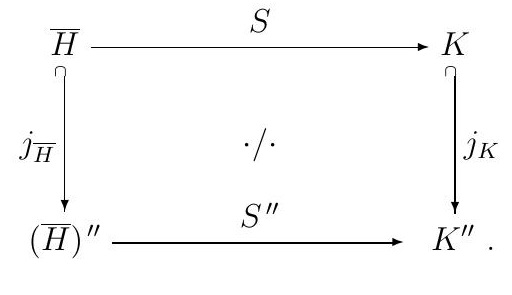} 
\end{center}
\vspace{-1em}
\end{remark}
\begin{remark}\label{rem:Dirac_revisited_I}
\sref{Theorem}{thm:Riesz_linear_version} should be taken into account, when 
P. Dirac's bra-ket formalism is applied. If namely $H$ and $K$ are two Hilbert 
spaces, and if $T \in \mathfrak{L}(H,K)$, then in the sense of Dirac, each $x 
\in H$ is purely symbolically written as ``ket (vector)'': 
$\mid\! x \rangle := x \in H$, implying that $\mid\! Tx \rangle = 
T\mid\! x \rangle$. For any $y \in {H}$, 
\begin{align}\label{eq:def_of_bra}
\langle y\!\mid \, := \langle y \mid\cdot\,\rangle_H = 
\langle \cdot, y\,\rangle_H\in H^\prime\,
\end{align} 
is actually contained in the dual space of $H$, whence $\langle y\!\mid\! S = 
\langle S^\adj y\!\mid$ for all $S \in \id{L}{}{}{}{}(K, H)$. $H^\prime$, 
however coincides \textit{linear-isometrically} with the Hilbert space 
$\overline{H}$ (via $\Phi_{\overline{H}}$)! In the sense of Dirac, 
$\langle y\!\!\mid$ is known as the so-called ``bra (vector)''.\footnote{
The wording ``bra-ket'' originates from the word ``bracket'', where the middle 
letter ``c'' is simply left out.} In other words, Dirac's ``bra'' is an element 
of the Hilbert space $\overline{H}$, whereas a ``ket'' is an element of the 
\textit{different} Hilbert space $H$ (cf. also \cite[Chapter 0.3]{AS2017}).  
\end{remark}
The great strength of Dirac's bra-ket formalism is that it can be used to perform 
complex calculations and algebraic transformations that automatically lead to 
the correct result. Nevertheless, to avoid unavoidable drawbacks of Dirac's 
notation, we should always keep in mind the rigorous mathematics underlying 
Dirac's formal language, including the use of the crucial concept of the 
conjugate Hilbert space. 
\section{Hilbert-Schmidt operators and the tensor product of Hilbert spaces}
\label{sec:HS_ops_and_tp_of_Hilbert_spaces}
This section is primarily designed to provide a short and strongly simplified 
(yet still rigorous) construction of the $n$-fold tensor product 
$\bigotimes\limits_{i=1}^n H_i$ of Hilbert spaces. However, compared to 
the technically quite demanding and very extensive introduction of R. V. Kadison 
and J. R. Ringrose, built on multilinear functional analysis 
(cf. \cite[Chapter 2.6]{KR1983}), our - recursive - method avoids the 
implementation of their so called ``weak multilinear Hilbert-Schmidt mappings''. 
In fact, we will recognise, that it suffices to apply some elementary facts 
from (linear) operator theory, instantly following from the covenient structure 
of the very well-known class of all Hilbert-Schmidt operators and the (less 
familiar) structure of the conjugate Hilbert space. The latter plays also a key 
role in \cite[Chapter 2.6]{KR1983}.    

Also in this regard, the class of all finite-rank linear operators between the 
two  Hilbert spaces $\overline{H}$ and $K$ proves to be very important; 
especially when it comes to their use in the foundations and philosophy of 
quantum mechanics and quantum information theory. In order to make this clear, 
let $H$ and $K$ be arbitrary Hilbert spaces and $(x, z) \in H \times K$. Let 
$T \in \id{L}{}{}{}{}(H, K)$ be an arbitrary rank one operator. Then $T(H) = 
[z_0]$, for some $z_0 \in K\setminus\{0\}$. Thus, if we consider $x_0 : = 
\tfrac{1}{\Vert z_0\Vert^2}\,T^\adj z_0 \in H$, it follows that 
$\langle x, x_0 \rangle_H = \tfrac{1}{\Vert z_0\Vert^2}\,\langle Tx, z_0 \rangle_K$.
Consequently, 
\[
Tx = \langle x, x_0 \rangle_H\,z_0
\]
for all $x \in H$, whence $T = \langle \cdot, x_0 \rangle_H\,z_0 = 
\langle x_0, \cdot \rangle_{\overline{H}}\,z_0$. That observation directly 
leads us to a concrete construction of the tensor product of Hilbert 
spaces, that relies heavily on the conjugate of a  Hilbert space (cf. 
\sref{Theorem}{thm:Hilbert_Schmidt_and_tensor_product}-(iii)). To this end, 
fix  Hilbert spaces $H$ and $K$ and some $(x, y) \in H \times K$. The 
conjugate of $K$ in mind, let us highlight explicitly the mapping rule of the 
(given) scalar multiplication on $K$, described by $\C \times K 
\longrightarrow K, (\lambda, x) \mapsto \lambda \diamond x$, say (such as, e.g., 
$(\lambda, x) \mapsto \lambda\ast x = \overline{\lambda}\,x$). To avoid possible ambiguities 
regarding the choice of the underlying Hilbert spaces (cf. remark below), we 
introduce a notation, enabling a concise description of the following 
elementary building blocks (the so called ``dyads'') of bounded \textit{linear} 
finite rank operators, belonging to $\id{F}{}{}{}{}(\overline{H}, K)$ (but not 
to $\id{F}{}{}{}{}(H, K)$\,(!)):  
\begin{center}
\setlength{\fboxsep}{0.05pt} 
\fbox{\parbox{0.99\columnwidth}{ 
\begin{align}\label{eq:dyad}
\tp{x}{-0.4ex}{\scbox{0.70}{$H$}}{-0.4ex}{\scbox{0.70}{$K$}}{z}
: = 
\langle \cdot, J_H x \rangle_{\overline{H}}\,\diamond z = 
\langle\cdot, \Phi_H x\rangle\,\diamond z\,\in \id{F}{}{}{}{}
(\overline{H}, K)\hspace{1.5em}((x,z) \in H \times K)
\end{align}
}}
\end{center}
\begin{remark}\label{rem:the_operator_Lambda_z}
Actually, the operator \eqref{eq:dyad} turns out to be a special case of an 
important linear operator, which plays a fundamental role in the theory of 
tensor products of Banach spaces. In our view, the use of this operator 
serves to improve the understanding of the construction \eqref{eq:dyad}, since 
it avoids the implementation of additional structure, originating from the 
conjugate of a  Hilbert space. Regarding details, see \cite{Oe2026}. 
So, let $E$ and $F$ be arbitrary (real or complex) Banach spaces. Let 
$E \otimes_{\tiny{\textup{alg}}} F$ be the algebraic tensor product of $E$ and 
$F$ and $w = \sum\limits_{i=1}^{n}x_i \otimes y_i \in 
E \otimes_{\tiny{\textup{alg}}} F$. Recall the construction of the crucial 
canonical metric injection $j_E : E \stackrel{1}{\hookrightarrow} 
E^{\prime\prime}$, whose mapping rule is given by $\langle a, j_E x\rangle : = 
\langle x, a\rangle$ for all $(x, a) \in E \times E^\prime$. The following 
operator is well-defined and independent of the tensor representation of $w$:
\begin{align}\label{eq:Lambda_z}
\Lambda_w : = \sum\limits_{i=1}^{n} \langle\bcdot,\, j_E x_i\rangle\,y_i = 
\sum\limits_{i=1}^{n} \langle x_i,\, \cdot\rangle\,y_i \in 
\id{F}{}{}{}{}(E^\prime, F)\,,  
\end{align}
implying that $\Lambda_w a = \sum\limits_{i=1}^{n} \langle x_i,\, a\rangle
\,y_i$ for all $a \in E^\prime$. The deeper meaning of \eqref{eq:Lambda_z} is 
a result of the following nontrivial representation:  
\begin{align*}
\{\Lambda_w : w \in E \otimes F\} = \{L \in \id{F}{}{}{}{}(E^\prime, F) : 
L^\prime(F^\prime) \subseteq j_E(E)\} = \id{F}{\text{w}^\ast, \text{w}\,}{}{}{}
(E^\prime, F) \varsubsetneq \id{F}{}{}{}{}(E^\prime, F)\,,
\end{align*}
where $\id{F}{\text{w}^\ast, \text{w}\,}{}{}{}(E^\prime, F)$ denotes the set 
of all weak${}^\adj$-weak continuous finite rank operators (cf. \cite{Oe2026}). 
The Hilbert space case \eqref{eq:dyad} (in terms of \eqref{eq:Lambda_z}) can be 
represented as a composed operator. Namely, due to 
\sref{Theorem}{thm:Riesz_linear_version} and the fact that $J_H$ is an 
antiunitary operator that satisfies $J_H J_{\overline{H}} = Id_{\overline{H}}$, 
it follows that 
\[
\tp{x}{-0.4ex}{\scbox{0.70}{$H$}}{-0.4ex}{\scbox{0.70}{$K$}}{z} 
\stackrel{\eqref{eq:dyad}}{=} 
\langle\cdot, J_H x \rangle_{\overline{H}}\,z = 
\langle x, J_{\overline{H}}\,\cdot\rangle_H\,z =
\langle x, \Phi_{\overline{H}}\,\cdot\rangle\,z =
\langle \Phi_{\overline{H}}\,\cdot, j_H x\rangle\,z = 
\Lambda_{x \otimes z}\Phi_{\overline{H}}
\]
for all $(x, z) \in H \otimes K$.
\end{remark}
If the choice of the Hilbert spaces $H$ and $K$ is understood, we suppress the 
Hilbert space symbols and put $x\,\underline{\otimes}\,z \equiv 
\tp{x}{-0.3ex}{\scbox{0.70}{$H$}}{-0.4ex}{\scbox{0.70}{$K$}}{z}$. However, the 
designation ``$\tp{{}}{-0.4ex}{\scbox{0.70}{$H$}}{-0.4ex}{\scbox{0.70}
{$K$}}{{}}$'' is nonnegligible! Namely, if $(x, z) \in H \times K$ is given, we 
have to distinguish carefully the four - linear - operators
\[
\tp{x}{-0.3ex}{\scbox{0.70}{$H$}}{-0.4ex}{\scbox{0.70}{$K$}}{z} \in 
\id{F}{}{}{}{}(\overline{H},K)\,\text{ and }\, 
\tp{J_H x}{-0.3ex}{\scbox{0.70}{$\overline{H}$}}{-0.3ex}{\scbox{0.70}{$K$}}{z}
\in \id{F}{}{}{}{}(H,K),
\]
together with
\[
\tp{x}{-0.4ex}{\scbox{0.70}{$H$}}{-0.5ex}{\scbox{0.70}{$\overline{K}$}}
{J_K z} \in \id{F}{}{}{}{}(\overline{H},\overline{K})\,\text{ and }\, 
\tp{J_H x}{-0.45ex}{\scbox{0.70}{$\overline{H}$}}{-0.45ex}{\scbox{0.70}
{$\overline{K}$}}{J_K z} \in \id{F}{}{}{}{}(H,\overline{K}).
\]
\eqref{eq:dyad} implies that
\begin{align}\label{eq:dyad_conj_HS_cases}
\begin{split}
\tp{x}{-0.4ex}{\scbox{0.70}{$H$}}{-0.4ex}{\scbox{0.70}{$K$}}{z} 
&= \langle\cdot, J_H x \rangle_{\overline{H}}\,z\,,\\ 
\tp{J_H x}{-0.45ex}{\scbox{0.70}{$\overline{H}$}}{-0.4ex}{\scbox{0.70}{$K$}}{z} 
&= \langle\cdot, x\rangle_H z\,,\\
\tp{x}{-0.4ex}{\scbox{0.70}{$H$}}{-0.5ex}{\scbox{0.70}{$\overline{K}$}}{J_K z} 
&= \langle\cdot, J_H x \rangle_{\overline{H}} 
\ast J_K z\,,\\ 
\tp{J_H x}{-0.45ex}{\scbox{0.70}{$\overline{H}$}}{-0.45ex}{\scbox{0.70}
{$\overline{K}$}}{J_K z} 
&= \langle\cdot, x \rangle_H\ast J_K z\,, 
\end{split}
\end{align}
where the explicit highlighting of $\diamond$ in $K$ is suppressed again, of 
course. Obviously, $\Vert \tp{x}{-0.4ex}{\scbox{0.70}{$H$}}{-0.4ex}{\scbox{0.70}
{$K$}}{z}\Vert = \Vert x\Vert_H\,\Vert z\Vert_K$, and $\tp{(\lambda x)}{-0.4ex}
{\scbox{0.70}{$H$}}{-0.4ex}{\scbox{0.70}{$K$}}{(\mu z)} 
\stackrel{\eqref{eq:dyad_conj_HS_cases}}{=} \lambda\mu
(\tp{x}{-0.4ex}{\scbox{0.70}{$H$}}{-0.4ex}{\scbox{0.70}{$K$}}{z})$ for all 
$\lambda, \mu \in \C$. If $x \not= 0$ and $z \not= 0$, then the dyad 
$\tp{x}{-0.4ex}{\scbox{0.70}{$H$}}{-0.4ex}{\scbox{0.70}{$K$}}{z}$ has rank one, 
of course.

Consequently, our symbolic notation \eqref{eq:dyad} removes the 
ambiguity and inconvenience of notation, explicitly highlighted in 
\cite[Chapter 0.3]{AS2017}. \eqref{eq:dyad_conj_HS_cases} namely reflects the 
underlying Hilbert spaces and the implied correct structure of the related dyad 
at once. Here, we should add the following clarifying observation:
\begin{remark}
In Dirac notation, where the important structure of $\overline{H}$ is not 
included (cf. \sref{Remark}{rem:Dirac_revisited_I}), the dyad $\tp{J_H x}
{-0.3ex}{\scbox{0.70}{$\overline{H}$}}{-0.3ex}{\scbox{0.70}{$K$}}{z}$ coincides 
with the ``outer product'' $\mid\!\! z\rangle\langle x\!\!\mid$. This follows 
from
\[
(\mid\! z\rangle\langle x\!\mid)\!\mid\! y\rangle : =
\langle x \mid y\rangle_H\!\mid\! z\rangle =
\langle y, x\rangle_{H}\,z = 
(\tp{J_H x}{-0.3ex}{\scbox{0.70}{$\overline{H}$}}{-0.3ex}
{\scbox{0.70}{$K$}}{z})y  
\]
for all $y \equiv \mid\! y\rangle \in H$. We do not continue to make 
use of Dirac's bra-ket symbol ``$\mid\!\! z\rangle\langle x\!\!\mid : H 
\longrightarrow K$'' in this paper, though. The ''swap'' of the building blocks 
$\langle x\!\mid \, \stackrel{\eqref{eq:def_of_bra}}{=} 
\langle x \mid\cdot\,\rangle_H = \langle\cdot, x\rangle_H \in H^\prime$ and 
$\mid\!\! z\rangle \equiv z \in K$ might namely lead to a minor confusion 
regarding the given underlying domain (in our case, $H$) and the given codomain 
(in our case, $K$). However, in the finite-dimensional complex Hilbert space 
case, the following equality in fact is satisfied for all $m, n \in \N$ and 
$(x, z) \in \C_2^n \times \C_2^m$, where $\F_2^\nu : = 
(\F^\nu, \Vert\cdot\Vert_2)$ denotes the $\nu$-dimensional 
Hilbert space over $\F \in \{\R, \C\}$, equipped with the Euclidean norm 
($\nu \in \N$); tacitly assumed that we identify $\M_{m,n}(\C)$ with 
$\id{L}{}{}{}{}(\C^n, \C^m)$:
\begin{align}\label{eq:rank_one_matrices}
\tp{J_{\C_2^n}x}{-0.4ex}{\scbox{0.70}{$\overline{\C_2^n}$}}{-0.45ex}{\scbox{0.70}
{$\C_2^m$}}{z} =
z x^\ast 
(\equiv\,\mid\!\! z\rangle\langle x\!\!\mid\,) \in \M_{m,n}(\C)\,.
\end{align}
To recognise this, we just have to calculate all image column vectors in $\C^m$: 
\[
(\tp{J_{\C_2^n}x}{-0.4ex}{\scbox{0.70}{$\overline{\C_2^n}$}}{-0.45ex}{\scbox{0.70}
{$\C_2^m$}}{z})e_j \stackrel{\eqref{eq:dyad_conj_HS_cases}}{=} 
\langle e_j, x\rangle_{\C_2^n}\,z = \overline{x_j} z = (z x^\ast)e_j
\]
for all $j \in [n]$. Similarly, it follows from \eqref{eq:dyad} that 
\begin{align}\label{eq:rank_one_matrices_II}
\tp{x}{-0.4ex}{\scbox{0.70}{$\C_2^n$}}{-0.45ex}{\scbox{0.70}{$\C_2^m$}}{z} =
zx^\top (\equiv\, \mid\!\! z\rangle\langle x\!\!\mid\,) \in 
\M_{m,n}(\C)\,.
\end{align}
\end{remark}
\noindent A straightforward calculation shows that 
\begin{align}\label{eq:adjoint_of_dyad}
(\tp{J_H x}{-0.45ex}{\scbox{0.70}{$\overline{H}$}}{-0.5ex}{\scbox{0.70}
{$K$}}{z})^\adj =
\tp{J_K z}{-0.45ex}{\scbox{0.70}{$\overline{K}$}}{-0.48ex}{\scbox{0.70}{$H$}}{x}
\end{align}
In particular,
\begin{align}\label{eq:adjoint_of_dyad_II}
(\tp{x}{-0.3ex}{\scbox{0.70}{$H$}}{-0.4ex}{\scbox{0.70}{$K$}}{z})^\adj =
\tp{J_K z}{-0.3ex}{\scbox{0.70}{$\overline{K}$}}{-0.4ex}{\scbox{0.70}
{$\overline{H}$}}{J_H x}
\end{align}
(since $J_{\overline{H}} J_H x = x$ and $H = \overline{\overline{H}}$\,). 
\noindent \eqref{eq:linear_Riesz} implies the well-known fact, that 
every \textit{bounded} finite-rank operator $L \in \id{F}{}{}{}{}(H,K)$ can be 
written as
\begin{align}\label{eq:rep_of_finite_rank_operators}
L = \sum_{i=1}^n  \langle \cdot, x_i\rangle_{H}\, z_i 
\stackrel{\eqref{eq:dyad_conj_HS_cases}}{=} \sum_{i=1}^n \tp{J_H x_i}
{-0.45ex}{\scbox{0.70}{$\overline{H}$}}{-0.4ex}{\scbox{0.70}{$K$}}{z_i}\,,
\end{align} 
for some $n \in \N$, $(x_1, \ldots, x_n) \in 
H^n$ and $(z_1, \ldots, z_n) \in K^n$ (cf., e.g., 
\cite[Theorem B.10.]{H2011}, \cite[Proposition 15.3.4]{J1981}, 
\cite[Chapter VI.5]{W2011}), whence 
$L^\adj \stackrel{\eqref{eq:adjoint_of_dyad}}{=} \sum_{i=1}^n
\tp{J_K z_i}{-0.4ex}{\scbox{0.70}{$\overline{K}$}}{-0.45ex}{\scbox{0.70}{$H$}}
{x_i}$. Moreover, we also have:
\begin{align}\label{eq:trace_of_a_dyad}
\tup{tr}(\tp{J_H x}{-0.4ex}{\scbox{0.70}{$\overline{H}$}}{-0.45ex}
{\scbox{0.70}{$H$}}{y}) = \tup{tr}(\langle \cdot, x\rangle_H\, y) = 
\langle y , x \rangle_H
\end{align}
for all $x, y \in H$, where $\tup{tr} : \id{F}{}{}{}{}(H,H) \longrightarrow \C$ 
denotes the linear (non-normalised) trace functional (cf., e.g., 
\cite[Chapter 6]{DJT1995} and Pietsch's excellent survey \cite{P2014}). In 
particular, we reobtain the well-known representation of the values of vector 
states $\omega_x$ on an arbitrarily given concrete $\tup{C}^\adj$-algebra 
$\mathcal{A} \subseteq \id{L}{}{}{}{}(H)$: 
\begin{align}\label{eq:vector_state_omega_x}
\langle T, \omega_x\rangle : = \langle Tx, x\rangle_H = 
\tup{tr}(\tp{J_H x}{-0.4ex}{\scbox{0.70}{$\overline{H}$}}{-0.45ex}
{\scbox{0.70}{$H$}}{Tx}) \stackrel{\eqref{eq:dyads_and_composition}}{=} 
\tup{tr}(T(\tp{J_H x}{-0.4ex}{\scbox{0.70}{$\overline{H}$}}{-0.45ex}
{\scbox{0.70}{$H$}}{x}))  
\end{align}
for all $x \in S_H : = \{x \in H : \Vert x\Vert_H = 1\}$ and $T \in \mathcal{A}$. 
Already at this point we should note that the finite rank operator $D_x : = 
\tp{J_H x}{-0.4ex}{\scbox{0.70}{$\overline{H}$}}{-0.45ex}{\scbox{0.70}
{$H$}}{x}$ in particular is nuclear and positive, and $\idn{N}{}{}{}{}(D_x) = 
\Vert x\Vert_H^2 = 1$ (cf. \sref{Section}{sec:op_ideals_and_Grothendieck} and 
\sref{Corollary}{cor:rep_of_nuclear_density_operators} therein). Consequently, 
since $\overline{\overline{H}} = H$, it follows that
\begin{align}\label{eq:algebraic_tp_of_two_Hilbert_spaces}
\underline{\otimes} : H \times K \longrightarrow \id{F}{}{}{}{}(\overline{H}, K), 
(x, z) \mapsto x \,{}_H\underline{\otimes}{}_K\, z
\end{align}
is a bounded \textit{bilinear} operator, with $\Vert \underline{\otimes}\Vert = 
1$. Moreover, $\id{F}{}{}{}{}(\overline{H}, K)$ coincides algebraically with the 
linear hull of its subset $\underline{\otimes}(H \times K)$, and 
\begin{align}\label{eq:trace_of_a_dyad_II}
\textup{tr}(\tp{x}{-0.45ex}{\scbox{0.70}{$H$}}{-0.5ex}{\scbox{0.70}
{$\overline{H}$}}{J_H y}) = \textup{tr}(\tp{x}{-0.8ex}{\scbox{0.70}
{$\overline{\overline{H}}$}}{-0.85ex}{\scbox{0.70}
{$\overline{H}$}}{J_H y}) \stackrel{\eqref{eq:trace_of_a_dyad}}{=}
\langle J_H y, J_H x \rangle_{\overline{H}} = \langle x, y \rangle_H
\end{align}
for all $x, y \in H$. \cite[Criterion 2.3]{DF1993} therefore implies
\begin{proposition}\label{prop:algebraic_tp_of_two_Hilbert_spaces}
Let $H$ and $K$ be  Hilbert spaces. Then 
\[
H \otimes_{\tiny{\textup{alg}}} K : = (\id{F}{}{}{}{}(\overline{H}, K), 
\underline{\otimes})
\]
is an algebraic tensor product of $H$ and $K$ $($via $\underline{\otimes}$$)$.
\end{proposition}
\begin{remark}
Actually, since $\overline{H} \cong H^\prime$ and $H$ is reflexive, 
\sref{Proposition}{prop:algebraic_tp_of_two_Hilbert_spaces} is a 
particular case of a well-known result from the general tensor product theory 
of Banach spaces. If namely, $E$ and $F$ are arbitrarily given normed spaces 
(over $\F$), then
\[
\id{F}{}{}{}{}(E, F) = E^\prime \otimes_{\tiny{\textup{alg}}} F
\]
coincides with the \textit{algebraic} tensor product of the dual Banach space 
$E^\prime$ and the normed space $F$ (cf., e.g., 
\cite[Chapter 2.4, Example (5)]{DF1993}). 
\end{remark}
\noindent Because of \eqref{eq:dyad}, it can be also readily seen that
\begin{align}\label{eq:dyads_and_composition}
T(\tp{J_K z}{-0.4ex}{\scbox{0.70}{$\overline{K}$}}{-0.5ex}{\scbox{0.70}
{$H$}}{x}) = \tp{J_K z}{-0.4ex}{\scbox{0.70}{$\overline{K}$}}{-0.5ex}
{\scbox{0.70}{$L$}}{Tx}\,\text{ and }\,(\tp{J_K z}{-0.4ex}{\scbox{0.70}
{$\overline{K}$}}{-0.5ex}{\scbox{0.70}{$H$}}{x})S = \tp{J_L(S^\adj z)}
{-0.4ex}{\scbox{0.70}{$\overline{L}$}}{-0.5ex}{\scbox{0.70}{$H$}}{x}
\end{align}
for all $T \in \id{L}{}{}{}{}(H, L)$, $S \in \id{L}{}{}{}{}(L, K)$ and 
$(x, z) \in H \times K$. Consequently, if $L = \sum\limits_{i=1}^n \tp{x_i}
{-0.4ex}{\scbox{0.70}{$H$}}{-0.4ex}{\scbox{0.70}{$K$}}{z_i} \in 
\id{F}{}{}{}{}(\overline{H}, K)$ is arbitrarily chosen, then
\[
\id{F}{}{}{}{}(K, \overline{H}) \ni L^\adj = \sum\limits_{j=1}^n 
\tp{J_K z_j}{-0.45ex}{\scbox{0.70}{$\overline{K}$}}{-0.55ex}{\scbox{0.70}
{$\overline{H}$}}{J_H x_j} \stackrel{\eqref{eq:dyad_conj_HS_cases}}{=} 
\sum\limits_{j=1}^n \langle\cdot, z_j\rangle_K \ast J_H x_j = 
\sum\limits_{j=1}^n \overline{\langle \cdot, z_j\rangle}_K\,x_j\,.
\]
A straightforward calculation therefore implies 
\begin{align}\label{eq:finite_rank_ops_are_Hilbert_Schmidt}
\textup{tr}(L^\adj L) \stackrel{\eqref{eq:dyads_and_composition}}{=}
\sum\limits_{i=1}^n\textup{tr}(\tp{x_i}{-0.45ex}{\scbox{0.70}{$H$}}{-0.5ex}
{\scbox{0.70}{$\overline{H}$}}{J_H(L^\adj z_i)}) 
\stackrel{\eqref{eq:trace_of_a_dyad_II}}{=} \sum\limits_{i=1}^n
\langle x_i, L^\adj z_i \rangle_H =
\sum_{i=1}^n\sum\limits_{j=1}^n 
\langle x_i, x_j \rangle_H\,\langle z_i, z_j \rangle_K\,.
\end{align}
Dyads (and hence all bounded finite-rank operators) therefore turn out to be 
the building blocks of the larger set $\id{S}{2}{}{}{}\,(H, K)$ of all 
\textit{Hilbert-Schmidt operators} between $H$ and $K$, also known as operators 
of the \textit{2nd Schatten-von Neumann class}. A comprehensive operator 
theoretic treatment of these (compact) operators - without any use of 
multilinear functional analysis - can be found, for example, in 
\cite{C2000, DJT1995, MV1997, P2014, W2011}. (Regarding the role of 
trace-class operators, which we deliberately ignore at this moment, cf. 
\sref{Theorem}{thm:char_of_trace_class_ops}). Recall that a 
\textit{bounded} linear operator $T \in \id{S}{2}{}{}{}\,(H, K)$ if and only if 
there exists an orthonormal basis $\{e_i : i \in I\}$ of $H$, such that 
$\sum\limits_{i \in I}\Vert Te_i\Vert_K^2 < \infty$; i.e., if and only if the 
net $\big(\sum\limits_{i \in M}\Vert Te_i\Vert_K^2\big)_{M \in \mathcal{F}(I)}$ 
converges in $K$, whose limit will then be denoted as $\sum\limits_{i \in I}
\Vert Te_i\Vert_K^2$, where $\mathcal{F}(I)$ denotes the set of all finite 
subsets of $I$, directed by set inclusion. In fact, we have
\[
\sum\limits_{i \in I}\Vert Te_i\Vert_K^2 = 
\sup\big\{\sum\limits_{i \in M}\Vert Te_i\Vert_K^2 : M \in 
\mathcal{F}(I)\big\}\,.
\]
Purists recognise that for any function $f : I \longrightarrow [0, \infty)$, the 
``sum'' $\sum\limits_{i \in I} f(i)$ precisely coincides with the Lebesgue 
integral of $f$ relative to the counting measure on the index set $I$. 
Regarding an explicit canonical construction of the tensor product of 
(possibly infinite-dimensional and nonseparable) Hilbert spaces, it is 
very important to recall that for any $T \in 
\id{S}{2}{}{}{}\,(H, K)$, 
\[
\sigma_2(T) : = \big(\sum\limits_{i \in I}\Vert Te_i\Vert^2\big)^{1/2} =
(\sum\limits_{i \in I}\langle e_i, T^\adj T e_i\rangle)^{1/2} = 
\sqrt{\textup{tr}(T^\adj T)}
\] 
\textit{does not depend on the choice of the orthonormal basis 
$\{e_i : i \in I\}$ of $H$}. In fact, to see this at once, we just have to 
apply Parseval's identity twice, implying that 
\begin{align}\label{eq:T_HS_iff_T_ast_HS_and_indep_of_ONB}
\sum\limits_{i \in I}\Vert Te_i\Vert_K^2 = \sum\limits_{i \in I}\sum\limits_
{j \in J}\vert\langle Te_i, f_j\rangle\vert^2 = \sum\limits_{j \in J}
\Vert T^\adj f_j\Vert_H^2
\end{align}
for any orthonormal basis $\{f_j : j \in J\}$ of $K$ (cf., e.g., 
\cite[Proposition 18.1]{C2000}). 
It is well-known that $(\id{S}{2}{}{}{}\,(H, K), \langle \cdot, \cdot\rangle_
{\id{S}{2}{}{}{}\,})$ is a Hilbert space, where the inner product is given by
\begin{align}\label{eq:HS_inner_product}
\langle S, T\rangle_{\id{S}{2}{}{}{}\,} \equiv \langle S, T\rangle_
{{\id{S}{2}{}{}{}\,}(H,K)} : = \tup{tr}(T^\adj S) = \sum\limits_{i \in I}
\langle Se_i, Te_i \rangle_K\,.
\end{align}
The Hilbert-Schmidt norm 
\begin{align}\label{eq:HS_norm}
\mathfrak{S}_2(H, K) \ni T \mapsto \sigma_2(T) = 
\sqrt{\langle T, T\rangle_{\mathfrak{S}_2}} = \sqrt{\tup{tr}(T^\adj T)}
\end{align}
is also often denoted as $\idn{S}{2}{}{}{}\,(\cdot)$ (since Hilbert-Schmidt 
operators between  Hilbert spaces actually are components of the 
quasi-normed Banach operator ideal of $2$-approximable operators 
$(\id{S}{2}{}{}{}\,, \idn{S}{2}{}{}{}\,)$ - cf. \cite[Chapter 19.8 and Chapter 
20.2]{J1981}). If $T \in \id{S}{2}{}{}{}\,(H, K)$ and $x = \sum\limits_{i \in I} 
\langle x, e_i\rangle_H\,e_i \in H$, then Parseval's identity and the assumed 
continuity of $T$ imply that
\[
\Vert T x\Vert_K^2 = \Vert\sum\limits_{i\in I}\langle x, e_i\rangle_H\,Te_i\Vert_K^2
\leq \sum\limits_{i\in I}\vert\langle x, e_i\rangle_H\vert^2\,
\sum\limits_{i \in I}\Vert Te_i\Vert^2 = \Vert x\Vert_H^2\,\sigma_2^2\,, 
\]
whence $\Vert T \Vert \leq \sigma_2(T)$. Due to 
\eqref{eq:finite_rank_ops_are_Hilbert_Schmidt}, it follows that every bounded 
finite-rank operator $L \in {\id{F}{}{}{}{}\,}(H,K)$ is Hilbert-Schmidt, and its 
Hilbert-Schmidt norm is given by
\[
\sigma_2(L) = \Big(\sum_{i=1}^n\sum\limits_{j=1}^n 
\langle x_i, x_j \rangle_H\,\langle z_i, z_j \rangle_K\Big)^{1/2} \geq 0\,.
\]
The following reformulation of some known (yet somewhat dispersed) results 
concerning Hilbert-Schmidt operators proves to be a door opener towards an 
explicit construction of the tensor product of Hilbert spaces, 
which however is a canonical one; meaning that its explicit construction does 
not depend on the choice of the orthonormal bases of the underlying Hilbert 
spaces. For the convenience of the reader, we offer a short, yet complete proof.   
\begin{theorem}\label{thm:Hilbert_Schmidt_and_tensor_product}
Let $G$, $H$, $K$ and $L$ be  Hilbert spaces. Then the following statements 
hold:
\begin{enumerate}
\item $R \in \id{S}{2}{}{}{}\,(H, K)$ if and only if $R^\adj \in 
\id{S}{2}{}{}{}\,(K, H)$. In each of these cases, $\sigma_2(R) = 
\sigma_2(R^\adj)$.
\item $R \in \id{S}{2}{}{}{}\,(H, K)$ if and only if $R^\prime \in 
\id{S}{2}{}{}{}\,(K^\prime, H^\prime)$ if and only if $R^{\prime\prime} \in 
\id{S}{2}{}{}{}\,(H^{\prime\prime}, K^{\prime\prime})$. In each of these cases, 
$\sigma_2(R) = \sigma_2(R^\prime) = \sigma_2(R^{\prime\prime})$.
\item Let $T \in \id{S}{2}{}{}{}\,(H, K)$,  $A \in \id{L}{}{}{}{}(G, H)$ and 
$B \in \id{L}{}{}{}{}(K, L)$. Then $BTA \in \id{S}{2}{}{}{}\,(G, L)$, and 
$\sigma_2(BTA) \leq \Vert B \Vert\,\sigma_2(T)\,\Vert A \Vert$.   
\item 
Let $R \in \id{S}{2}{}{}{}\,(H, K)$. Then
\begin{align}\label{eq:key_representation}
\langle R, \tp{J_H x}{-0.45ex}{\scbox{0.70}{$\overline{H}$}}{-0.4ex}
{\scbox{0.70}{$K$}}{z} \rangle_{\id{S}{2}{}{}{}\,} = \langle Rx, z \rangle_K
\end{align}
for all $(x, z) \in H \times K$.
\item
If $x_1, x_2 \in H$ and $z_1, z_2 \in K$, then
\[
\langle \tp{x_1}{-0.45ex}{\scbox{0.70}{$H$}}{-0.5ex}{\scbox{0.70}{$K$}}{z_1}, 
\tp{x_2}{-0.45ex}{\scbox{0.70}{$H$}}{-0.5ex}{\scbox{0.70}{$K$}}{z_2}\rangle_
{\id{S}{2}{}{}{}\,} = \langle x_1, x_2\rangle_{H}\,
\langle z_1, z_2\rangle_K\,.
\]
\item The set of all finite-rank operators $\id{F}{}{}{}{}(H, K)$ is a
$\sigma_2$-dense subset of $\id{S}{2}{}{}{}\,(H, K)$. 
\item If $\{e_i : i \in I\}$ is an orthonormal basis of 
$H$ and $\{f_j : j \in J\}$ is an orthonormal basis of $K$, 
then the family $\{\tp{J_H e_i}{-0.45ex}{\scbox{0.70}{$\overline{H}$}}
{-0.5ex}{\scbox{0.70}{$K$}}{f_j} : (i,j) \in I \times J\} \subseteq 
\id{F}{}{}{}{}(H, K)$ is an orthonormal basis of $\id{S}{2}{}{}{}\,(H, K)$. In 
particular, 
\begin{align}\label{eq:HS_ONB_rep}
\id{S}{2}{}{}{}\,(H, K) = \big\{\sum_{(i,j)\in I \times J} \lambda_{ij}\, 
\tp{J_H e_i}{-0.45ex}{\scbox{0.70}{$\overline{H}$}}{-0.5ex}
{\scbox{0.70}{$K$}}{f_j} : \sum\limits_{(i,j)\in I \times J}
\vert\lambda_{ij}\vert^2 < \infty\big\}\,,
\end{align}
and if $R \in \id{S}{2}{}{}{}\,(H, K)$, then $\lambda_{ij} = 
\langle Re_i, f_j\rangle_K$ for all $(i,j) \in I \times J$, and 
\begin{align}\label{eq:HS_norm}
\sum_{(i,j)\in I \times J}\vert\lambda_{ij}\vert^2 = \sigma_2^2(R)\,.
\end{align}
\end{enumerate}
\end{theorem}
\begin{proof}
(i) follows directly from \eqref{eq:T_HS_iff_T_ast_HS_and_indep_of_ONB} and the 
operator equality $R^{\adj\adj} = R$. 
\\[0.2em]
(ii) follows from (i), the description of $\sigma_2^2(R^\prime)$ as a summable 
family, whose value does not depend on the choice of the orthonormal basis in 
the dual space $K^\prime$ and \sref{Remark}{rem:linear_Riesz_vs_semilinear_Riesz}, 
including \eqref{eq:inner_product_on_the_dual_of_H}.\\[0.2em]
(iii) Let $T \in \id{S}{2}{}{}{}\,(H, K)$,  $A \in \id{L}{}{}{}{}(G, H)$ and 
$B \in \id{L}{}{}{}{}(K, L)$. Obviously, $BT \in \id{S}{2}{}{}{}\,(H, L)$, and 
$\sigma_2(BT) \leq \Vert B \Vert\,\sigma_2(T)$. Due to (i), it follows that 
$(BT)^\adj \in \id{S}{2}{}{}{}\,(L, H)$, and $\sigma_2((BT)^\adj) = 
\sigma_2(BT)$. Hence, $(BTA)^\adj = A^\adj (BT)^\adj \in \id{S}{2}{}{}{}\,
(K, G)$, and $\sigma_2((BTA)^\adj) \leq \Vert A^\adj \Vert\,\sigma_2((BT)^\adj) = 
\sigma_2(BT)\,\Vert A \Vert$. The claim now follows from a further application 
of (i).
\\[0.2em]   
(iv) Fix $(x, z) \in H \times K$. Then
\begin{align*}
\langle R, J_H x \,\underline{\otimes}\, z \rangle_{\id{S}{2}{}{}{}\,} 
&= \tup{tr}((J_H x \,\underline{\otimes}\, z)^\adj\,R)
\stackrel{\eqref{eq:adjoint_of_dyad}}{=} \tup{tr}
((J_K z \,\underline{\otimes}\, x)\,R)\\ 
&\!\!\!\stackrel{\eqref{eq:dyads_and_composition}}{=} \tup{tr}
((J_H(R^\adj z) \,\underline{\otimes}\, x) 
\stackrel{\eqref{eq:trace_of_a_dyad}}{=} \langle x, R^\adj z \rangle_H = 
\langle Rx, z \rangle_K\,.
\end{align*} 
(v) is a special case of \eqref{eq:finite_rank_ops_are_Hilbert_Schmidt}.
\\[0.2em]
(vi) Consider the closed subspace $\{0\} \not= \id{A}{}{}{}{}(H,K) : = 
\overline{\id{F}{}{}{}{}(H, K)}^{\sigma_2}$ of the Hilbert space 
$\left({\id{S}{2}{}{}{}\,}(H,K)\right.$, $\left.\langle \cdot, \cdot\rangle_
{{\id{S}{2}{}{}{}}}\right)$. Let $T \in \id{A}{}{}{}{}(H,K)^\bot \subseteq 
\id{F}{}{}{}{}(H, K)^\bot$. Then $\langle Tx, z\rangle = 0$ for all $(x, z) \in 
H \times K$ (due to \eqref{eq:rep_of_finite_rank_operators} and 
\eqref{eq:key_representation}), whence $T = 0$. Consequently, 
${\id{S}{2}{}{}{}\,}(H,K) = \id{A}{}{}{}{}(H,K)$.\\[0.2em]
(vii) We just have to apply standard Hilbert space theory to the Hilbert 
space $\left({\id{S}{2}{}{}{}\,}(H,K)\right.$, $\left.\langle \cdot, \cdot\rangle_
{{\id{S}{2}{}{}{}}}\right)$. So, let $S \in {\id{S}{2}{}{}{}\,}(H,K)$ such that 
$\langle S, J_H e_i \,\underline{\otimes}\, f_j \rangle_
{\id{S}{2}{}{}{}\,} = 0$ for all $(i,j) \in I \times J$. 
\eqref{eq:key_representation} implies that $\langle Se_i, f_j \rangle_K = 0$ 
for all $(i,j) \in I \times J$, whence $S = 0$. 
It follows that $\{J_H e_i \,\underline{\otimes}\, f_j : (i, j) \in I \times J\}$ 
is an orthonormal basis of $S \in {\id{S}{2}{}{}{}\,}(H,K)$. Parseval's 
identity (together with the double-sum rearrangement rule) concludes the proof 
of (vii) (cf., e.g., \cite[Satz V.4.9]{W2011}).
\end{proof}
Observe that if $\{e_i : i \in I\}$ is an orthonormal basis of $H$ and 
$\{f_j : j \in J\}$ is an orthonormal basis of $K$, then
\begin{align*}
\id{S}{2}{}{}{}\,(H, K) \stackrel{\eqref{eq:HS_ONB_rep}}{=} \big\{\sum_{i\in I}
\tp{J_H e_i}{-0.45ex}{\scbox{0.70}{$\overline{H}$}}{-0.5ex}{\scbox{0.70}
{$K$}}{y_i} : y_i \in K\,\,\forall i \in I\big\} 
\stackrel{\eqref{eq:HS_ONB_rep}}{=} \big\{\sum_{j\in J} \tp{J_H x_j}
{-0.45ex}{\scbox{0.70}{$\overline{H}$}}{-0.5ex}{\scbox{0.70}{$K$}}{f_j} : 
x_j \in H\,\,\forall j \in J\big\}   
\end{align*}
\begin{remark}\label{rem:HS_ops_on_H_are_not_C_star}
[\textbf{$(\id{S}{2}{}{}{}\,(H), \sigma_2)$ is not a $\tup{C}^\adj$-algebra}]
Let $H$ be a Hilbert space over $\C$. 
\sref{Theorem}{thm:Hilbert_Schmidt_and_tensor_product} clearly implies that the 
normed $\adj$-algebra $\id{S}{2}{}{}{}\,(H) : = \id{S}{2}{}{}{}\,(H, H)$, 
together with norm $\sigma_2$ is a Banach $\adj$-algebra (cf., e.g., 
\cite[Chapter 2.1]{M1990})). The astute reader might raise the question, whether 
$\id{S}{2}{}{}{}\,(H)$, together with $\sigma_2$ even is a $\tup{C}^\adj$-algebra? 
This is not the case (unlike $(\id{L}{}{}{}{}(H), \Vert\cdot\Vert)$). In 
order to recognise this fact, fix an arbitrary orthonormal basis 
$\{e_i : i \in I\}$ of $H$ and consider the finite-rank operator $P : = 
\sum\limits_{i=1}^2 \tp{J_H e_i}{-0.45ex}{\scbox{0.70}
{$\overline{H}$}}{-0.48ex}{\scbox{0.70}{$H$}}{e_i} \in \id{S}{2}{}{}{}\,(H)$. 
\eqref{eq:adjoint_of_dyad}, together with \eqref{eq:dyads_and_composition} 
implies that $P^2 = P = P^\adj$ is a projection. Thus, 
$(\sigma_2(P))^2 = \tup{tr}(P^\adj P) = \tup{tr}(P) 
\stackrel{\eqref{eq:trace_of_a_dyad}}{=} 2$. Consequently, it follows that
\[
\sigma_2(P^\adj P) = \sigma_2(P) = \sqrt{2} < 2 = (\sigma_2(P))^2\,.
\] 
\end{remark}
\noindent Actually, \sref{Proposition}{prop:algebraic_tp_of_two_Hilbert_spaces}, 
together with \sref{Theorem}{thm:Hilbert_Schmidt_and_tensor_product} completely 
reflects the explicit structure of any given tensor product of 
Hilbert spaces. Recall that a Hilbert space tensor product of $H$ and $K$ is a 
 \textit{Hilbert space} $H \scbox{0.80}{$\boxtimes$} K$, together with a 
bilinear mapping $B_\boxtimes : H \times K \longrightarrow H \boxtimes K, (x,z) 
\mapsto B_\boxtimes(x,z)$, such that the linear hull of the image $B_\boxtimes
(H \times K)$ is $\Vert \cdot\Vert_{H \scbox{0.80}{$\boxtimes$} K}$-dense in 
$H \scbox{0.80}{$\boxtimes$} K$ and $\langle x_1 \,\otimes\, z_1, 
x_2 \,\otimes\,z_2 \rangle_{H \scbox{0.80}{$\boxtimes$} K} = 
\langle x_1, x_2\rangle_H\,\langle z_1, z_2\rangle_K$ for all $x_1, x_2 \in H$ 
and $z_1, z_2 \in K$ (cf., e.g., \cite{B2013}). 
\sref{Theorem}{thm:Hilbert_Schmidt_and_tensor_product} therefore implies that 
the completion of the  inner product space
\[
H \otimes_2 K : = (H \otimes_{\tiny{\textup{alg}}} K, 
\langle \cdot, \cdot\rangle_{\id{S}{2}{}{}{}\,})
\] 
with respect to the Hilbert-Schmidt norm $\sigma_2$ in fact is a Hilbert space 
tensor product of the  Hilbert spaces $H$ and $K$, which coincides 
precisely with the Hilbert space $(\id{S}{2}{}{}{}\,(\overline{H}, K), 
\langle \cdot, \cdot\rangle_{\id{S}{2}{}{}{}\,})$. Following the usual notation 
for Banach space tensor products of Banach spaces, we refer to this space as
\[
H \widetilde{\otimes}_2 K : = \overline{H \otimes_2 K}^{\sigma_2} = 
\id{S}{2}{}{}{}\,(\overline{H}, K)\,.
\]  
In particular, it follows that
\[
\id{S}{2}{}{}{}\,(H, K) = \overline{H} \widetilde{\otimes}_2 K\,.
\]
Observe that $H \widetilde{\otimes}_2 K$ does not just denote the algebraic 
tensor product $H \otimes_{\tiny{\textup{alg}}} K$. In general, the latter vector 
space is not complete, much less a Hilbert space. In particular, we reobtain 
\cite[Proposition 2.6.9]{KR1983}, stating that \textit{any} (Hilbert space) 
tensor product of  Hilbert spaces $H$ and $K$ can be isometrically 
identified with the class of all Hilbert-Schmidt operators from the 
\textit{conjugate Hilbert space $\overline{H} \cong H^\prime$} into $K$ - 
independent of the choice of the orthonormal bases of the underlying  
Hilbert spaces $H$ and $K$ (cf. also \cite[Exercise 18.5]{C2000}):   
\begin{corollary}\label{cor:char_of_tensor_product_of_HS_as_Hilbert_Schmidt}
Let $H$ and $K$ be two  Hilbert spaces and $H \otimes K$ be a 
given Hilbert space tensor product of the Hilbert spaces $H$ and $K$, with 
norm $t \mapsto \sigma_{H, K}(t) : = \sqrt{\langle t, t\rangle_{H \otimes K}}$. 
Then
\[
U_\otimes : H \otimes K \stackrel{\cong}{\longrightarrow} 
H \widetilde{\otimes}_2 K = \id{S}{2}{}{}{}\,(\overline{H}, K)
\] 
is a canonical, uniquely determined isometric isomorphism, satisfying
\[
U_\otimes(x \otimes y) = x \,{}_H\underline{\otimes}{}_K\, y \,
\text{ for all }\,(x, y) \in H \times K\,. 
\]
\end{corollary}
\begin{remark}\label{rem:L_H_acting_on_L_S2H}
An instant application of 
\sref{Corollary}{cor:char_of_tensor_product_of_HS_as_Hilbert_Schmidt} gives us 
a concrete representation of the faithful (and hence isometric) 
$\adj$-representation of the $\textup{C}^\adj$-algebra $\id{L}{}{}{}{}(H)$ on 
the Hilbert space $\id{S}{2}{}{}{}\,(H)$. It can be namely readily verified that 
the latter is given as 
\[
\pi_{\textup{L}} : \id{L}{}{}{}{}(H) \stackrel{1}{\hookrightarrow} 
\id{L}{}{}{}{}(\overline{H}\otimes H) \cong \id{L}{}{}{}{}
(\id{S}{2}{}{}{}\,(H)), T \mapsto Id_{\overline{H}}\otimes T \mapsto U_\otimes 
(Id_{\overline{H}}\otimes T) U_\otimes^\adj\,.
\]
Observe that by construction 
\begin{align}\label{eq:pi_L_acting_from_the_left}
\pi_L(T)(S) = TS\,\text{ for all }\,(T, S) \in \id{L}{}{}{}{}(H) \times 
\id{S}{2}{}{}{}\,(H)
\end{align}
(``composition from the left'').
\end{remark}      
A further fruitful application of 
\sref{Corollary}{cor:char_of_tensor_product_of_HS_as_Hilbert_Schmidt} arises, 
when we look for relations between the Hilbert spaces $L^2(\mu_1)$, $L^2(\mu_2)$ 
and $L^2(\mu_1 \otimes \mu_2)$, where each $\mu_i$ is a $\sigma$-finite measure 
on a measurable space $(\Omega_i, {\mathscr{F}}_i)$ ($i \in \{1,2\}$), and 
$\mu_1 \otimes \mu_2$ denotes the (thus well-defined and unique) $\sigma$-finite 
product measure on the measurable space $(\Omega_1 \times \Omega_2, 
{\mathscr{F}}_1 \otimes {\mathscr{F}}_2)$. Independent of any further 
assumptions (such as separability, resp. finite-dimensionality of the 
underlying Hilbert spaces), the following statement always holds:
\begin{proposition}\label{prop:reprs_of_L2_mu1_otimes_L2_mu_2}
Let $(\Omega_1, {\mathscr{F}}_1, \mu_1)$ and $(\Omega_2, {\mathscr{F}}_2, 
\mu_2)$ be two $\sigma$-finite measure spaces. Then the Hilbert spaces
\[
L^2(\mu_1 \otimes \mu_2) \cong \id{S}{2}{}{}{}\,
(\overline{L^2(\mu_1)}, L^2(\mu_2)) \cong L^2(\mu_1) \otimes 
L^2(\mu_2)
\]
are canonically isometric isomorphic. 
\end{proposition}
\begin{proof}
Due to \sref{Corollary}{cor:char_of_tensor_product_of_HS_as_Hilbert_Schmidt}, we 
have to look for the construction of the first isometric isomorphism only. 
It's rather surprising that such an operator can be found quite easily. To this 
end, let $h \in L^2(\mu_1 \otimes \mu_2)$ be given. Put $H_1 : = L^2(\mu_1)$, 
$H_2 : = L^2(\mu_2)$ and $H : = L^2(\mu_1 \otimes \mu_2)$. Consider 
the well-defined linear(!) operator $T_h \in \id{L}{}{}{}{}
(\overline{H_1}, H_2)$, induced by
\[
\overline{H_1} \ni J_{H_1}f \mapsto T_h J_{H_1}f(\omega_2) : = 
\overline{\langle J_{H_1} h(\cdot, \omega_2), J_{H_1}f\rangle_{H_1}} = 
\langle h(\cdot, \omega_2), f\rangle_{H_1} = \int\limits_{\Omega_1} 
h(\omega_1, \omega_2)\overline{f(\omega_1)}\mu_1(\tup{d}\omega_1)\,.
\]
It namely follows from Fubini's theorem that 
\begin{align*}
\Vert T_h J_{H_1}f\Vert_{H_2}^2 & = & \int\limits_{\Omega_2}
\big\vert\langle h(\cdot, \omega_2), f\rangle_{H_1}\big\vert^2\,
\mu_2(\tup{d}\omega_2) \leq \lb\int\limits_{\Omega_2}
\Vert h(\cdot, \omega_2)\Vert^2_{H_1}\,\mu_2(\tup{d}\omega_2)\rb
\Vert J_{H_1}f\Vert_{H_1}^2\\ 
& = & \lb\int\limits_{\Omega_2}\big(\int\limits_{\Omega_1}
\vert h(\omega_1, \omega_2)\vert^2\,\mu_1(\tup{d}\omega_1)\big)
\mu_2(\tup{d}\omega_2)\rb \Vert J_{H_1}f\Vert_
{H_1}^2 = \Vert h\Vert_H^2\,\Vert J_{H_1}f\Vert_{H_1}^2\,.
\end{align*}
In fact, $T_h$ is even a Hilbert-Schmidt operator, and $H 
\stackrel{1}{\hookrightarrow} \id{S}{2}{}{}{}\,(\overline{H_1}, H_2), h \mapsto 
T_h$ is an isometry. In order to recognise this, fix an arbitrary orthonormal 
basis $\{e_i : i \in I\}$ of $H_1$. Because of Parseval's identity, it follows 
(similarly as above) that
\[
\sum\limits_{i \in I} \Vert T_h J_{H_1}e_i\Vert_{H_2}^2 = 
\int\limits_{\Omega_2}\lb
\sum\limits_{i \in I}\big\vert\langle h(\cdot, \omega_2), e_i\rangle_{H_1}
\big\vert^2\rb\mu_2(\tup{d}\omega_2) = 
\int\limits_{\Omega_2}\Vert h(\cdot, \omega_2)\Vert^2_{H_1}\,\mu_2(\tup{d}\omega_2) = 
\Vert h\Vert_H^2\,.
\]
It remains to prove that the isometry $H \ni h \mapsto T_h \in \id{S}{2}{}{}{}\,
(\overline{H_1}, H_2)$ is onto. To this end, note first that for any $(f_1, f_2) 
\in H_1 \times H_2$, the product function $(\omega_1, \omega_2) \mapsto 
f_1(\omega_1)f_2(\omega_2) = : f_1 \boxdot f_2(\omega_1, \omega_2)$ leads to 
a well-definend element $f_1 \boxdot f_2 \in H$, which clearly satisfies the
operator equality  
\[
T_{f_1 \boxdot f_2} = \tp{f_1}{-0.45ex}{\scbox{0.70}{$H_1$}}{-0.5ex}
{\scbox{0.70}{$H_2$}}{f_2}\,.
\] 
\sref{Theorem}{thm:Hilbert_Schmidt_and_tensor_product}-(vii) and a standard 
Cauchy net-argument therefore concludes the proof (since $H \ni h \mapsto T_h 
\in \id{S}{2}{}{}{}\,(\overline{H_1}, H_2)$ is an isometry).    
\end{proof}             
Since every Hilbert-Schmidt operator $T \in \id{S}{2}{}{}{}\,
(\overline{H},K)$ in particular is compact (due to 
\sref{Theorem}{thm:Hilbert_Schmidt_and_tensor_product}-(vi)), 
a straightforward application 
\sref{Corollary}{cor:char_of_tensor_product_of_HS_as_Hilbert_Schmidt}, together 
with the well-known spectral representation theorem for compact operators (cf., 
e.g., \cite[Satz VI.3.6]{W2011}, resp. \cite[Chapter 20.1]{J1981}) and 
\sref{Theorem}{thm:char_of_trace_class_ops} directly imply the following 
generalised version of the so called ``Schmidt decomposition'' of elements of 
the tensor product of two  Hilbert spaces, which both could be 
infinite-dimensional and even nonseparable:
\begin{corollary}[\textbf{Schmidt decomposition: the general case}]
\label{cor:Schmidt_decomposition_revisited}
Let $H$ and $K$ be  Hilbert spaces. Let $z \in H \otimes K$. Then there 
exist countable $($possibly finite$)$ orthonormal sequences $(u_n)_{n \in \A}$, 
$(v_n)_{n \in \A}, (u_n, v_n) \in H \times K$, and a sequence $(s_n)_{n \in \A} 
\in c_0$, consisting of real nonnegative scalars, such that $s_1 \geq 
s_2 \geq s_3 \geq \ldots \geq 0$, 
\[
\sum\limits_{n \in \A} s_n = \idn{N}{}{}{}{}((U_\otimes z)^\adj U_\otimes z) = 
\Vert z\Vert_{H \otimes K}^2 < \infty\,,
\] 
and
\[
z = \sum\limits_{n \in \A} \sqrt{s_n}\,u_n \otimes v_n\,.
\]
Moreover, $(s_n)_{n \in \A} \in c_0$ consists of eigenvalues of the positive 
operator 
\[
\sum\limits_{n \in \A} s_n\,\tp{u_n}{-0.44ex}{\scbox{0.70}{$H$}}{-0.48ex}
{\scbox{0.70}{$\overline{H}$}}{J_H u_n} = (U_\otimes z)^\adj 
U_\otimes z \in
\id{N}{}{}{}{}(\overline{H},\overline{H})\,,
\] 
each of them repeated a finite number of times, corresponding to the dimension 
of their respective eigenspace. In particular, if $z \in S_{H \otimes K}$ is a 
unit vector in $H \otimes K$, then $\sum\limits_{n \in \A} s_n = 1$.
\end{corollary}
\sref{Theorem}{thm:Riesz_linear_version}, together with 
\sref{Corollary}{cor:char_of_tensor_product_of_HS_as_Hilbert_Schmidt} provides 
a deeper understanding of the strongly flexible, rich structure of the tensor 
product of two  Hilbert spaces. We namely have (cf. also 
\eqref{eq:HS_equals_P2} and the powerful general representations 
\cite[Theorem 17.5 and Corollary 22.1]{DF1993}):    
\begin{theorem}\label{thm:structure_of_Hilb_space_tp}
Let $H_1, H_2$ and $H, K$ be Hilbert spaces. If $H_1 \cong H_2$, then 
$H_1 \otimes K \cong H_2 \otimes K$ and $H \otimes H_1 \cong H \otimes H_2$. 
Moreover, also the following equalities hold isometrically, linearily and 
canonically: 
\begin{enumerate}
\item
\[
H^\prime \otimes K \cong \id{S}{2}{}{}{}\,(H, K)\,.
\]
\item
\[
H \otimes K \cong K \otimes H \cong \overline{H} \otimes \overline{K} \cong 
H^\prime \otimes K^\prime\,.
\]
\item 
\[
\id{S}{2}{}{}{}\,(H, \overline{K}) \cong \overline{H} \otimes \overline{K} \cong
\overline{K \otimes H} \cong (K \otimes H)^\prime \cong (H \otimes K)^\prime\,.
\]
\end{enumerate}
\end{theorem}
\begin{proof}
Due to 
\sref{Corollary}{cor:char_of_tensor_product_of_HS_as_Hilbert_Schmidt}, a 
convenient operator theoretic proof of all statements can be provided.\\[0.2em]
(i) instantly follows from \sref{Theorem}{thm:Riesz_linear_version} and
\sref{Corollary}{cor:char_of_tensor_product_of_HS_as_Hilbert_Schmidt}.\\[0.2em]
\noindent
(ii) Because of \sref{Remark}{rem:linear_Riesz_vs_semilinear_Riesz} and 
\sref{Theorem}{thm:Hilbert_Schmidt_and_tensor_product}, a remaining simple 
calculation implies that 
\[
\id{S}{2}{}{}{}\,(\overline{H}, K) \stackrel{\cong}{\longrightarrow} 
\id{S}{2}{}{}{}\,(\overline{K}, H), R \mapsto \Phi_H^{-1}\, R^\prime\,
\Phi_{\overline{K}}
\]  
is a well-defined, canonical isometric isomorphism, whose inverse is given by
$\id{S}{2}{}{}{}\,(\overline{K}, H) \ni S \mapsto \Phi_K^{-1} S^\prime 
\Phi_{\overline{H}}$. Thus,
\[
H \otimes K \cong \id{S}{2}{}{}{}\,(\overline{H}, K) \cong 
\id{S}{2}{}{}{}\,(\overline{K}, H) \cong K \otimes H \cong 
\id{S}{2}{}{}{}\,(\overline{K}, H) \cong \id{S}{2}{}{}{}\,(H, \overline{K}) 
\cong \overline{H} \otimes \overline{K}\,.
\]
The penultimate isometric equality follows from 
\sref{Theorem}{thm:Hilbert_Schmidt_and_tensor_product}-(i), of course.\\[0.2em]
(iii) Clearly, only the second isometric equality requires an 
explicit specification of the mapping rule. Observe that it does not follow 
from (ii) (due to the semilinearity issue)! Because of the crucial trace 
duality theorem for Hilbert-Schmidt operators namely (cf. 
\cite[Theorem 20.2.6]{J1981}), it follows that 
$\id{S}{2}{}{}{}\,(H, \overline{K}) \stackrel{\cong}{\longrightarrow} 
(\id{S}{2}{}{}{}\,(\overline{K}, H))^\prime, T \mapsto \tup{tr}(\bcdot\,T)$ is 
a well-defined isometric isomorphism, which does not depend on the choice of the 
underlying orthonormal bases of $H$ and $K$, respectively. 
\sref{Theorem}{thm:Riesz_linear_version} therefore implies 
that 
\[
\overline{H} \otimes \overline{K} \cong \id{S}{2}{}{}{}\,(H, \overline{K}) \cong 
(\id{S}{2}{}{}{}\,(\overline{K}, H))^\prime \cong (K \otimes H)^\prime \cong 
\overline{K \otimes H}\,.
\]
\end{proof}
\begin{example}[\textbf{No-cloning theorem}]\label{ex:no_cloning_theorem}
Fix a Hilbert space $H$ over $\C$, and let $x, y \in S_H : = \{x \in H : 
\Vert x \Vert = 1\}$, such that $x$ is not orthogonal to $y$ and $x$ and $y$ 
are not linearly dependent. Let $e \in S_H\setminus\{x,y\}$. Is it possible to 
construct a unitary operator (i.e., an onto isometry) which ``copies'' the 
state vector $a \in S_H$ by means of the entangled quantum system $a \otimes e$ 
to the entangled quantum system $a \otimes a$, 
for all $a \in \{x,y\}$? The answer is negative. It is impossible ``to create 
an identical copy of an arbitrary quantum state''. More precisely, the 
following well-known result holds:
\begin{noc*}
Let $H$ be a Hilbert space over $\C$ and $x, y \in S_H$ be two unit vectors in 
$H$. Fix some $e \in S_H\setminus\{x,y\}$. If $x$ is not orthogonal to $y$ and 
$x$ and $y$ are not linearly dependent, then there is no unitary operator 
$U : H {\,\otimes\,} H \longrightarrow H {\,\otimes\,} H$, no 
$\zeta_{x,e} \in \T$ and no $\zeta_{y,e} \in \T$ (where $\T : = 
\{z \in \C : \vert z \vert = 1\}$), such that $U(x \otimes e) = 
\zeta_{x,e}\,x\otimes x$ and $U(y \otimes e) = \zeta_{y,e}\,y\otimes y$. 
\end{noc*} 
\begin{proof}
Assume by contradiction that such two pairs $(U, \zeta_{x,e})$ and 
$(U, \zeta_{y,e})$ exist. Then
\begin{align*}
\vert\langle x\mid y\rangle\vert^2 &= \vert\langle x\mid y\rangle^2\vert = 
\vert\langle x \otimes x \mid y \otimes y \rangle\vert = 
\vert \langle U(x \otimes e) \mid U(y \otimes e)\rangle\vert\\
&= \vert \langle x \otimes e \mid y \otimes e\rangle\vert = 
\vert\langle x\mid y\rangle \vert\,\vert\langle e\mid e\rangle\vert = 
\vert\langle x\mid y\rangle\vert.
\end{align*}
Thus, $\vert\langle x\mid y\rangle\vert \in \{0,1\}$, implying that either
$x$ is orthogonal to $y$ or $\vert\langle x\mid y\rangle\vert = 
\Vert x \Vert\,\Vert y \Vert$ (since $\Vert x \Vert = 1 = \Vert y \Vert$ by 
assumption). However, in the latter case, the Cauchy-Schwarz inequality even 
would turn into an equality, implying that in this case $x$ and $y$ would be 
linearly dependent, which contradicts the assumption.
\end{proof}
\end{example}
A nontrivial dependence on the choice of the orthonormal basis of a 
given Hilbert space is also reflected in the following known result: 
\begin{proposition}\label{prop:tensor_product_as_direct_sum}
Let $H$ and $K$ be two Hilbert spaces and $\{f_j : j \in J\}$ an 
orthonormal basis of $K$. Then 
\[
H \,{\otimes}\, K \cong \bigoplus_{j \in J}H \equiv H^J
\]
are isometrically isomorphic. The onto isometry $\psi : H \,{\otimes}\, K 
\longrightarrow H^J$ satisfies
\[
\psi(x \,{\otimes}\, z) : = \big(\langle z , f_j\rangle_K\,x\big)_{j \in J} = 
\big((J_K z \,{}_{\overline{\scriptstyle{K}}}\underline{\otimes}_
{\scriptstyle{H}}\, x)f_j \big)_{j \in J}\,\text{ for all } (x,z) \in H \times K.
\]
Its inverse is given by
\[
\psi^{-1}(w) = \sum_{j \in J} w_j \,{\otimes}\, f_j\,\text{ for all } w = 
(w_j)_{j \in J} \in H^J.
\]
In particular, if $z \in H \,{\otimes}\, K$, then there exists a uniquely determined 
$\xi \equiv (\xi_j)_{j \in J} \in H^J$, such that $z = \sum_{j \in J} 
\xi_j \otimes f_j$ and $\Vert \xi \Vert_{H^J} = \Vert z \Vert_{H \,{\otimes}\, K}$. 
Similarly, if $\{e_i : i \in I\}$ an orthonormal basis of $H$, then
\[
H \,{\otimes}\, K \cong \bigoplus_{i \in I}K \equiv K^I.
\]
\end{proposition}
\begin{proof}
\cite[Remark 2.3]{DF1993}, together with \eqref{eq:dyad} implies that $\psi$ is 
a well-defined (algebraic) isomorphism. So, we only have to show that $\psi$ is 
an isometry. To this end, let $z \in H \,{\otimes}\, K$. Let $\{e_i : i \in I\}$ be 
an orthonormal basis of $H$. Then $z = \sum\limits_{(i,j) \in I \times J}\lambda_{ij}\,
e_i \,{\otimes}\, f_j$, where the  numbers $\lambda_{ij}$ satisfy 
$\sum\limits_{(i,j) \in I \times J}\vert \lambda_{ij}\vert^2 = 
\Vert z \Vert_{H \,{\otimes}\, K}^2 < \infty$ (due to 
\eqref{eq:HS_norm}).
Hence, $\psi(z) = \big(\sum_{i \in I}\lambda_{ij}\,e_i\big)_{j \in J}$, and it 
follows that
\[
\Vert \psi(z) \Vert_{{H^J}}^2 = 
\sum_{j \in J}\big\langle \sum_{i \in I}\lambda_{ij}\,e_i, \sum_{\nu \in I}
\lambda_{\nu j}\,e_\nu\big\rangle_H = \sum_{i \in I}\sum_{j \in J}\vert 
\lambda_{ij}\vert^2 = \Vert z \Vert_{H \,{\otimes}\, K}^2.
\]
\end{proof}
\begin{example}[\bf{The Hilbert space $\C^m_2 \,{\otimes}\, \C^n_2$ and the Kronecker 
product}]\label{ex:Kronecker_product}
Fix $m,n \in \N$. \sref{Proposition}{prop:tensor_product_as_direct_sum} clearly 
implies the well-known isometric identity
\[
\C^n_2 \,{\otimes}\, \C^m_2 \cong \C^{nm}_2 = \C^{mn}_2 \cong \C^m_2 \,{\otimes}\, \C^n_2\,,
\]
where $\C^k_2$ denotes the finite-dimensional vector space $\C^k$ (of 
column vectors), equipped with the Euclidean norm $\Vert \cdot \Vert_2$ 
($k \in \N$). Let $(x,y) \in \C^m \times \C^n$. The construction of $\psi$ 
implies that $\psi(y \,{\otimes}\, x) = (x_1 y, \ldots, x_m y) \in (\C^n)^m$. 
Thus, the following concrete mapping rule is well-defined: 
\[
\C^m_2 \,{\otimes}\, \C^n_2 \ni y \,{\otimes}\, x \mapsto 
x \otimes_{\text{Kr}} y \equiv \text{vec}(y x^\top) \equiv
\text{vec}({\mid\! y\rangle}{\langle x\!\mid}) \in \C^{mn}, 
\]
where the - noncommutative - Kronecker product of $x \in \C^n \cong 
\M_{n,1}(\C)$ and $y \in \C^m \cong \M_{m,1}(\C)$ is given by 
\[
x \otimes_{\text{Kr}} y : = (x_1 y^\top\,\brokenvert\,x_2 y^\top\,\brokenvert\,
\ldots\brokenvert\,x_n y^\top)^\top \in \C^{mn} \cong \M_{mn,1}(\C)\,.  
\]
$\text{vec}(y x^\top) : = (x_1 y_1, \ldots, x_1 y_m,\ldots, x_n y_1, \ldots, 
x_n y_m)^\top \in \C^{mn} \cong \M_{mn,1}(\C)$ denotes the \textit{column} 
vector constructed by stacking the columns of the matrix $y x^\top \in 
\M_{m,n}(\C)$ on top of each other (cf., e.g., {\cite[Section 1.2 and 
(3.4.4)]{Oe2024}}). Observe also that 
\[
\text{vec}(\tp{x}{-0.4ex}{\scbox{0.70}{$\C_2^n$}}{-0.45ex}{\scbox{0.70}
{$\C_2^m$}}{y}) \stackrel{\eqref{eq:rank_one_matrices_II}}{=} \text{vec}
(y x^\top) = \text{vec}(y \overline{x}^\adj) \equiv x \otimes_{\text{Kr}} y 
\stackrel{\eqref{eq:rank_one_matrices}}{=} \text{vec}(\tp{J_{\C_2^n} x}{-0.4ex}
{\scbox{0.70}{$\overline{\C_2^n}$}}{-0.45ex}{\scbox{0.70}{$\C_2^m$}}{y})\,.
\]
Consequently, since
\[
\text{vec} : (\M_{m,n}(\C), \sigma_2(\cdot)) 
\stackrel{\cong}{\longrightarrow} \C_2^{mn}, A \mapsto \text{vec}(A)
\]
is an isometric isomorphism, it follows from 
\sref{Example}{eq:conjugate_HS_in_the_separable_case}, 
\eqref{eq:rank_one_matrices} and 
\sref{Corollary}{cor:char_of_tensor_product_of_HS_as_Hilbert_Schmidt} that also
\begin{align}\label{eq:HS_tensor_product_vs_Kronecker_product}
\C_2^n \otimes \C_2^m \cong \overline{\C_2^n} \otimes \C_2^m \cong 
(\M_{m,n}(\C), \sigma_2(\cdot)) \cong \C_2^{mn} \cong 
(\M_{mn,1}(\C), \sigma_2(\cdot))  
\end{align}
is a (linear) isometric isomorphism, which maps an arbitrary 
$\sum\limits_k x_k \otimes y_k \in \C_2^n \otimes \C_2^m$ onto the 
Kronecker product $\sum\limits_k x_k \otimes_{\text{Kr}}\,y_k \in 
\M_{mn,1}(\C)$ (via $\sum\limits_k x_k \otimes y_k \mapsto \sum\limits_k
J_{\C_2^n} x_k \otimes y_k \mapsto \sum\limits_k y_k \overline{x_k}^{\,\adj} = 
\sum\limits_k y_k x_k^\top \mapsto \sum\limits_k\text{vec}(y_k x_k^\top) \equiv 
\sum\limits_k x_k \otimes_{\text{Kr}}y_k$). However, 
in general $x \otimes_{\text{Kr}} y \not= y \otimes_{\text{Kr}} x$ (even if 
$m=n$). Here, we should recall the fact (for example, from 
{\cite[Section 3.4]{Oe2024}}) that
\begin{align*}
\sum\limits_k x_k \otimes_{\text{Kr}} y_k = 
K_{m,n}\sum\limits_k y_k \otimes_{\text{Kr}} x_k\,,
\end{align*}
where the so called commutation matrix $K_{m,n} : = 
\sum_{i=1}^m\sum_{j=1}^n e_i\,e_j^\top \otimes_{\text{Kr}} e_j\,e_i^\top \in 
M_{mn}(\R)$ is unitary. In summary, we recognise that in 
the finite-dimensional case we just have to compute the -  explicitly 
computable - Kronecker product of two vectors in $\C^m$ and $\C^n$ to capture 
the tensor product structure of $\C^m_2 \,{\otimes}\, \C^n_2$ - and hence of any 
$k$-fold tensor product of finite-dimensional Hilbert spaces - completely 
(cf. \sref{Definition}{def:n_fold_tp} and 
\sref{Example}{ex:quantum_teleportation_matrix}).
\end{example}
At first glance, if $H, K$ and $L$ are  Hilbert spaces, it 
seems to be a simple exercise that the  Hilbert space $H \,{\otimes}\, 
(K \,{\otimes}\, L)$ coincides isometrically with the  Hilbert space 
$(H \,{\otimes}\, K) \,{\otimes}\, L$ (because of 
\sref{Theorem}{thm:Hilbert_Schmidt_and_tensor_product}-(vii) and 
\sref{Proposition}{prop:tensor_product_as_direct_sum}). However, the 
proof of the next result shows, that an explicit description of this (true) 
result in terms of Hilbert-Schmidt operators is not as obvious as it seems 
(cf. also \cite[Exercise 4.4 and Chapter 29]{DF1993}). 
\begin{theorem}\label{thm:associativity_law_of_HS_tensorprod}
Let $H, K, L$ be  Hilbert spaces. Put $H_1 : = 
\mathfrak{S}_2(\overline{K}, L)$ and $H_2 : = \mathfrak{S}_2(\overline{H}, K)$
Then the equalities 
\[
H \,{\otimes}\, (K \,{\otimes}\, L) \cong \mathfrak{S}_2(\overline{H}, 
H_1) \cong \mathfrak{S}_2(\overline{H_2}, L) \cong (H \,{\otimes}\, K) \,
{\otimes}\, L
\]
hold isometrically and canonically. Thereby, each tensor element 
$x \otimes (z \otimes {w}) \in H \,{\otimes}\, (K \,{\otimes}\, L)$ 
is isometrically mapped onto $(x \otimes z) \otimes {w} \in 
(H \,{\otimes}\, K) \,{\otimes}\, L$.
\end{theorem}
\begin{proof}
\sref{Theorem}{thm:Hilbert_Schmidt_and_tensor_product} allows us to provide
an explicit construction of the required isometric isomorphism 
between $\mathfrak{S}_2(\overline{H}, H_1)$ and $\mathfrak{S}_2
(\overline{H_2}, L)$. So, let $T \in \mathfrak{S}_2(\overline{H}, H_1)$ and 
$R \in H_2 = \mathfrak{S}_2(\overline{H}, K)$. Let $\{e_i : i \in I\}$ 
be an orthonormal basis of $H = \overline{\overline{H}}$ and 
$\{f_j : j \in J\}$ an orthonormal basis of $K$. Then $\{\tp{e_i}{-0.45ex}
{\scbox{0.70}{$H$}}{-0.5ex}{\scbox{0.70}{$K$}}{f_j} : (i,j) \in I \times J\}$ 
is an orthonormal basis of the Hilbert space $H_2$ (due to 
\sref{Theorem}{thm:Hilbert_Schmidt_and_tensor_product}-(vii)). H\"{o}lder's 
inequality therefore implies that
\[
w_{R,T} : = 
\sum_{(i,j) \in I \times J}\langle J_{H_2}R, J_{H_2}(\tp{e_i}{-0.45ex}
{\scbox{0.70}{$H$}}{-0.5ex}{\scbox{0.70}{$K$}}{f_j})\rangle_{\overline{H_2}}\,
(T\,J_H e_i)J_K f_j \stackrel{\eqref{eq:key_representation}}{=}
\sum_{(i,j) \in I \times J}\overline{\langle R\,J_H e_i, f_j\rangle_K}\,
(T\,J_H e_i)J_K f_j \in L
\] 
is well-defined. Indeed, we have: 
\begin{align*}
\Vert w_{R, T}\Vert_L &\leq \big(\sum\limits_{(i,j) \in I \times J}
\vert\langle R\,J_H e_i, f_j\rangle_K\vert^2\big)^{1/2}\,
\big(\sum\limits_{(i,j) \in I \times J}\Vert(T\,J_H e_i)
J_K f_j\Vert_L^2\big)^{1/2}\\
&= \big(\sum\limits_{(i,j) \in I \times J}
\big\vert\langle R\,J_H e_i, f_j\rangle_K\big\vert^2\big)^{1/2}\,
\big(\sum\limits_{i \in I}\Vert(T\,J_H e_i)\Vert_{H_1}^2\big)^{1/2}
\stackrel{\eqref{eq:HS_norm}}{\leq} 
\Vert R\Vert_{H_2}\,\Vert T\Vert_{\mathfrak{S}_2}\,.
\end{align*} 
Moreover, $w_{R,T}$ does not depend on the choice of the orthonormal bases in 
$H$ and in $K$. \sref{Theorem}{thm:Hilbert_Schmidt_and_tensor_product}-(iii) 
namely implies on the one hand that
\begin{align}\label{eq:first_part}
\langle J_{H_2} R,\,J_{H_2}(\tp{\lambda e_i}{-0.45ex}
{\scbox{0.70}{$H$}}{-0.5ex}{\scbox{0.70}{$K$}}{\mu f_j})
\rangle_{\overline{H_2}} = 
\langle\tp{\lambda e_i}{-0.45ex}{\scbox{0.70}{$H$}}{-0.5ex}{\scbox{0.70}{$K$}}
{\mu f_j}, R \rangle_{H_2} =
\lambda\mu\langle \tp{e_i}{-0.45ex}{\scbox{0.70}{$H$}}{-0.5ex}
{\scbox{0.70}{$K$}}{f_j}, R \rangle_{H_2} 
\end{align}
for all $\lambda, \mu \in \C$ and $(i,j) \in I \times J$. On the other hand, 
the given structure of $T$ (including the definition of $H_1$) gives
\begin{align}\label{eq:second_part}
T\big(J_H(\lambda e_i)\big)\big(J_K(\mu f_j)\big) = 
T(\overline{\lambda}\ast J_H e_i)(\overline{\mu}\ast J_K f_j) = 
\overline{\lambda}\,\overline{\mu}(T J_H e_i)J_K f_j
\end{align}
for all $\lambda, \mu \in \C$ and $(i,j) \in I \times J$. Consequently, 
\[
\Psi : \mathfrak{S}_2(\overline{H}, H_1) \longrightarrow \mathfrak{S}_2
(\overline{H_2}, L), T \mapsto (R \mapsto w_{R, T})
\]
is a well-defined bounded canonical linear operator. In particular, 
\[
\Psi(T)\big(J_{H_2}(\tp{e_i}{-0.45ex}{\scbox{0.70}{$H$}}{-0.5ex}
{\scbox{0.70}{$K$}}{f_j})\big) = (T J_H e_i)J_K f_j 
\]
for all $(i,j) \in I \times J$. Next, we verify that $\Psi$ is an isometry. To 
this end, let $\{g_k : k \in K\}$ be an arbitrarily chosen orthonormal basis of 
$L$. \sref{Theorem}{thm:Hilbert_Schmidt_and_tensor_product}-(iv) implies that  
\[
\langle T\,J_H e_i, \tp{f_j}{-0.45ex}{\scbox{0.70}{$K$}}{-0.5ex}
{\scbox{0.70}{$L$}}{g_k}\rangle_{H_1} =
\big\langle(T J_H e_i)J_K f_j, g_k\big\rangle_L =
\big\langle\Psi(T)\big(J_{H_2}(\tp{e_i}{-0.45ex}{\scbox{0.70}{$H$}}{-0.5ex}
{\scbox{0.70}{$K$}}{f_j})\big), g_k\big\rangle_L
\]
for all $(i, j, k) \in I \times J \times K$. Hence,
\begin{align*}
\Vert \Psi(T)\Vert_{\mathfrak{S}_2}^2 & \stackrel{\eqref{eq:HS_norm}}{=} 
\sum_{((i,j), k) \in (I \times J) \times K}
\big\vert\big\langle\Psi(T)\big(J_{H_2}(\tp{e_i}{-0.45ex}
{\scbox{0.70}{$H$}}{-0.5ex}{\scbox{0.70}{$K$}}{f_j})\big), g_k\big\rangle_L
\big\vert^2\\ 
& \,\,\,= \sum_{(i, (j,k)) \in I \times (J\times K)} 
\big\vert\langle T J_H e_i, \tp{f_j}{-0.45ex}{\scbox{0.70}{$K$}}{-0.5ex}
{\scbox{0.70}{$L$}}{g_k}\rangle_{H_1}\big\vert^2
\stackrel{\eqref{eq:HS_norm}}{=} \Vert T\Vert_{\mathfrak{S}_2}^2\,.
\end{align*}
It remains to prove that $\Psi$ is onto. So, let $S \in \mathfrak{S}_2
(\overline{H_2}, L)$. It follows from \eqref{eq:HS_ONB_rep} that  
\[
S = \sum_{((i,j), k) \in (I \times J) \times K} \alpha_{ijk}\,
\big(\tp{e_i}{-0.45ex}{\scbox{0.70}{$H$}}{-0.5ex}
{\scbox{0.70}{$K$}}{f_j}\big)\tp{{}}{-0.4ex}{\scbox{0.70}{$H_2$}}{-0.5ex}
{\scbox{0.70}{$L$}}{g_k}\,, 
\]
where the family of  numbers $(\alpha_{ijk})_
{(i,j,k) \in I \times J \times K}$ satisfies $\Vert S\Vert_{\mathfrak{S}_2}^2 
\stackrel{\eqref{eq:HS_norm}}{=} \sum\limits_
{((i,j), k) \in (I \times J) \times K}\vert\alpha_{ijk}\vert^2 < \infty$. 
Consequently, the operator 
\[
T : = \sum_{i \in I}\sum_{(j,k) \in J \times K} \alpha_{ijk}\,\,
\tp{e_i}{-0.45ex}{\scbox{0.70}{$H$}}{-0.5ex}
{\scbox{0.70}{$H_1$}}{(\tp{f_j}{-0.45ex}{\scbox{0.70}{$K$}}{-0.5ex}
{\scbox{0.70}{$L$}}{g_k})} \in \mathfrak{S}_2(\overline{H}, H_1) 
\]
is well-defined (due to \eqref{eq:HS_ONB_rep} again). 
It namely satisfies
\[
\Vert T\Vert_{\mathfrak{S}_2}^2 = \sum\limits_{I \times (J \times K)} 
\vert\alpha_{i j k}\vert^2 = \sum\limits_{(I \times J) \times K} 
\vert\alpha_{i j k}\vert^2 = \Vert S\Vert_{\mathfrak{S}_2}^2 < \infty\,.
\] 
The construction of a dyad (cf. \eqref{eq:dyad}), together with the boundedness 
of $T$ clearly implies that for any $\nu \in I$, $T\,J_H e_\nu = 
\sum\limits_{(j,k) \in J \times K} \alpha_{\nu j k}\,
\tp{f_j}{-0.45ex}{\scbox{0.70}{$K$}}{-0.5ex}
{\scbox{0.70}{$L$}}{g_k}$.
Consequently,  
\[
\Psi(T)(\tp{e_\nu}{-0.45ex}{\scbox{0.70}{$H$}}{-0.5ex}
{\scbox{0.70}{$K$}}{f_\mu}) = (T J_H e_\nu)J_K f_\mu = \sum\limits_{k \in K}
\alpha_{\nu \mu k}\,g_k = S(\tp{e_\nu}{-0.45ex}{\scbox{0.70}{$H$}}{-0.5ex}
{\scbox{0.70}{$K$}}{f_\mu}) 
\]
for all $(\nu, \mu) \in I \times J$, and claim (vii) follows.
\end{proof}
Due to \sref{Theorem}{thm:associativity_law_of_HS_tensorprod}, the following 
recursive definition almost inevitably emerges:
\begin{definition}[\textbf{$n$-fold tensor product of Hilbert spaces}]
\label{def:n_fold_tp}
Let $n \in \N$ and $H_1, \ldots, H_n$ be Hilbert spaces. Put
\begin{align*}
{\bigotimes\limits_{i=1}^n}\, 
H_i : = \begin{cases} H_1 &\text{if } n=1\\ 
\big({\bigotimes\limits_{i=1}^{n-1}} 
\,H_i\big) \otimes H_n &\text{if } n > 1\,.\end{cases}
\end{align*}
\end{definition}   
Firstly, by means of a trivial proof by induction on $n$, it follows from 
\sref{Theorem}{thm:associativity_law_of_HS_tensorprod} at once that for any 
$n \in \N_{\geq 2}$,
\begin{align}\label{eq:rep_of_the_n_fold_tens_prod}
{\bigotimes\limits_{i=1}^n}\, 
H_i \cong H_1 \otimes \lb\bigotimes\limits_{i=2}^n\,H_i\rb 
\stackrel{(!)}{\cong} \id{S}{2}{}{}{}\lb \overline{H_1}, \,\bigotimes\limits_
{i=2}^n\,H_i\rb\,.
\end{align}
Consequently, a further proof by induction (now on the two variables $m \in \N$ 
and $n \in \N$), together with \eqref{eq:rep_of_the_n_fold_tens_prod} and a 
reiterative use of \sref{Theorem}{thm:associativity_law_of_HS_tensorprod} 
instantly leads to the ``multi-dimensional associativity rule'' for the 
$m+n$-fold tensor product (cf. \cite[Proposition 2.6.5]{KR1983}):
\begin{proposition}
Let $(m, n) \in \N \times \N$ and $H_1, \ldots, H_{m+n}$ $m+n$ Hilbert 
spaces. Then 
\[
\big({\bigotimes\limits_{i=1}^m}\, 
H_i\big) \otimes \big({\bigotimes\limits_{i=m+1}^{m+n}}\,H_i\big) \cong 
{\bigotimes\limits_{i=1}^{m+n}}\,H_i\,. 
\]
\end{proposition}
\noindent Despite the apparent mathematical simplicity of our next example 
(\sref{Example}{ex:quantum_teleportation_matrix}), implied by an application of 
the Kronecker product, it encodes a significant interpretation of the surprising 
effect of correlations in quantum mechanics. Details on its importance in the 
foundations of quantum mechanics, including a neat derivation (which we roughly 
outline under the proof of our purely linear algebraic construction) and a very 
detailed discussion of the postulates of quantum mechanics can be found in the 
highly recommendable versions of \cite[Chapter 1.3.7 and Chapter 2.2]{NC2010}.

{In order to align our linear algebraic approach with previous methods}, 
we adopt Dirac's bra-ket formalism again. As usual in quantum information 
theory, the standard orthonormal basis $\{\vert 0 \rangle, \vert 1 \rangle, 
\ldots, \vert n-1 \rangle\} \equiv \{e_1, e_2 \ldots, e_n\}$ of $\C^n$ is called 
\textit{computational basis}. Occasionally, to emphasize the dependence on the 
dimension $n$, we put $e_i^{(n)} \equiv e_i \equiv \vert i-1 \rangle$ if $i \in [n]$.
Note the important fact that 
\begin{align}\label{eq:Kronecker_product_of_standard_ONB_vectors} 
e_{j}^{(n)} \otimes_{\text{Kr}} e_{i}^{(m)} = e_{(j-1)m+i}^{(nm)} 
\end{align}
for all $(i, j) \in [m] \times [n]$ (regarding details, cf. \cite[Chapter 3.4]
{Oe2024}). In particular,
\begin{align}\label{eq:Kronecker_product_of_standard_2_dim_ONB_vectors}
e_{i}^{(m)} \otimes_{\text{Kr}} e_{i}^{(m)} = e_{i(m+1)-m}^{(m^2)}\,,\, 
e_{1}^{(2)} \otimes_{\text{Kr}} e_{i}^{(m)} = e_i^{(2m)}\,\text{ and }\,
e_{2}^{(2)} \otimes_{\text{Kr}} e_{i}^{(m)} = e_{i+m}^{(2m)}
\end{align}
for all $i \in [m]$. Recall that in the case $n=2$ a unit (column) vector 
$(\alpha, \beta)^\top \equiv \alpha \vert 0 \rangle + \beta \vert 1 \rangle$ is 
called a \textit{qubit} (short for \textit{quantum bit}). Unlike a classic 
binary bit corresponding to exactly one of two possible states of a classical 
system (either 0 or 1), the quantum bit is a model of a coherent superposition 
of both states simultaneously; very pictorially described in the form of 
Schr\"odinger's famous cat-thought experiment: ``$\vert \text{cat} \rangle = 
\frac{1}{\sqrt{2}}\vert \text{cat is alive} \rangle + \frac{1}{\sqrt{2}}\vert 
\text{cat is dead} \rangle$''. 

Because of the well-known associativity of the Kronecker product (respectively 
\sref{Theorem}{thm:associativity_law_of_HS_tensorprod}), the common approach in 
quantum information theory to drop the tensor product sign ``$\otimes$'' in 
tensor product elements, describing ``correlated'' qubits (such as $\vert 0 1 \rangle
\vert 1 \rangle \equiv \vert 0\rangle \vert 1 1 \rangle \equiv \vert 0 1 1 \rangle 
\equiv \vert 0 \rangle \otimes \vert 1 \rangle \otimes \vert 1 \rangle \equiv 
e_1^{(2)} \otimes e_2^{(2)} \otimes e_2^{(2)}$) is quite useful. That approach 
allows the use of a rather compact symbolic language which is particularly 
helpful when the structure of correlated state vectors has to be taken into 
account (i.e., entangled state vectors - cf., e.g., \cite{NC2010}), as is the 
case for the following crucial
\begin{example}[\textbf{Quantum teleportation matrix}]
\label{ex:quantum_teleportation_matrix}
Consider the unit vector $\phi^+ : = \frac{1}{\sqrt{2}}(1, 0, 0, 1)^\top \in 
\C^4$ and
\[
T : =  \frac{1}{\sqrt{2}}
\begin{pmatrix}[cccc:cccc]
1 & 0 & 0 & 0 & \,\,\,0 & 0 & 1 & 0\\
0 & 1 & 0 & 0 & \,\,\,0 & 0 & 0 & 1\\
0 & 0 & 1 & 0 & \,\,\,1 & 0 & 0 & 0\\
0 & 0 & 0 & 1 & \,\,\,0 & 1 & 0 & 0\\ \hdashline
1 & 0 & 0 & 0 & \,\,\,0 & 0 & \!\!\!\!-1 & 0\\
0 & 1 & 0 & 0 & \,\,\,0 & 0 & 0 & \!\!\!\!-1\\
0 & 0 & 1 & 0 & \,\,\,\!\!\!\!-1 & 0 & 0 & 0\\
0 & 0 & 0 & 1 & \,\,\,0 & \!\!\!\!-1 & 0 & 0\\
\end{pmatrix}
\equiv \frac{1}{\sqrt{2}}
\begin{pmatrix}[cc:cc]
Id_2 & 0 & 0 & Id_2\\ 
0 & Id_2 & Id_2 & 0\\ \hdashline
Id_2 & 0 & 0 & -Id_2\\
0 & Id_2 & -Id_2 & 0
\end{pmatrix} \in \M_8(\C)\,.
\] 
Then 
\[
\phi^+ = \frac{1}{\sqrt{2}}(e_1^{(4)} + e_4^{(4)}) = 
\frac{1}{\sqrt{2}}(e_1^{(2)} \otimes_{\text{Kr}} e_1^{(2)} + 
e_2^{(2)} \otimes_{\text{Kr}} e_2^{(2)}) \equiv \frac{\vert 00\rangle + 
\vert 11\rangle}{\sqrt{2}}.
\]
Moreover, there do not exist $a, b \in \C^2$, such that $\phi^+ = 
a \otimes_{\text{Kr}} b$. $T \in U(8)$ is a unitary matrix, and
\begin{align}\label{eq:quantum_teleportation_equation}
T(\emphas{\xi} \otimes_{\text{Kr}} \phi^+) = 
\frac{1}{2}(\alpha, \beta, \beta, \alpha,\alpha, -\beta, -\beta, \alpha)^\top = 
\frac{1}{2}\lb\sum\limits_{i=1}^4 e_i^{(4)} \otimes_{\text{Kr}} T_i\,
\emphas{\xi}\rb\,\emphas{\text{ for all }} \emphas{\xi = (\alpha, \beta)^\top 
\in \C^2},
\end{align}
where $T_1 : = Id_2$ and $T_2, T_3, T_4 \in U(2)$ denote the unitary Pauli 
matrices 
\[
T_2 : = \sigma_x : = \begin{pmatrix}
0 & 1\\
1 & 0
\end{pmatrix}, 
T_3 : = \sigma_z : = \begin{pmatrix}
1 & 0\\
0 & -1
\end{pmatrix} \text{ and } 
T_4 : = \sigma_x\sigma_z = \frac{1}{i}\,\sigma_y : = 
\begin{pmatrix}
0 & -1\\
1 & 0
\end{pmatrix}.
\]
In particular,
\begin{align}\label{eq:solving_for_xi}
T_i^{-1}\colvec{2}{w_{2i -1}}{w_{2i}} = \xi\,\text{ for all }\,
i \in [4]\,,
\end{align}
where $w : = T(\xi \otimes_{\text{Kr}} \phi^+)$. 
Moreover,
\begin{align}\label{eq:factorisation}
T = (H_1 \otimes_{\text{Kr}} Id_4)(U_{\tup{CNOT}} \otimes_{\text{Kr}} Id_2) = 
\begin{pmatrix}[c:c]
Id_4 & Id_4\\ \hdashline
Id_4 & -Id_4
\end{pmatrix}
\begin{pmatrix}[cc:cc]
Id_2 & 0 & 0 & 0\\ 
0 & Id_2 & 0 & 0\\ \hdashline
0 & 0 & 0 & Id_2\\
0 & 0 & Id_2 & 0
\end{pmatrix},
\end{align}
where the unitary matrices $H_1 \in U(2)$ and $U_{\tup{CNOT}} \in U(4)$ are 
given by 
\[
H_1 : = 
\tfrac{1}{\sqrt{2}}\begin{pmatrix}
1 & 1\\
1 & -1
\end{pmatrix} = \tfrac{1}{\sqrt{2}}(T_2 + T_3)
\hspace{2em}(\textit{Hadamard gate})
\]
and
\[
U_{\tup{CNOT}} : =
\begin{pmatrix}[cc:cc]
1 & 0 & 0 & 0\\ 
0 & 1 & 0 & 0\\ \hdashline
0 & 0 & 0 & 1\\
0 & 0 & 1 & 0
\end{pmatrix}  
= 
\begin{pmatrix}
1 & 0\\
0 & 1
\end{pmatrix} 
\oplus
\begin{pmatrix}
0 & 1\\
1 & 0
\end{pmatrix} = Id_2\,\oplus\,U_{\tup{NOT}} \hspace{1em}
(\textit{CNOT gate -- ``controlled NOT''}).
\] 
\begin{cp*}
Because of the structure of $T$ and $\phi^+$ and the straightforward 
calculation rules for the Kronecker product (cf., e.g.,  
\cite[Subsection 1.2]{Oe2024}), the proof is just an application of some 
elementary linear algebra with matrices. Of course, the relationship between 
the  Hadamard gate $H_1$, the CNOT gate $U_{\tup{CNOT}}$ and quantum information 
per se is somewhat hidden in our purely linear algebraic approach. The meaning 
of $H_1$ and $U_{\tup{CNOT}}$, including the construction of a related circuit 
and the fundamental role of the transmission of a measurement result over a 
classical communications channel (Alice has to send classical information to 
Bob), is described in detail in \cite[Example 1.3.7]{NC2010}. A closer look at 
\cite[Example 1.3.7]{NC2010} also reveals the crucial (albeit somewhat hidden) 
role of the Kronecker product therein.    

Due to \eqref{eq:Kronecker_product_of_standard_2_dim_ONB_vectors}, the 
representation of $\phi^+$ as a sum of the two Kronecker product parts is 
straightforward, as is the case for the factorisation \eqref{eq:factorisation} 
of $T$. It has a fundamental meaning in the foundations of quantum mechanics 
and quantum information, though; namely in the sense of a so called 
\textit{Bell state vector} (also known as \textit{EPR pair}), named after Bell, 
and Einstein, Podolsky and Rosen, who first pointed out properties of such 
state vectors (cf., e.g., \cite[Example 1.3.6]{NC2010}). 

So, let us fix an arbitrarily chosen $\xi = (\alpha, \beta)^\top = 
\alpha e_1^{(2)} + \beta e_2^{(2)} \in \C^2$ (an arbitrary qubit, for example). 
A straightforward calculation of the threefold Kronecker product (including an 
application of the associativity law, together with 
\eqref{eq:Kronecker_product_of_standard_ONB_vectors} and 
\eqref{eq:Kronecker_product_of_standard_2_dim_ONB_vectors}) implies that 
\begin{align}\label{eq:qubit_otimes_Bell_state}
\xi \otimes_{\text{Kr}} \phi^+ = \frac{\alpha}{\sqrt{2}}\sum\limits_{i=1}^2 
e_i^{(4)} \otimes_{\text{Kr}} e_i^{(2)} + \frac{\beta}{\sqrt{2}}
\sum\limits_{i=1}^2 e_{i+2}^{(4)} \otimes_{\text{Kr}} e_i^{(2)} = 
\frac{1}{\sqrt{2}}(\alpha, 0, 0, \alpha,\beta, 0, 0, \beta)^\top\,.
\end{align}
Consequently, the explicit $2^3$-tuple representation 
\eqref{eq:qubit_otimes_Bell_state} of $\xi \otimes_{\text{Kr}} \phi^+$ implies 
that 
\begin{align}
\begin{split} 
T(\xi \otimes_{\text{Kr}} \phi^+) 
&\stackrel{\eqref{eq:factorisation}}{=} \frac{1}{\sqrt{2}}
\begin{pmatrix}[c:c]
Id_4 & Id_4\\ \hdashline
Id_4 & -Id_4
\end{pmatrix}
(\alpha, 0, 0, \alpha, 0, \beta, \beta, 0)^\top = 
\frac{1}{2}(\alpha, \beta, \beta, \alpha,\alpha, -\beta, -\beta, \alpha)^\top\\
&\,\,\,= \frac{1}{2}\lb(e_1^{(4)} \otimes_{\text{Kr}} \binom{\alpha}{\beta} + 
e_2^{(4)} \otimes_{\text{Kr}} \binom{\beta}{\alpha} + 
e_3^{(4)} \otimes_{\text{Kr}} \binom{\alpha}{-\beta} + 
e_4^{(4)} \otimes_{\text{Kr}} \binom{-\beta}{\alpha}\rb,
\end{split}
\end{align} 
and \eqref{eq:quantum_teleportation_equation} follows. \eqref{eq:solving_for_xi} 
is a direct implication of the concrete tupel representation of the vector 
$w \in \C^8$ in \eqref{eq:quantum_teleportation_equation}. A straightforward 
(block) matrix multiplication shows that
\[
T^\adj = T^\top = \frac{1}{\sqrt{2}}
\begin{pmatrix}[cc:cc]
Id_2 & 0 & Id_2 & 0\\ 
0 & Id_2 & 0 & Id_2\\ \hdashline
0 & Id_2 & 0 & -Id_2\\
Id_2 & 0 & -Id_2 & 0
\end{pmatrix} = T^{-1},
\]
implying that $T \in U(8)$ is a unitary matrix.\qed
\end{cp*}
\end{example}
\begin{remark}
Observe that \eqref{eq:quantum_teleportation_equation} clearly shows that 
quantum teleportation does not violate the no-cloning theorem. Bob's qubit 
$\emphas{\xi}$, ``encapsulated'' in $\frac{1}{2}\lb\sum\limits_{i=1}^4 e_i^{(4)} 
\otimes_{\text{Kr}} T_i\,\emphas{\xi}\rb$ is not a copy of Alice's one (in 
$\emphas{\xi} \otimes_{\text{Kr}} \phi^+$). To perform the transmission 
$\emphas{\xi} \otimes_{\text{Kr}} \phi^+ \stackrel{T}{\mapsto} \frac{1}{2}
\lb\sum\limits_{i=1}^4 e_i^{(4)} \otimes_{\text{Kr}} T_i\,\emphas{\xi}\rb$, 
Alice is namely required to destroy her qubit $\xi$.  

A simple proof by contradiction shows that $H_1$ \textit{can't} be written as a 
Kronecker product of a $2$-dimensional complex row (resp. column) vector and a 
$2$-dimensional complex column (resp. row) vector:
\[
H_1 \not= (a_1, a_2) \otimes_{\tup{Kr}} \colvec{2}{a_3}{a_4} = 
\lb\colvec{2}{a_1}{a_2} \otimes_{\tup{Kr}} (a_3, a_4)\rb^{\!\!\top}\, 
\text{ for all } a_1, a_2, a_3, a_4 \in \C\,. 
\] 
\end{remark}
\section{Grothendieck is lurking: normed operator ideals and their corresponding 
enveloping  $\tup{C}^\adj$-algebras}
\label{sec:op_ideals_and_Grothendieck}
\noindent To gain an improved understanding of \sref{Theorem}{thm:split_char}, 
it is necessary to sketch how the tensor product of Hilbert spaces actually is 
implemented in the construction of the spatial tensor product of two von 
Neumann algebras. A clear and concise introduction to von Neumann algebras 
(and $\tup{C}^\adj$-algebras in general) can be found in e.g. 
\cite{B2006, BR1987, C2000, M1990, P1989, P2018, W2008, W2010, Z2008}. All 
details about the spatial tensor product are provided in 
\cite[II.9.1.3 and III.1.5.4]{B2006}, \cite[Chapter 2.7.2]{BR1987} and 
\cite[Chapter 6.3]{M1990}. Firstly, recall that a \textit{von Neumann algebra}, 
acting on some Hilbert space $H$, is a unital $\adj$-subalgebra of 
$\id{L}{}{}{}{}(H)$ which is closed in the weak operator topology (WOT); i.e., 
in the topology on $\id{L}{}{}{}{}(H)$, which is generated by the seminorms 
$\id{L}{}{}{}{}(H) \ni T \mapsto \Vert T\Vert_{x,y} : = 
\vert\langle Tx, y \rangle_H\vert$ ($x, y \in H$). Equivalently (due to von 
Neumann's Bicommutant Theorem), a $\adj$-subalgebra $\mathcal{N} \subseteq 
\id{L}{}{}{}{}(H)$ is a von Neumann algebra if and only if $\mathcal{N} = 
\mathcal{N}^{\,\#\#}$, where $\mathcal{R}^{\,\#} : = \{T \in \mathfrak{L}(H) : 
TS = ST\,\text{ for all }\, S \in \mathcal{R}\}$ denotes the commutant of 
$\mathcal{R} \subseteq \id{L}{}{}{}{}(H)$. (To avoid 
confusion with the symbolic notation of the topological dual of a 
locally convex space, we adopt the notation of \cite{W2011} and do not 
make use of the symbol ``\,$\mathcal{R}^{\,\prime}$\,'', although the latter 
is rather common among physicists and researchers in the operator algebra 
community.) In particular, since the WOT is weaker than the norm 
topology, any von Neumann algebra is a unital $\tup{C}^\adj$-algebra (but not 
conversely). \cite[\textsection\,12]{W2010} sheds light on the highly non-trivial 
relationship between unital $\tup{C}^\adj$-algebras and von Neumann algebras 
in general. In any case, the seminal Gel'fand-Na\u{\i}mark-Segal construction implies 
that every $\tup{C}^\adj$-algebra $\mathcal{A}$ is isometrically 
$\adj$-isomorphic to a norm-closed $\adj$-subalgebra of $\id{L}{}{}{}{}
(H_{\mathcal{A}})$ on some ``large'' complex Hilbert space $H_{\mathcal{A}}$. If 
$\mathcal{A}$ is separable, then also the Hilbert space $H_{\mathcal{A}}$ is 
separable. These already very deep results are rounded off by the 
possibly less well-known, yet very fruitful classical cornerstone 
theorem of S. Sherman and Z. Takeda, stating that the second dual Banach space 
$\mathcal{A}^{\prime\prime}$ of any $\tup{C}^\adj$-algebra $\mathcal{A}$ can be 
identified as the enveloping von Neumann algebra of $\mathcal{A}$ (cf., e.g., 
\cite{CY1961, SP2026}):
\begin{theorem}[\textbf{Sherman-Takeda}] 
If $\mathcal{A}$ is a $C^\adj$-algebra and $\pi : \mathcal{A} 
\stackrel{1}{\hookrightarrow} \id{L}{}{}{}{}(H)$ is its universal 
Gel'fand-Na\u{\i}mark representation on some complex Hilbert space $H$, then the 
bidual $\mathcal{A}^{\prime\prime}$ is a von Neumann algebra which is 
isometrically $\adj$-isomorphic to the von Neumann algebra 
$\pi(\mathcal{A})^{\,\#\#} = \overline{\pi(\mathcal{A})}^{\tup{WOT}}$.
\end{theorem}  
\begin{remark}[\bf{Spatial tensor product of von Neumann algebras versus minimal 
tensor product of $\tup{C}^\adj$-algebras}]
\label{rem:spatial_tp_and_minimal_tp}
Let $\mathcal{M} \subseteq \mathfrak{L}(H)$ and $\mathcal{N} \subseteq 
\mathfrak{L}(K)$ be two von Neumann algebras, acting on the Hilbert spaces 
$H$ and $K$, respectively. Then there is a unique \textit{injective} 
$\adj$-homomorphism
\[
\pi : \mathcal{M} \otimes_\adj \mathcal{N} \longrightarrow 
\mathfrak{L}(H \otimes K)\,,
\]
satisfying $\pi(S \otimes T)(x \otimes y) = Sx \otimes Ty$, $(S, T) \in 
\mathcal{M} \times \mathcal{N}, (x, y) \in H \times K$, where $\mathcal{M} 
\otimes_\adj \mathcal{N}$ denotes the $\adj$-algebra tensor product of 
$\mathcal{M}$ and $\mathcal{N}$ (cf., e.g., \cite[Theorem 6.3.3]{M1990}).

The \textit{WOT-closure} of $\pi(\mathcal{M} \otimes_\adj \mathcal{N})$ is 
called the \textit{spatial tensor product of $\mathcal{M}$ and $\mathcal{N}$}. 
The resulting \textit{von Neumann algebra} is denoted by $\mathcal{M}\,
\overline{\otimes}\,\mathcal{N}$:
\[
\mathcal{M}\,\overline{\otimes}\,\mathcal{N} : = 
\overline{\pi(\mathcal{M} \otimes_\adj \mathcal{N})\,}^{\textup{WOT}}\,.
\]
Moreover, viewed solely through the lens of the C${}^{\,\adj}$-structure, 
the injectivity of $\pi$ implies that   
\[
\mathcal{M} \otimes_\adj \mathcal{N} \ni z \mapsto \Vert z \Vert_{\adj} : = 
\Vert \pi(z)\Vert
\] 
is a well-defined C${}^{\,\adj}$-norm on $\mathcal{M}\,\otimes_{\adj}\,
\mathcal{N}$, which satisfies $\Vert S \otimes T \Vert_{\adj} = 
\Vert S \Vert\,\Vert T \Vert$ for all $(S, T) \in \mathcal{M} \times 
\mathcal{N}$. The \textit{completion of $\mathcal{M} \otimes_{\adj} \mathcal{N}$ 
with respect to $\Vert \cdot \Vert_{\adj}$} is called the \textit{minimal tensor 
product of $\mathcal{M}$ and $\mathcal{N}$}. The resulting 
\textit{$C^\adj$-algebra} is denoted by $\mathcal{M} \otimes_{\text{min}} 
\mathcal{N}$. Observe that by construction
\[
\widetilde{\pi}(\mathcal{M} \otimes_{\text{min}} \mathcal{N}) = 
\overline{\pi(\mathcal{M} \otimes_{\adj} \mathcal{N})}^{\Vert\cdot\Vert} 
\subseteq \mathcal{M}\,\overline{\otimes}\,\mathcal{N}\,,
\]
where $\widetilde{\pi}$ denotes the completion of $\pi$ (with respect to 
$\Vert\cdot\Vert_{\adj}$, of course). However, in general, $\mathcal{M}\,
\overline{\otimes}\,\mathcal{N}$ does not coincide with 
$\widetilde{\pi}(\mathcal{M} \otimes_{\text{min}} \mathcal{N})$ (since the 
operator norm closure of $\pi(\mathcal{M} \otimes_{\tiny{\textup{alg}}} 
\mathcal{N})$ can be strictly smaller than the WOT-closure of 
$\pi(\mathcal{M} \otimes_{\tiny{\textup{alg}}} \mathcal{N})$)! A minor warning 
regarding notation: the minimal tensor product of two C${}^{\,\adj}$-
algebras is commonly referred to as spatial tensor product as well.          
\end{remark}

{Especially interested readers might raise the question whether the class 
of all Hilbert-Schmidt operators, i.e., the tensor product of Hilbert spaces 
(\sref{Theorem}{thm:structure_of_Hilb_space_tp}), plays a prominent role in 
quantum physics \textit{beyond} its well-known modelling of composite 
physical systems. To get a first answer to this question, we must examine the 
underlying operator ideal structure of Hilbert-Schmidt operators.

For the convenience of the reader let us recall the fundamental concept of an 
operator ideal, where we adopt \cite[Definition 1.4 and Definition 2.1]
{DJP2001}:

Assume that for all Banach spaces $E, F$, there is a nonempty set 
$\id{A}{}{}{}{}(E, F)$ which is contained in $\id{L}{}{}{}{}(E, F)$. Then
\[
\id{A}{}{}{}{} : = \bigcup\limits_{E, F \in\,\tup{BAN}}\id{A}{}{}{}{}(E, F)
\]
is said to be an \textit{operator ideal}, if the following conditions are
satisfied for all Banach spaces $E, F, E_1$ and $F_1$:\\[0.1em]
\begin{description}
\item[$(\textbf{OI}_1)$] $\langle\bcdot, a \rangle y \in \id{A}{}{}{}{}(E, F)$ 
for $(a, y) \in E^\prime \times F$.
\item[$(\textbf{OI}_2)$] $S + T \in \id{A}{}{}{}{}(E, F)$ for $S, T \in 
\id{A}{}{}{}{}(E, F)$.
\item[$(\textbf{OI}_3)$] $BTA \in \id{A}{}{}{}{}(E_1, F_1)$ for $T \in 
\id{A}{}{}{}{}(E, F)$, $A \in \id{L}{}{}{}{}(E_1, E)$ and $B \in \id{L}{}{}{}{}
(F, F_1)$.
\end{description}
A mapping $\idn{A}{}{}{}{} : \id{A}{}{}{}{} \longrightarrow [0, \infty)$ is 
called an \textit{operator ideal quasi-norm} if the following conditions are 
satisfied for all Banach spaces $E, F, E_1$ and $F_1$:  
\begin{description}
\item[$(\textbf{q-IN}_1)$] $\idn{A}{}{}{}{}(\langle\bcdot, a \rangle y) = 
\Vert a\Vert\,\Vert x\Vert$ for $(a, y) \in E^\prime \times F$.
\item[$(\textbf{q-IN}_2)$] There exists a constant $c \geq 1$, such that 
$\idn{A}{}{}{}{}(S+T) \leq c\,(\idn{A}{}{}{}{}(S) + \idn{A}{}{}{}{}(T)$) for 
$S, T \in \id{A}{}{}{}{}(E, F)$ (quasi-triangle inequality).
\item[$(\textbf{q-IN}_3)$] $\idn{A}{}{}{}{}(BTA) \leq \Vert B\Vert
\idn{A}{}{}{}{}(T) \Vert A\Vert$ for $T \in \id{A}{}{}{}{}(E, F)$, $A \in 
\id{L}{}{}{}{}(E_1, E)$ and $B \in \id{L}{}{}{}{}(F, F_1)$.
\end{description} 
Here, the constant $c\geq 1$ does not depend on the underlying Banach spaces 
$E$ and $F$. Let $0 < p \leq 1$. If condition $(\textbf{q-IN}_2)$ is replaced 
through
\begin{description}
\item[$(p\,\textbf{IN}_2)$] $\idn{A}{}{}{}{}(S+T)^p \leq \idn{A}{}{}{}{}(S)^p 
+ \idn{A}{}{}{}{}(T)^p$ for $S, T \in \id{A}{}{}{}{}(E, F)$ ($p$-triangle 
inequality), 
\end{description}
then $\idn{A}{}{}{}{}$ is known as \textit{operator ideal} $p$-norm. The 
$2$-tuple $(\id{A}{}{}{}{}, \idn{A}{}{}{}{})$ is said to be a 
\textit{quasi-normed operator ideal}, resp. a $p$-normed operator ideal (if 
$(\textbf{q-IN}_2)$ is replaced by $(p\,\textbf{IN}_2)$). If all components 
$\id{A}{}{}{}{}(E, F)$ of a normed ($p=1$) operator ideal  
$(\id{A}{}{}{}{}, \idn{A}{}{}{}{})$ are complete, $(\id{A}{}{}{}{}, 
\idn{A}{}{}{}{})$ is a \textit{Banach operator ideal}. Quasi-Banach operator 
ideals and $p$-Banach operator ideals are defined analogously. It is well-known 
that an operator ideal quasi-norm $\idn{A}{}{}{}{}$ need not be continuous in 
its own topology (cf. \cite[Example 6.1.9]{P1980}). However, if $(\id{A}{}{}{}{}, 
\idn{A}{}{}{}{})$ is a given quasi-normed operator ideal with constant 
$c \geq 1$, then $\idn{A}{}{}{}{}$ is equivalent to a $p_c$-norm on 
$\id{A}{}{}{}{}$, where $0 < p_c : = \ln(2)/(\ln(2)+\ln(c)) \leq 1$, and 
the latter is continuous in its own topology (\cite[Theorem 6.2.5]{P1980}). 
Conversely, every $p$-normed operator ideal ($0 < p \leq 1)$ is a quasi-normed 
one, with $1 \leq c : = 2^{1/p-1}$ (\cite[Remark 6.2.1-(1)]{P1980}).

Starting from two given quasi-normed operator ideals $(\id{A}{}{}{}{}, 
\idn{A}{}{}{}{})$ and $(\id{B}{}{}{}{}, \idn{B}{}{}{}{})$, it is possible to 
construct other quasi-normed operator ideals in various ways, including 
\begin{itemize}
\item[$-$] $(\id{A}{}{}{}{}\circ\id{B}{}{}{}{}, \idn{A}{}{}{}{}\circ
\idn{B}{}{}{}{})$\,\,\,\,(product ideal)
\item[$-$] $(\id{B}{}{}{}{}^{-1}\circ\id{A}{}{}{}{}, \idn{B}{}{}{}{}^{-1}\circ
\idn{A}{}{}{}{})$\,\,\,\,(left-quotient)
\item[$-$] $(\id{A}{}{}{}{}\circ\id{B}{}{}{}{}^{-1}, \idn{A}{}{}{}{}\circ
\idn{B}{}{}{}{}^{-1})$\,\,\,\,(right-quotient)
\end{itemize}
and
\begin{itemize}
\item[$-$] $(\id{A}{}{\tup{d}}{}{}, \idn{A}{}{\tup{d}}{}{})$\,\,\,\,(dual 
operator ideal)
\item[$-$] $(\id{A}{}{\tup{reg}}{}{}, \idn{A}{}{\tup{reg}}{}{})$\,\,\,\,
(regular hull) 
\item[$-$] $(\id{A}{}{\tup{min}}{}{}, \idn{A}{}{\tup{min}}{}{}) : = 
(\overline{\id{F}{}{}{}{}}\circ\id{A}{}{}{}{}\circ\overline{\id{F}{}{}{}{}}, 
\Vert\bcdot\Vert\circ\idn{A}{}{}{}{}\circ\Vert\bcdot\Vert)$\,\,\,\,(minimal 
kernel)
\item[$-$] $(\id{A}{}{\tup{max}}{}{}, \idn{A}{}{\tup{max}}{}{}) : = 
(\overline{\id{F}{}{}{}{}}^{\,-1}\circ\id{A}{}{}{}{}\circ
\overline{\id{F}{}{}{}{}}^{\,-1}, \Vert\bcdot\Vert^{\,-1}\circ
\idn{A}{}{}{}{}\circ\Vert\bcdot\Vert^{\,-1})$\,\,\,\,(maximal hull)
\item[$-$] $(\id{A}{}{\tup{inj}}{}{}, \idn{A}{}{\tup{inj}}{}{})$\,\,\,\,
(injective hull)
\item[$-$] $(\id{A}{}{\tup{sur}}{}{}, \idn{A}{}{\tup{sur}}{}{})$\,\,\,\,
(surjective hull)\,.
\end{itemize}
If in addition $(\id{A}{}{}{}{}, \idn{A}{}{}{}{})$ is complete, an operator 
$T \in \id{A}{}{}{}{}(E,F)$ is said to be \textit{$\id{A}{}{}{}{}$-approximable}, 
if there exists a sequence of finite rank operators $(L_n)_{n\in\N} \subseteq 
\id{F}{}{}{}{}(E,F)$, such that $\lim\limits_{n \to \infty}\idn{A}{}{}{}{}
(T - L_n) = 0$. The class of these operators is the so-called 
\textit{approximative kernel} $\id{A}{}{\bf{(a)}}{}{}$. Restricting the 
quasi-norm of $\id{A}{}{}{}{}$ to $\id{A}{}{\bf{(a)}}{}{}$, 
$\idn{A}{}{\bf{(a)}}{}{}(T) : = \idn{A}{}{}{}{}(T)$, then also 
$(\id{A}{}{\bf{(a)}}{}{}, \idn{A}{}{\bf{(a)}}{}{})$ is a quasi-Banach operator 
ideal. If $E$ and $F$ are arbitrary Banach spaces, then it is obvious that 
$\id{A}{}{\bf{(a)}}{}{}(E,F) \stackrel{1}{=} 
\overline{\id{F}{}{\bf{}}{}{}(E,F)}^{\id{A}{}{}{}{}(E;F)}$. Recall the 
nonobvious fact that also $(\id{A}{}{\tup{min}}{}{}, \idn{A}{}{\tup{min}}{}{})$ 
is a $p$-normed operator ideal if $(\id{A}{}{}{}{}, \idn{A}{}{}{}{})$ is 
$p$-normed, $0 < p \leq 1$ (cf. \cite[Corollary 3.1]{Oe2003}). Within the scope 
of this work, also the following two Banach operator ideal constructions will 
play a fundamental role:
\begin{enumerate}
\item  Let $\id{A}{}{\conj}{}{}(E,F)$ be the set of all $T\in \mathfrak{L}(E,F)$
which satisfy 
\[
\idn{A}{}{\conj}{}{}(T):=\sup \{\vert\tup{tr}(TL)\vert : L\in \mathfrak{F}(F,E),
\idn{A}{}{}{}{}(L)\leq 1\} < \infty . 
\]
Then a Banach operator ideal $(\id{A}{}{\conj}{}{}, \idn{A}{}{\conj}{}{})$ is 
obtained, known as \textit{conjugate} of $(\id{A}{}{}{}{}, 
\idn{A}{}{}{}{})$.
\item  Let $\id{A}{}{\adj}{}{}(E,F)$ be the set of all $T\in 
\id{L}{}{}{}{}(E,F)$ which satisfy 
\begin{align*}
\idn{A}{}{\adj}{}{}(T)&:=\sup\{\vert \tup{tr}(T J_M^E S Q_K^F)\vert : M \in 
\mbox{FIN}(E), K\in \mbox{COFIN}(F),\mathbf{A}(S)\leq 1\}\\
&\,\,= \sup\{\idn{A}{}{\conj}{}{}(Q_K^F T J_M^E) : M \in 
\mbox{FIN}(E), K\in \mbox{COFIN}(F)\} < \infty . 
\end{align*}
Then a Banach operator ideal $(\id{A}{}{\adj}{}{},\id{A}{}{\adj}{}{})$ is 
obtained. It is called the \textit{adjoint} of $(\id{A}{}{\adj}{}{}, 
\idn{A}{}{\adj}{}{})$.
\end{enumerate}
Their explicit construction, many examples and structural properties can be 
found in the standard works on operator ideals, including 
\cite{DF1993, DJP2001, J1981, P1981, P1987} and \cite{JO1982, P2014}, where 
\cite{JO1982} focuses on structural properties of the conjugate of 
$(\id{A}{}{}{}{}, \idn{A}{}{}{}{})$.

Let $(\mathfrak{A},\mathbf{A})$ be a $p$-normed operator ideal and 
$(\mathfrak{B},\mathbf{B})$ be a $q$-normed operator ideal $(0 < p, q \leq 1)$. 
At this point, let us emphasize that - throughout our manuscript - the 
``nonnormable'' cases $p < 1$ and $q < 1$ are explicitly included, unless 
specifically stated otherwise. We make use of the following symbolic shortcuts:

We write $\mathfrak{A}\subseteq \mathfrak{B}$ if, regardless of the choice of 
the Banach spaces $E$ and $F$, we have $\mathfrak{A}(E,F)\subseteq \mathfrak{B}
(E,F)$. If $E_0$ is a fixed Banach space, we write $\mathfrak{A}(E_0,\cdot ) 
\subseteq \mathfrak{B}(E_0,\cdot )$ if, regardless of the choice of the Banach 
space $G$ we have $\mathfrak{A}(E_0,G)\subseteq \mathfrak{B}(E_0,G)$. Similarly, 
$\mathfrak{A}(\cdot ,E_0)\subseteq \mathfrak{B}(\cdot, E_0)$ is defined. The 
metric inclusion $(\mathfrak{A},\mathbf{A})\subseteq (\mathfrak{B},\mathbf{B})$ 
is shortened by 
\[
\mathfrak{A}\stackrel{1}{\subseteq }\mathfrak{B}\,.
\]
Thus, $\mathfrak{A}\stackrel{1}{\subseteq }\mathfrak{B}$ if and only if 
$\mathfrak{A} \subseteq \mathfrak{B}$ and $\idn{B}{}{}{}{}(T) \leq 
\idn{A}{}{}{}{}(T)$ for all $T \in \id{A}{}{}{}{}$. If $\idn{B}{}{}{}{}(L) \leq 
\idn{A}{}{}{}{}(L)$ for all finite rank operators $L \in \mathfrak{F}$, we use 
the shortcut 
\[
\mathfrak{A} \stackrel{\mathfrak{F}}{\subseteq } \mathfrak{B}\,.
\]
Similarly, 
\[
\mathfrak{A} \stackrel{\mathfrak{E}}{\subseteq } \mathfrak{B}
\] 
indicates that $\idn{B}{}{}{}{}(S) \leq \idn{A}{}{}{}{}(S)$ for all elementary 
operators $S \in \mathfrak{E}$. If $p = q$, we write 
\[
\mathfrak{A}\stackrel{1}{=}\mathfrak{B}
\]
to indicate that the isometric equality $(\mathfrak{A},\mathbf{A}) = 
(\mathfrak{B},\mathbf{B})$ holds. If $s_1, s_2 \in \{\tup{min}, \tup{d}, \tup{reg}, 
\tup{inj}, \tup{sur}, \vdash, \dashv, \conj, \adj, \tup{max}\}$, and if 
$s_1 \not= s_2$, then in general $\id{A}{}{s_1\,s_2}{}{} : = 
(\id{A}{}{s_1}{}{})^{s_2}$ cannot be compared with $\id{A}{}{s_2\,s_1}{}{} = 
(\id{A}{}{s_2}{}{})^{s_1}$. Nevertheless, to facilitate readability, we often 
ignore brackets if the arrangement of $s_1$ and $s_2$ is fully understood. 
 
If $E$ is an arbitrary Banach space, it is trivial that $(\id{A}{}{}{}{}(E), 
\idn{A}{}{}{}{}(\cdot))$ is a Banach algebra for any given Banach operator 
ideal $(\id{A}{}{}{}{}, \idn{A}{}{}{}{})$. In the complex Hilbert space case, we 
can say a bit more. An application of the (left) polar decomposition of operators 
in $\id{L}{}{}{}{}(H, K)$ and their adjoints namely, where $H$ and $K$ are 
Hilbert spaces over $\C$ (cf. \cite[Theorem 3.1.5]{S2016}) implies that 
\textit{each} quasi-Banach operator ideal $(\id{A}{}{}{}{}, \idn{A}{}{}{}{})$, 
induces the Banach $\adj$-algebra $\id{A}{}{}{}{}(H)$. More precisely, we have:
\begin{proposition}
\label{prop:Banach_op_ideal_components_on_H_are_Banach_ast_algebras}
Let $(\id{A}{}{}{}{}, \idn{A}{}{}{}{})$ be a quasi-normed operator ideal and 
$H, K$ be Hilbert spaces over $\C$. Then $T \in \id{A}{}{}{}{}(H, K)$, if and 
only if $\vert T\vert \in \id{A}{}{}{}{}(H)$, if and only if $T^\adj \in 
\id{A}{}{}{}{}(K, H)$, and
\[
\idn{A}{}{}{}{}(T) = \idn{A}{}{}{}{}(\vert T\vert) = \idn{A}{}{}{}{}(T^\adj)\,
\text{ for all }\, T \in \id{A}{}{}{}{}(H, K)\,.
\]
In particular, $\id{A}{}{}{}{}(H)$ is a $\adj$-subalgebra of 
$\id{L}{}{}{}{}(H)$ as well as a $\adj$-algebra. If 
$(\id{A}{}{}{}{}, \idn{A}{}{}{}{})$ is a normed $($resp. Banach$)$ operator 
ideal, then $(\id{A}{}{}{}{}(H), \idn{A}{}{}{}{}(\cdot))$ is a normed 
$($resp. Banach$)$ $\adj$-algebra.
\end{proposition}
\begin{proof}
Nothing is to show if $T = 0$. So, let $T \in \id{L}{}{}{}{}(H, K)\!\setminus\!\{0\}$, 
and $T = W \vert T\vert$ be the (left) polar 
decomposition of $T$, with uniquely given partial isometry 
$W \in \id{L}{}{}{}{}(H, K)\!\setminus\!\{0\}$. In particular, it follows that 
$\Vert W^\adj\Vert = \Vert W\Vert = 1$. Thus, $W^\adj T = \vert T\vert$ (due to 
\cite[Theorem 3.1.5-(c)]{S2016}). Hence, $T \in \id{A}{}{}{}{}(H, K)$ if and 
only if $\vert T\vert \in \id{A}{}{}{}{}(H)$, and in either case, it follows 
that
\[
\idn{A}{}{}{}{}(\vert T\vert) \leq \Vert W^\adj\Vert\,\idn{A}{}{}{}{}(T) = 
\idn{A}{}{}{}{}(T) = \idn{A}{}{}{}{}(W\,\vert T\vert) \leq \Vert W\Vert\,
\idn{A}{}{}{}{}(\vert T\vert) = \idn{A}{}{}{}{}(\vert T\vert)\,.
\]
In particular, $T^\adj \in \id{A}{}{}{}{}(K, H)$ if and only if 
$\vert T^\adj \vert \in \id{A}{}{}{}{}(K)$, and  
\[
\idn{A}{}{}{}{}(\vert T^\adj \vert) = \idn{A}{}{}{}{}(T^\adj)\,.  
\]
To conclude the proof, let now $T^\adj = U \vert T^\adj\vert$ be the (left) 
polar decomposition of the adjoint $T^\adj$, with uniquely given partial 
isometry $U \in \id{L}{}{}{}{}(K, H)\!\setminus\!\{0\}$. Then  
\[
T = T^{\adj\adj} = \vert T^\adj\vert^\adj U^\adj = \vert T^\adj\vert U^\adj\,
\text{ and }\,(TU)^\adj = U^\adj T^\adj = \vert T^\adj\vert = 
\vert T^\adj\vert^\adj = TU\,.
\]
The claim now follows at once (since also $\Vert U^\adj\Vert = \Vert U\Vert = 
1$).   
\end{proof}
\begin{remark}\label{rem:op_ideals_on_H_in_gen_not_C_star}
There exist (maximal) Banach operator ideals $(\id{A}{0}{}{}{}\,, 
\idn{A}{0}{}{}{}\,)$, such that $\id{A}{0}{}{}{}\,(H)$ is not a 
$\tup{C}^\adj$-algebra (cf. \sref{Remark}{rem:HS_ops_on_H_are_not_C_star} and 
\eqref{eq:HS_equals_P2}).
\end{remark}
A certainly nontrivial implication of 
\sref{Proposition}{prop:Banach_op_ideal_components_on_H_are_Banach_ast_algebras}, 
revealing an interesting link between normed operator ideals and 
$\tup{C}^\adj$-algebra theory, is the subject of 
\sref{Theorem}{thm:enveloping_C_star_alg_of_normed_op_ideal} below. If namely 
$\mathcal{A}$ is an arbitrary normed (not necessarily complete) $\adj$-algebra, 
a connection of \cite[\textsection\,32]{T2025} and \cite[Chapter 6.1, p. 175]
{M1990} discloses a lucid and fully self-contained construction of the so called 
\textit{enveloping $C^\adj$-algebra $\tup{C}^\adj(\mathcal{A})$} (where 
Murphy's assumed existence of a $\tup{C}^\adj$-seminorm $p$ on $\mathcal{A}$ is 
satisfied by the well-defined Gel'fand-Na\u{\i}mark seminorm 
(\cite[Definition 32.6]{T2025})). Consequently, by taking 
\cite[Theorem 33.9]{T2025} (a generalisation of the universal representation of 
Gel'fand and Na\u{\i}mark in the $\tup{C}^\adj$-algebra case) into account, we obtain    
\begin{theorem}\label{thm:enveloping_C_star_alg_of_normed_op_ideal}
Let $(\id{A}{}{}{}{}, \idn{A}{}{}{}{})$ be a normed operator ideal and $H$ be 
a Hilbert space over $\C$. Then the universal representation 
\[
\pi : (\id{A}{}{}{}{}(H), \idn{A}{}{}{}{}(\cdot)) \longrightarrow 
\id{L}{}{}{}{}(H^\pi) 
\]
of the normed $\adj$-algebra $(\id{A}{}{}{}{}(H), \idn{A}{}{}{}{}(\cdot))$ on 
the complex Hilbert space $H^\pi$ exists. In particular, 
$\Vert\pi(\cdot)\Vert : \id{A}{}{}{}{}(H) \longrightarrow [0, \infty), 
R \mapsto \Vert\pi(R)\Vert$ is a well-defined $C^\adj$-seminorm on the 
$\adj$-algebra $\id{A}{}{}{}{}(H)$. The completion of the quotient $\adj$-algebra 
$\lb\id{A}{}{}{}{}(H)/\textup{ker}(\pi), \mid\mid\mid \cdot \mid\mid\mid\rb$ 
in the $C^\adj$-norm 
\[
\id{A}{}{}{}{}(H)/\textup{ker}(\pi) \ni R + \textup{ker}(\pi) \mapsto 
\mid\mid\mid\!R + \textup{ker}(\pi)\!\mid\mid\mid : = \Vert\pi(R)\Vert
\] 
is a $C^\adj$-algebra, that is given by $\tup{C}^\adj(\id{A}{}{}{}{}(H))$, 
the enveloping $C^\adj$-algebra of $(\id{A}{}{}{}{}(H), \idn{A}{}{}{}{}
(\cdot))$$)$. 
In particular, $\tup{C}^\adj(\id{A}{}{}{}{}(H))$ is isometrically 
$\adj$-isomorphic to a norm-closed $\adj$-subalgebra of $\id{L}{}{}{}{}
(H^{\pi})$.  
\end{theorem}
Next, we sketch how Banach operator ideals, or - \textit{equivalently} - 
completions of normed tensor products of Banach spaces in the sense of 
A. Grothendieck actually also play a significant role in quantum physics; 
particularly in algebraic quantum field theory (cf. \cite{BDF1987, BS1991, 
F2016, FR2020, RS2010, S1994, S1996, W1987}) and in the investigation of 
compatibility of dichotomic measurements in general probabilistic theories 
(cf. \cite{BJN2022}). Given that point of view, 
\sref{Theorem}{thm:structure_of_Hilb_space_tp} unfolds as a particular case.

In order to better categorise the role of Banach operator ideals and tensor 
products of Banach spaces without inner product structure in quantum physics 
(including the class of $p$-nuclear, $p$-integral, approximable and compact 
operators), we briefly examine the history of an essential part of the classical 
theory of Banach spaces in functional analysis, which actually emerged from 
Grothendieck's seminal article ``R\'{e}sum\'{e} de la Th\'{e}orie M\'{e}trique 
des Produits Tensoriels Topologiques'' (\cite{G1953}). In this regard, we would 
like to draw particular attention to \cite{P2015}, which provides a detailed 
overview of this remarkable development.

In the late thirties, tensor products entered the area of functional analysis, 
due to the works of F. J. Murray, J. von Neumann(!) and R. Schatten. However, it 
was Grothendieck who revealed the structural richness of tensor products of 
Banach spaces and who used their various norms to construct related 
classes of bounded linear operators. In this context, he actually established 
the origin of the ``local'' theory of Banach spaces, i.e., the study of the 
structure of Banach spaces in terms of their finite-dimensional subspaces. 
\cite{DFS2008} gives a very readable and comprehensive account of the tensor 
product theory, developed in \cite{G1953} while maintaining the symbolic 
language of Grothendieck. Actually, \cite{G1953} appeared in 1956. 

It was not until the end of the sixties when the scientific community gave 
\cite{G1953} particular recognition. The interest in Grothendieck's work revived 
when J. Lindenstrauss and A. Pe{\l}czy\'{n}ski ``translated'' its main results 
in the more traditional language of operators and matrices (cf. \cite{DFS2008, 
LP1968}). They presented important applications to the theory of absolutely 
$p$-summing operators and reformulated results, which were written in terms of 
tensor products by Grothendieck, into properties of linear operators and 
operator ideals. 

Almost at the same time, a general theory of operator ideals on the class of 
Banach spaces was developed by Pietsch and his academic school in Jena, yet 
without the use and the abstract language of Grothendieck's tensor norms. Due to 
Pietsch's seminal book ``Operator Ideals'' (\cite{P1980}), that theory became 
a central theme in Banach space theory. A comprehensive overview of Pietsch's 
theory and application of operator ideals is given in \cite{DJP2001}.

In 1993, A. Defant and K. Floret published their pathbreaking and comprehensive 
monograph ``Tensor Norms and Operator Ideals'' (\cite{DF1993}). Here, 
deep interconnections between operator ideals, the ``local'' theory of Banach 
spaces and tensor norms are revealed with a high level of attention to detail. 
They made very clear that tensor products and operator ideals are closely 
connected and showed in detail how to transform tensor products to operator 
ideals and conversely, revealing that normed tensor products of Banach spaces 
(in the sense of Grothendieck) and Banach operator ideals (in the sense of Pietsch) 
are ``two sides of the same coin''! At this point, we 
should recall that in general many norms can be constructed on the algebraic 
tensor product $E \otimes F$ if $E$ and $F$ are Banach spaces which are not both 
Hilbert spaces (including $\tup{C}^\adj$-algebras and operator spaces); 
unlike the Hilbert space case (cf. 
\sref{Corollary}{cor:char_of_tensor_product_of_HS_as_Hilbert_Schmidt}). The 
underlying functional analytic details can be found in \cite{DF1993, DJT1995, 
J1981, P1980, R2002}.

A straightforward application of ``good old-school operator theory'' 
shows us that the \textit{smallest Banach ideal} $(\id{N}{}{}{}{}, 
\idn{N}{}{}{}{})$ of nuclear operators plays a predominant role in quantum 
physics. Owing to their importance, we have to recall explicitly the 
construction of that class of bounded linear operators. Let $1 \leq p < \infty, 
\frac{1}{p} + \frac{1}{p^\ast} = 1$, $E, F$ be 
arbitrary Banach spaces and $T \in \id{L}{}{}{}{}(E, F)$. $T$ is $p$-nuclear if 
there are sequences $(a_n)_{n \in \N} \in l_p^{\small{strong}}(E^\prime)$ 
and $(y_n)_{n \in \N} \in l_{p^\ast}^{\small{weak}}(F)$, such that
\[
Tx = \sum\limits_{n=1}^\infty \langle x, a_n\rangle y_n\,\text{ for all } x \in E,
\] 
and the series converges in $\id{L}{}{}{}{}(E, F)$. $(\id{N}{p}{}{}{}\,, 
\idn{N}{p}{}{}{}\,) \stackrel{1}{=} (\id{N}{p}{\textup{min}}{}{}\,, 
\idn{N}{p}{\textup{min}}{}{}\,)$ is a minimal Banach ideal.  
\[
\idn{N}{p}{}{}{}\,(T) : = \inf\big(\Vert (a_n)\Vert_p^{\small{strong}}
\,\Vert (y_n)\Vert_{p^\ast}^{\small{weak}}\big)\hspace{1cm} (1  < p < \infty)
\]
and
\[
\idn{N}{1}{}{}{}\,(T) : = \inf\big(\sum\limits_{n=1}^\infty 
\Vert a_n\Vert\,\Vert y_n\Vert\big)
\] 
are norms, where the infimum is taken over all possible factorisations, 
respectively, $\Vert (a_n)\Vert_p^{\small{strong}}: = 
\big(\sum\limits_{n=1}^\infty \Vert a_n\Vert^p\big)^{1/p}$ 
and $\Vert (y_n)\Vert_{p^\ast}^{\small{strong}} : = 
\sup\limits_{b \in B_{F^\prime}}\big(\sum\limits_{n=1}^\infty 
\vert \langle y_n, b\rangle\vert^{p^\ast}\big)^{1/{p^\ast}}$, given that 
$1 < p < \infty$ (cf., e.g., \cite[Proposition 5.23]{DJT1995}).

Let $(\id{S}{1}{}{}{}\,, \sigma_1)$ be the complete 
\textit{quasi-normed} operator ideal of $1$-approximable operators and 
$(\id{I}{}{}{}{}, \idn{I}{}{}{}{}) \stackrel{1}{=} 
(\id{N}{}{\textup{max}}{}{}\,, \idn{N}{}{\textup{max}}{}{}\,)$ 
be the maximal Banach ideal of integral operators (cf. \cite{DF1993, DJT1995, 
J1981, P1980}). Regarding nontrivial properties of these operator ideals which 
prove to be very useful for applications in quantum mechanics, we have to refer 
to the following three known key results; namely 
\sref{Theorem}{thm:a_corollary_of_Grothendieck}, 
\sref{Theorem}{thm:char_of_trace_class_ops} (if $H = K$) and 
\sref{Theorem}{thm:trace_duality_in_the_Hilbert_space_case} (cf. 
\cite[Proposition 18.8]{C2000}, \cite[Corollary 22.2.1]{DF1993}, 
\cite[Theorem 5.30 and Theorem 5.31]{DJT1995}, \cite[Proposition 20.2.5]{J1981}, 
\cite[Propositions 16.18, 16.24, 16.26, 18.18, 
Corollary 16.4 and Lemma 16.21]{MV1997}, 
\cite[Theorem 2.4.13 and Theorem 2.4.15]{M1990}, \cite[Theorem 15.5.3]{P1980}, 
and \eqref{eq:HS_equals_P2} below). Recall also 
\sref{Proposition}{prop:Banach_op_ideal_components_on_H_are_Banach_ast_algebras}. 
The \textit{injective tensor norm} $\varepsilon$ and the 
\textit{projective tensor norm} $\pi$ play a key role in this. So, let us 
quickly recall their explicit construction:
\begin{align*}
\varepsilon(z; E, F) :&= \sup\big\{\vert (a\,\otimes\,b)z \vert : 
(a, b) \in B_{E^\prime} \times B_{F^\prime}\big\}\\
&= \sup\big\{\vert\sum\limits_{i=1}^n \langle x_i, a\rangle
\langle y_i, b\rangle\vert : (a, b) \in B_{E^\prime} \times B_{F^\prime}\big\}
\end{align*}
and
\[
\pi(z; E, F) : = \inf\big\{\sum\limits_{i=1}^n \Vert x_i \Vert\,
\Vert y_i \Vert : n \in \N, z = \sum\limits_{i=1}^n 
x_i \otimes y_i\big\},
\]
where $z = \sum\limits_{i=1}^n x_i \otimes y_i \in 
E \otimes_{\tiny{\textup{alg}}} F$.

First of all, an inclusion of \sref{Theorem}{thm:Riesz_linear_version} 
implies 
\begin{theorem}\label{thm:a_corollary_of_Grothendieck}
Let $H$ and $K$ be Hilbert spaces over $\C$. Then the following isometric 
identities are satisfied:
\[
\id{I}{}{}{}{}(H,K) \stackrel{1}{=} \id{N}{}{}{}{}(H,K) \stackrel{1}{=} 
\id{S}{1}{}{}{}\,(H,K) \cong \overline{H} \widetilde{\otimes}_{\pi} K
\]
and
\[
\id{K}{}{}{}{}(H,K) \stackrel{1}{=} \id{\overline{F}}{}{}{}{}(H,K) 
\stackrel{1}{=} \overline{H} \widetilde{\otimes}_{\varepsilon} K\,.
\]
In particular,
\begin{align}\label{eq:H_otimes_pi_K}
\id{I}{}{}{}{}(\overline{H},K) \stackrel{1}{=} \id{N}{}{}{}{}(\overline{H},K) 
\stackrel{1}{=} \id{S}{1}{}{}{}\,(\overline{H},K) \cong 
H \widetilde{\otimes}_{\pi} K
\end{align}
and
\begin{align}\label{eq:H_otimes_vareps_K}
\id{K}{}{}{}{}(\overline{H},K) \stackrel{1}{=} \id{\overline{F}}{}{}{}{}
(\overline{H},K) \cong H \widetilde{\otimes}_{\varepsilon} K\,.
\end{align} 
\end{theorem}
\begin{theorem}\label{thm:char_of_trace_class_ops}
Let $H$ and $K$ be Hilbert spaces over $\C$, $T \in \id{L}{}{}{}{}(H, K)$ and 
$\vert T \vert : = (T^\adj T)^{1/2} \in \id{L}{}{}{}{}(H)^+$. Then the 
following statements are equivalent:
\begin{enumerate}
\item $T \in \id{N}{}{}{}{}\,(H, K)$.
\item $T^\adj \in \id{N}{}{}{}{}\,(K, H)$.
\item $\vert T \vert \in \id{N}{}{}{}{}\,(H)$.
\item $\vert T \vert^{1/2} \in \id{S}{2}{}{}{}\,(H)$.
\item $T$ is the composition of two Hilbert-Schmidt operators.
\item $\vert T\vert$ is the composition of two Hilbert-Schmidt operators.
\end{enumerate}
$\id{N}{}{}{}{}(H) \stackrel{1}{=} \id{S}{1}{}{}{}\,(H)$ 
coincides with the class of all trace-class operators. Moreover, if 
$\{e_i : i \in I\}$ is an arbitrarily chosen orthonormal basis in $H$, then
\begin{align}\label{eq:nuclear_norm_vs_HS_norm}
\idn{N}{}{}{}{}(S^\adj) = \idn{N}{}{}{}{}(S) = \idn{N}{}{}{}{}(\vert S\vert) = 
\sum\limits_{i \in I} \langle \vert S\vert e_i, e_i \rangle_K = 
\tup{tr}(\vert S\vert) = (\idn{S}{2}{}{}{}\,(\vert S\vert^{1/2}))^2  
\end{align}
for all $S \in \id{N}{}{}{}{}(H, K)$. 
Moreover, $T \in \id{S}{2}{}{}{}\,(H, K)$, if and only if $T^\adj T \in 
\id{N}{}{}{}{}(H)$, if and only if $\vert T\vert^2 \in \id{N}{}{}{}{}(H)$, 
if and only if $\vert T\vert \in \id{S}{2}{}{}{}\,(H)$, in which case
\begin{align*}
\idn{N}{}{}{}{}(\vert T\vert^2) = \idn{N}{}{}{}{}(T^\adj T) = 
(\idn{S}{2}{}{}{}\,(T))^2 = (\idn{S}{2}{}{}{}\,(\vert T\vert))^2\,.
\end{align*}
\end{theorem}
\begin{proof}
If $H = K$, the result is well-known (cf., e.g., \cite[Proposition 18.8]{C2000}). 
The general result now follows from 
\sref{Proposition}{prop:Banach_op_ideal_components_on_H_are_Banach_ast_algebras} 
and a straightforward application of the polar decomposition of $T$, where the 
latter implies the equivalence of the conditions (v) and (vi).     
\end{proof}
If we apply the spectral representation theorem for the particular case of 
compact normal operators on a Hilbert space (cf. 
\sref{Corollary}{cor:Schmidt_decomposition_revisited} and the preceeding 
remark), together with \eqref{eq:dyad_conj_HS_cases}, we obtain another 
nonnegligible 
\begin{corollary}\label{cor:rep_of_nuclear_density_operators}
Let $H$ be a Hilbert space and $T \in \id{K}{}{}{}{}(H)$. Then the following 
statements are equivalent:
\begin{enumerate}
\item $T \in \id{N}{}{}{}{}(H), T \geq 0$ and $\idn{N}{}{}{}{}(T) = 1$.
\item $T \in \id{N}{}{}{}{}(H), T \geq 0$ and $\tup{tr}(T) = 1$.
\item There exists a countable $($possibly finite$)$ orthonormal sequence $(e_n)_
{n \in \A}$ of eigenvectors of $T$ and a sequence $(p_n)_{n \in \A} \in l_1$, 
consisting of eigenvalues of $T$, each of them repeated a finite number of 
times, corresponding to the dimension of their respective eigenspace, such that
$0 \leq p_n \leq 1$ for all $n \in \A$, $\Vert(p_n)\Vert_{l_1} = 
\sum\limits_{n \in \A} p_n = 1$, and  
\[
Tx = \sum\limits_{n \in \A} p_n\,\langle x, e_n\rangle_H\,e_n = 
\sum\limits_{n \in \A} p_n\,\lb\tp{J_H e_n}{-0.45ex}{\scbox{0.70}
{$\overline{H}$}}{-0.45ex}{\scbox{0.70}{$H$}}{e_n}\rb\!(x)\,\text{ for all }\, x 
\in H\,,
\] 
In particular,
\[
U_\otimes^{-1}(T) = \sum\limits_{n \in \A} p_n\, e_n \otimes e_n\,,
\]
where $U_\otimes : \overline{H} \otimes H \stackrel{\cong}{\longrightarrow} 
\id{S}{2}{}{}{}\,(H)$ denotes the canonical unitary operator, introduced in 
\sref{Corollary}{cor:char_of_tensor_product_of_HS_as_Hilbert_Schmidt}.
\end{enumerate} 
In each case, 
\[
\sigma_{\overline{H}, H}(U_\otimes^{-1}(T)) = 
\idn{S}{2}{}{}{}\,(T) \leq \idn{N}{}{}{}{}(T) = \tup{tr}(T) = 1\,.
\]
\end{corollary}
\begin{proof}
(i) $\Rightarrow$ (ii): Since $T \geq 0$, it follows that $T = \vert T\vert$, 
whence $\tup{tr}(T) = \tup{tr}(\vert T\vert) = \idn{N}{}{}{}{}(T) = 1$ (cf. 
\cite[Definition 18.3 and Definition 18.10]{C2000}).\\[0.2em]
\noindent (ii) $\Rightarrow$ (iii): This nontrivial implication follows from 
the spectral representation theorem for compact normal operators on a Hilbert 
space, applied to the nuclear positive (and hence compact normal) operator $T = 
\vert T\vert \in \id{N}{}{}{}{}(H) \subseteq \id{K}{}{}{}{}(H)$. By definition, 
each $p_n$ (the $n$-th eigenvalue of $T = \vert T\vert = (T^\adj T)^{1/2}$) 
coincides with the $n$-th singular value of $T$, for all $n \in \A$ (cf. 
\cite[proof of Theorem VI.3.2 and Satz VI.5.5]{W2011}).\\[0.2em]
\noindent (iii) $\Rightarrow$ (i): This follows from the proof of 
\cite[Satz VI.5.5]{W2011}.
\end{proof}  
Trace duality leads to the following important particular case of a general 
representation result for nonnormed operator ideals (cf. \cite{Oe2026}), whose 
proof therefore does not include any applications of tensor norms, such as in 
the maximal Banach operator ideal case, where the latter is fully reflected in 
\cite[Theorem 17.5 and Proposition 22.6]{DF1993}. Due to the reflexivity of a 
Hilbert space $K$, the inverse of the canonical metric injection $j_K : K 
\stackrel{1}{\hookrightarrow} K^{\prime\prime}$, described in 
\sref{Remark}{rem:the_operator_Lambda_z}, is well-defined, and we arrive at 
\begin{theorem}\label{thm:trace_duality_in_the_Hilbert_space_case}
Let $(\id{A}{}{}{}{}, \idn{A}{}{}{}{})$ be a $p$-normed operator ideal 
$(0 < p \leq 1)$. Let $H$ and $K$ be  Hilbert spaces. Then
\[
\Theta_{\id{A}{}{\tup{min}}{}{}} : (\id{A}{}{\tup{min}}{}{}(K, H))^\prime 
\stackrel{\cong}{\longrightarrow} \id{A}{}{\adj}{}{}(H, K), \xi \mapsto 
j_K^{-1}\,T_\xi 
\]
is an isometric isomorphism, where the operator $T_\xi \in \id{A}{}{\adj}{}{}
(H, K^{\prime\prime})$ is defined as 
\begin{align}\label{eq:construction_of_T_xi}
\langle b, T_\xi x\rangle : = \langle \langle\cdot, b\rangle\,x, \xi\rangle
\hspace{2em} ((x,b) \in H \times K^\prime)\,. 
\end{align}
The inverse is given by
\begin{align}\label{eq:the_inverse}
\Theta_{\id{A}{}{\tup{min}}{}{}}^{-1}(T) = \textup{tr}(\cdot\,T)\hspace{2em}(T \in 
\id{A}{}{\adj}{}{}(H, K))\,.
\end{align}
In particular,
\[
\langle A, \xi\rangle = \tup{tr}(A\,\Theta_{\id{A}{}{\tup{min}}{}{}}(\xi))
\]
and
\begin{align}\label{eq:splitting_role_of_Theta_Amin}
\langle \tp{J_K y}{-0.45ex}{\scbox{0.70}{$\overline{K}$}}{-0.50ex}
{\scbox{0.70}{$H$}}{x}, \xi\rangle = 
\langle \Theta_{\id{A}{}{\tup{min}}{}{}}(\xi)x, y \rangle_K\,. 
\end{align}
for all $(A, \xi) \in (\id{A}{}{\tup{min}}{}{}(K, H), 
((\id{A}{}{\tup{min}}{}{}(K, H))^\prime)$ and $(x, y) \in H \times K$. Moreover, 
if $H = K$ is finite-dimensional, then $\Theta_{\id{A}{}{\tup{min}}{}{}}^{-1}
(Id_H) = \textup{tr}$.
\end{theorem}
\begin{remark}\label{rem:positivity_criterium}
Let $(\id{A}{}{}{}{}, \idn{A}{}{}{}{})$ be a $p$-normed operator ideal 
$(0 < p \leq 1)$. Let $H$ be  a Hilbert space. Let $\xi \in 
(\id{A}{}{\tup{min}}{}{}(H))^\prime$. If $\xi \geq 0$, then $\Theta_
{\id{A}{}{\tup{min}}{}{}}(\xi) \in \id{A}{}{\adj}{}{}(H) \cap 
\id{L}{}{}{}{}(H)^+$ is a positive operator. In order to recognise 
this at once, we just have to apply \eqref{eq:splitting_role_of_Theta_Amin} to 
$K = H$ and $y = x$, i.e., to the positive (rank one) operator 
$\tp{J_H x}{-0.45ex}{\scbox{0.70}{$\overline{H}$}}{-0.5ex}{\scbox{0.70}
{$H$}}{x} = \langle\cdot, x\rangle_H\,x \in \id{A}{}{\tup{min}}{}{}(H) \cap 
\id{L}{}{}{}{}(H)^+$. We do not know whether for any $p$-normed operator ideal 
$(\id{A}{}{}{}{}, \idn{A}{}{}{}{}) \not\!\!{\stackrel{1}{=}}\,(\id{L}{}{}{}{}, 
\Vert \cdot \Vert)$ ($0 < p \leq 1$) the converse is also satisfied. If namely 
$0 \leq \Theta_{\id{A}{}{\tup{min}}{}{}}(\xi) = : D_\xi = D_\xi^{1/2} D_\xi^{1/2} 
\in \id{A}{}{\adj}{}{}(H)$ and if $0 \leq A \in \id{A}{}{\tup{min}}{}{}(H)$, 
then $D_\xi^{1/2}(D_\xi^{1/2}A) = D_\xi A \in (\id{A}{}{\adj}{}{} \circ 
\id{A}{}{\tup{min}}{}{})(H) \stackrel{1}{\subseteq} \id{N}{}{}{}{}(H)$. However, 
we do not know whether also the - positive - operator $(D_\xi^{1/2} A)D_\xi^{1/2}$ 
is nuclear. If it were, then we could in fact apply the trace equalities 
\eqref{eq:trace_trick} below (due to \cite[Lemma 2.1]{LNRR1981})! A positive 
answer holds in the case $(\id{A}{}{}{}{}, \idn{A}{}{}{}{}) : = 
(\id{L}{}{}{}{}, \Vert \cdot \Vert)$ (cf. \eqref{eq:trace_trick} below).
\end{remark}    
Consequently, even if $H$ is infinite-dimensional 
and nonseparable, \sref{Corollary}{cor:rep_of_nuclear_density_operators} and 
\sref{Theorem}{thm:trace_duality_in_the_Hilbert_space_case} in particular imply 
that states on the $\tup{C}^\adj$-algebra of \textit{compact} operators on $H$, 
$(\id{K}{}{}{}{}(H), \Vert\cdot\Vert_{\id{L}{}{}{}{}(H)})$ - which is 
not a von Neumann algebra if $H$ is infinite-dimensional (since $Id_H$ is 
compact if and only if $\tup{dim}(H) < \infty$) - are in a one-to-one 
correspondence with positive \textit{nuclear} operators of trace one (cf. 
also \cite[Chapter 2.2]{H2003}). More precisely, if $\xi$ is a positive linear 
functional on the Banach space $\id{K}{}{}{}{}(H)$ of norm $1$ (i.e., if 
$\xi(S) \geq 0$ for all $S \in \id{K}{}{}{}{}(H)^+$ and $\Vert \xi \Vert = 
1$ - cf., e.g., \cite[Chapter 2.1, p. 19-20]{H2003} and 
\cite[Definition on p. 89]{M1990}), then there is a nuclear, positive operator 
$D_\xi \in \id{N}{}{}{}{}(H) \cap \id{L}{}{}{}{}(H)^+$, such that 
$\tup{tr}(D_\xi) = 1 = \idn{N}{}{}{}{}(D_\xi)$ and
\begin{align}\label{eq:the_compact_case}
\xi(A) = \tup{tr}(A D_\xi) = \tup{tr}(D_\xi A) \text{ for all } A \in 
\id{K}{}{}{}{}(H)\,.
\end{align}
In fact, the required operator $D_\xi$ is given by $D_\xi : = 
\Theta_{\id{L}{}{}{\tup{min}}{}}(\xi) = \Theta_{\id{\overline{F}}{}{}{}{}}
(\xi) = j_H^{-1}\,T_\xi$ (due to \eqref{eq:the_inverse}, applied to 
$(\id{A}{}{}{}{}, \idn{A}{}{}{}{}) : = (\id{L}{}{}{}{}, \Vert \cdot\Vert)$ 
here)! That construction is built on a few nontrivial results from operator 
ideal theory; namely that - in the Hilbert space case - the operator ideal 
component $\id{N}{}{}{}{}(H)$ coincides isometrically with the component of all 
integral operators $\id{I}{}{}{}{}(H) \stackrel{1}{=} \id{A}{}{\adj}{}{}(H)$ 
(cf., e.g., \cite[Theorem 5.30]{DJT1995}) 
and that the operator ideal component $\id{K}{}{}{}{}(H)$ coincides 
isometrically with the component of all $\Vert \cdot\Vert_
{\id{L}{}{}{}{}(H)}$-approximable operators, $\id{\overline{F}}{}{}{}{}(H) 
\stackrel{1}{=} \id{A}{}{\tup{min}}{}{}(H)$ (cf., e.g., \cite[Theorem 18.3.1 
and Proposition 18.5.4]{J1981}). Moreover, since the rank one operator 
$\tp{J_H x}{-0.45ex}{\scbox{0.70}{$\overline{H}$}}{-0.4ex}
{\scbox{0.70}{$H$}}{x} = \langle\cdot, x\rangle_H\,x \in \id{K}{}{}{}{}(H)^+$, 
it follows from \sref{Remark}{rem:positivity_criterium} that $D_\xi = 
\Theta_{\id{L}{}{}{\tup{min}}{}}(\xi)$ is positive, if $\xi$ is positive. 
Conversely, if the nuclear operator $D_\xi$ is positive, then also $\xi$ is 
positive. This follows from \sref{Theorem}{thm:char_of_trace_class_ops} 
and the well-known fact that $\textup{tr}(ST) = \textup{tr}(TS)$ if $S \in 
\id{S}{2}{}{}{}\,(H)$ and $T \in \id{S}{2}{}{}{}\,(H)$ both are Hilbert-Schmidt 
operators. If namely $A \in \id{L}{}{}{}{}(H)^+$ is an arbitrarily given 
positive operator (not necessarily compact\,!), then 
\begin{align}\label{eq:trace_trick}
\textup{tr}(D_\xi A) = \textup{tr}(D_\xi^{1/2}(D_\xi^{1/2}A)) = \textup{tr}
((D_\xi^{1/2}A)D_\xi^{1/2}) 
\stackrel{\eqref{eq:nuclear_norm_vs_HS_norm}}{\geq} 0
\end{align}
(since $D_\xi^{1/2} \in \id{S}{2}{}{}{}\,(H)$ and hence also $D_\xi^{1/2}A \in 
\id{S}{2}{}{}{}\,(H)$). Such an operator $D_\xi$, satisfying 
\eqref{eq:the_compact_case} is even uniquely defined. In particular, if $H$ is 
finite-dimensional, then we reobtain the well-known fact that 
\eqref{eq:the_compact_case} is satisfied for any state $\xi$ on 
$\id{L}{}{}{}{}(H)$. Note also that $\id{K}{}{}{}{}(H)$ is a unital 
$\tup{C}^\adj$-algebra (with ${\bf{\unit}} : = Id_H$) if and only if $H$ is 
finite-dimensional. 

Clearly, these observations, together with 
\sref{Theorem}{thm:trace_duality_in_the_Hilbert_space_case} imply the following 
particular important well-known trace duality result:
\begin{corollary}\label{cor:trace_duality_particular_case}
Let $H$ and $K$ be  Hilbert spaces. Then
\begin{align*}
(\id{K}{}{}{}{}(H, K))^\prime \cong \id{N}{}{}{}{}(K, H) 
\stackrel{1}{\hookrightarrow} (\id{N}{}{}{}{}(K, H))^{\prime\prime} \cong 
(\id{L}{}{}{}{}(H, K))^\prime
\end{align*}
and
\begin{align*}
(\id{N}{}{}{}{}(K, H))^\prime \cong \id{L}{}{}{}{}(H, K) \cong 
(\id{K}{}{}{}{}(H, K))^{\prime\prime}\,.
\end{align*}
In particular, if $\id{A}{}{}{}{} \in \{\id{N}{}{}{}{}, \id{L}{}{}{}{}\}$, 
then the Banach space $\id{A}{}{}{}{}(K, H)$ is 1-complemented in the bidual 
$(\id{A}{}{}{}{}(K, H))^{\prime\prime}$. 
\end{corollary}
This immediately raises a natural question: what can be said if we 
consider linear functionals $\rho \in (\id{A}{}{\adj}{}{}(K, H))^\prime \cong 
(\id{A}{}{\tup{min}}{}{}(H, K))^{\prime\prime}$ (cf. also 
\sref{Theorem}{thm:normal_states_on_vNAs} and 
\sref{Corollary}{cor:vector_state_version_of_Gleason})? A partial answer follows 
from \sref{theorem}{thm:trace_duality_in_the_Hilbert_space_case} and well-known 
properties of the weak${}^\adj$ topology. More precisely, we have:
\begin{proposition}\label{prop:weak_star_functionals_on_max_BOI_components}
Let $(\id{A}{}{}{}{}, \idn{A}{}{}{}{})$ be a $p$-normed operator ideal 
$(0 < p \leq 1)$. Let $H$ and $K$ be  Hilbert spaces and $\rho \in 
(\id{A}{}{\adj}{}{}(K, H))^\prime$. If $\Theta_{\id{A}{}{\tup{min}}{}{}}^
{\,\prime}\rho$ is weak${}^\adj$ continuous on $(\id{A}{}{\tup{min}}{}{}(H, K))^
{\,\prime}$, then there exists a unique linear operator $A_\rho \in 
\id{A}{}{\tup{min}}{}{}(H, K)$, such that $\idn{A}{}{\tup{min}}{}{}(A_\rho) = 
\Vert \rho\Vert$ and
\[
\langle T, \rho \rangle = \textup{tr}(A_\rho T)\,\text{ for all } T \in 
\id{A}{}{\adj}{}{}(K, H)\,.
\]
If $H = K$, and if the functional $\rho$ is positive, then $A_\rho \in 
\id{A}{}{\tup{min}}{}{}(H)$ is also positive. 
\end{proposition}
\begin{proof}
Put $\E : = \id{A}{}{\tup{min}}{}{}(H, K)$ and $\F : = \id{A}{}{\adj}{}{}(K, H)$. 
Since by assumption, $\E^{\prime\prime} \ni \Theta_{\id{A}{}{\tup{min}}{}{}}^{\,\prime}
\rho$ is weak${}^\adj$ continuous, it follows that 
$\Theta_{\id{A}{}{\tup{min}}{}{}}^{\,\prime}\rho = j_\E A_\rho$, for a unique 
element $A_\rho \in \E$ (cf., e.g., \cite[Korollar VIII.3.4]{W2011}). 
Hence, if $Y \in \F$ is arbitrarily chosen, it follows that
\begin{align}\label{eq:weak_star_topology_and_trace_duality}
\langle Y, \rho \rangle = 
\big\langle \Theta_{\id{A}{}{\tup{min}}{}{}}
\big(\Theta_{\id{A}{}{\tup{min}}{}{}}^{-1}(Y)\big), \rho \big\rangle =
\big\langle \Theta_{\id{A}{}{\tup{min}}{}{}}^{-1}(Y), j_\E A_\rho\big\rangle =
\langle A_\rho, \Theta_{\id{A}{}{\tup{min}}{}{}}^{-1}(Y)\rangle 
\stackrel{\eqref{eq:the_inverse}}{=} \textup{tr}(A_\rho Y)\,.
\end{align}
Since also the dual operator of the isometric isomorphism 
$\Theta_{\id{A}{}{\tup{min}}{}{}} : \E^\prime \stackrel{\cong}{\longrightarrow} 
\F$ is an isometric isomorphism (cf., e.g., 
\sref{Proposition}{prop:isom_isomorphisms_between_Banach_spaces}), it follows 
in particular that
\[
\Vert \rho\Vert_{\F^\prime} = \Vert \Theta_{\id{A}{}{\tup{min}}{}{}}^{\,\prime}
\rho\Vert_{\E^{\prime\prime}} = \Vert j_\E A_\rho\Vert_{\E^{\prime\prime}} = 
\Vert A_\rho\Vert_\E\,.
\]
Finally, if $H = K$, and if $\rho \geq 0$, it follows 
that for any $x \in H$ $S_x : = \tp{J_H x}{-0.4ex}{\scbox{0.70}
{$\overline{H}$}}{-0.45ex}{\scbox{0.70}{$H$}}{x} \in \F \cap 
\id{L}{}{}{}{}(H)^+$, whence
\[
0 \leq \langle S_x, \rho \rangle 
\stackrel{\eqref{eq:weak_star_topology_and_trace_duality}}{=}
\textup{tr}(A_\rho S_x) \stackrel{\eqref{eq:vector_state_omega_x}}{=}  
\langle A_\rho x, x \rangle_H\,,
\]
and the claim follows. 
\end{proof}
\begin{remark}
Let $H$ be a fixed infinite-dimensional Hilbert space over $\C$ and 
$(\id{A}{}{}{}{}, \idn{A}{}{}{}{})$ be a Banach operator ideal. If 
$\id{A}{}{\tup{min}}{}{}(H)$ is a $\tup{C}^\adj$-algebra (where its norm 
is given by $\idn{A}{}{\tup{min}}{}{}(\cdot)$, of course), then the famous 
Gel'fand-Na\u{\i}mark-Segal construction implies that $\id{A}{}{\tup{min}}{}{}(H)$ is 
isometrically $\adj$-isomorphic to a C${}^{\,\adj}$-subalgebra of $\id{L}{}{}{}{}
(H_{\oplus})$, via a faithful representation $\pi : 
\id{A}{}{\tup{min}}{}{}(H) \stackrel{1}{\hookrightarrow} \id{L}{}{}{}{}
(H_{\oplus})$, for some ``large'' Hilbert space $H_{\oplus}$ over $\C$, which in 
general is \textit{not separable}; even if $H$ were (cf. \cite[III.5.2.2]{B2006}). 
Consequently, in general, if $\id{A}{}{\tup{min}}{}{}(H)$ is not separable, and 
if $\id{A}{}{\tup{min}}{}{}(H) \notin \{\{0\}, \id{L}{}{}{}{}(H)\}$, we cannot 
apply Calkin's result, stating that the unique proper closed ideal in 
$\id{L}{}{}{}{}(H_{\oplus})$ would be given by $\id{K}{}{}{}{}(H_\otimes) 
\stackrel{1}{=} \pi\lb\id{A}{}{\tup{min}}{}{}(H)\rb$\, \text{if} $H_{\oplus}$ 
were separable. The subtle classification of all closed ideals in 
$\id{L}{}{}{}{}(H_{\oplus})$ in the nonseparable case, which requires the 
input of ordinal and infinite cardinal numbers and a related generalisation of 
the closed Banach operator ideal of compact operators, is provided in \cite{G1967}.    
\end{remark}  
In the Hilbert space model of quantum mechanics, the states of a physical 
system are typically identified with probability measures on the complete 
orthomodular lattice $Proj(\id{L}{}{}{}{}(H)) : = \{P \in \id{L}{}{}{}{}(H) : P = 
P^\adj = P^2\}$ of all (orthogonal) projections in a Hilbert space $H$, which 
embodies the ``Hilbert space logic'' of the quantum system. The projections are 
in a one-to-one correspondence with closed subspaces of $H$ onto which they 
project. $Proj(\id{L}{}{}{}{}(H))$ is the algebraic description of the lattice of 
all closed subspaces of $H$. A significant example of a probability measure on 
projections is given by certain positive linear functionals on a 
\textit{von Neumann algebra} $\mathcal{M} \subseteq \id{L}{}{}{}{}
(H_{\mathcal{M}})$, acting on the Hilbert space $H_{\mathcal{M}}$. Two known 
deep results that are of great importance for the $\tup{C}^\adj$-algebra model of 
quantum physics in general transfers \eqref{eq:the_compact_case} to the 
representation of \textit{normal} states on arbitrary von Neumann algebras (cf., 
e.g., \cite[Theorem 2.4.21]{BR1987} and \cite[Theorem 46.4]{C2000}). Recall here 
that a $\tup{C}^\adj$-algebra $\mathcal{M}$ is faithfully representable as a von 
Neumann algebra on some Hilbert space $H_{\mathcal{M}}$, if and only if 
$\mathcal{M}$ - as Banach space - is isometrically isomorphic to the dual of 
some Banach space. In fact, it can be verified that any von Neumann algebra 
$\mathcal{M} \subseteq \id{L}{}{}{}{}(H_{\mathcal{M}})$ satisfies 
$\mathcal{M} \cong {\mathcal{M}}_\ast^{\,\prime}$, with Banach space 
$\mathcal{M}_\ast : = \id{N}{}{}{}{}(H_{\mathcal{M}})/\mathcal{M}^\perp$ (the 
predual of $\mathcal{M}$ - cf., e.g., \cite[Theorem 4.6.17]{P1989}). Here,  
\[
\mathcal{M}^\perp : = \{S \in \id{N}{}{}{}{}(H_{\mathcal{M}}) : 
\tup{tr}(ST) = 0\,\text{ for all } T \in \mathcal{M}\}
\]
denotes the $\idn{N}{}{}{}{}(H_{\mathcal{M}})$-closed subset of 
$\id{N}{}{}{}{}(H_{\mathcal{M}})$, and the isometric isomorphism is given by
\[
\Sigma_{\mathcal{M}} : \mathcal{M} \stackrel{\cong}{\longrightarrow} 
{\mathcal{M}}_\ast^{\,\prime}, T \mapsto (D + \mathcal{M}^\perp \mapsto 
\tup{tr}(DT))\,.  
\] 
In particular, it follows that $\id{N}{}{}{}{}(H_{\mathcal{M}})/
\mathcal{M}^\perp \equiv \mathcal{M}_\ast$ is a subspace of the dual Banach 
space $\mathcal{M}^\prime$. By construction, the following 
(composed) isometric embedding is namely well-defined:
\[
\Sigma_{\mathcal{M}}^{\,\prime}\,j_{\mathcal{M}_\ast} : \mathcal{M}_\ast 
\stackrel{1}{\hookrightarrow} (\mathcal{M}_\ast)^{\prime\prime} \cong \mathcal{M}^\prime\,,
\]
where the canonical metric injection $j_{\mathcal{M}_\ast} : \mathcal{M}_\ast 
\stackrel{1}{\hookrightarrow} (\mathcal{M}_\ast)^{\prime\prime}$ appears again; 
this time, as the canonical isometry from the predual of $\mathcal{M}$ into its 
bidual (cf. \sref{Remark}{rem:the_operator_Lambda_z}). The construction implies 
that
\begin{align}\label{eq:tr_of_D_cdot}
\mathcal{M}^\prime \ni \Sigma_{\mathcal{M}}^{\,\prime}\,j_{\mathcal{M}_\ast}
(D + \mathcal{M}^\perp) = \tup{tr}(D\,\cdot)\,\text{ for all }\,
D \in \id{N}{}{}{}{}(H_{\mathcal{M}})\,. 
\end{align}
Consequently, $\Sigma_{\mathcal{M}}^{\,\prime}\,j_{\mathcal{M}_\ast} : 
{\mathcal{M}}_\ast \stackrel{1}{\hookrightarrow} {\mathcal{M}}^\prime$ is an 
isometry, with range $\{\tup{tr}(D\,\cdot) : D \in \id{N}{}{}{}{}
(H_{\mathcal{M}})\}$, implying that the predual $\mathcal{M}_\ast$ in fact can 
be identified with the closed subspace of all \textit{normal} continuous linear 
functionals on $\mathcal{M}$, where the latter is defined as 
$\{\tup{tr}(D\,\cdot) : D \in \id{N}{}{}{}{}(H_{\mathcal{M}})\} \subseteq 
{\mathcal{M}}^\prime$. In particular, $\tup{tr}(D\,\cdot)$ is of norm one if 
and only if there exists a sequence of nuclear operators $(R_n)_{n \in \N} 
\subseteq \id{N}{}{}{}{}(H_{\mathcal{M}})$, such that 
\begin{align*}
\tup{tr}(R_n T) = 0\,\text{ for all }\,T \in \mathcal{M}, n \in \N\,
\text{ and }\,\lim\limits_{n \to \infty} \idn{N}{}{}{}{}(D-R_n) \leq 1
\end{align*}
(due to the definition of a quotient norm), whence $0 \leq \lim\limits_
{n \to \infty}\idn{N}{}{}{}{}(R_n) \leq \idn{N}{}{}{}{}(D) + 1 < \infty$ (due 
to the triangle inequality). Regarding a particularly thorough derivation of the 
canonical isometric isomorphism $\Sigma_{\mathcal{M}} : \mathcal{M} 
\stackrel{\cong}{\longrightarrow} {\mathcal{M}}_\ast^{\,\prime}$ and 
\sref{Theorem}{thm:normal_states_on_vNAs} below, we refer the reader to 
\cite[Vigier's Theorem 4.1.1]{M1990} and \cite[Theorem 4.2.9 and 
Theorem 4.2.10]{M1990}. Regarding nonuniqueness of the nuclear operator 
$D_\psi$ in \sref{Theorem}{thm:normal_states_on_vNAs}, cf. also 
\sref{Remark}{eq:a_further_char_of_normal_states_on_a_vNA}. Hence, if we take 
also \eqref{eq:trace_trick} into account, we arrive at 
\begin{theorem}
\label{thm:normal_states_on_vNAs}
Let $\psi$ be a positive linear functional on a von Neumann algebra 
$\mathcal{M}$ acting on the Hilbert space $H_{\mathcal{M}}$. Then the 
following statements are equivalent:
\begin{enumerate}
\item Whenever $(T_\lambda)$ is a bounded increasing 
net in $\mathcal{M}^+$, then $\psi(T_\lambda) \nearrow \psi(\sup(T_\lambda))$.
\item $\psi$, restricted to the projection lattice $Proj(\mathcal{M}) : = 
Proj(\id{L}{}{}{}{}(H_{\mathcal{M}})) \cap \mathcal{M}$, is a completely 
additive measure on $P(\mathcal{M})$; i.e., if $(E_i)_{i \in I}$ is a pairwise 
orthogonal family of projections in $\mathcal{M}$, then 
\[
\psi\big(\sum\limits_{i \in I}E_i\big) = \sum\limits_{i \in I}\psi(E_i)\,,
\]
where $\sum\limits_{i \in I}E_i$ denotes the limit in the strong operator 
topology $(SOT)$ of $\id{L}{}{}{}{}(H_{\mathcal{M}})$.
\item The linear functional $\big(\Sigma_{\mathcal{M}}^{\,\prime}\big)^{-1} \psi$ 
is weak${}^\adj$ continuous on ${\mathcal{M}}_\ast^{\,\prime}$ 
$($i.e., continuous with respect to the weak${}^\adj$ topology 
$\sigma({\mathcal{M}}_\ast^{\,\prime}, \mathcal{M}_\ast)$$)$.
\item $\psi = \Sigma_{\mathcal{M}}^{\,\prime}\,j_{\mathcal{M}_\ast}x_\psi$, for 
a unique $x_\psi \in \mathcal{M}_\ast$\,.   
\item $\psi$ is a normal bounded linear functional on $\mathcal{M}$, induced 
by some positive operator $D_\psi \in \id{N}{}{}{}{}(H_{\mathcal{M}})$; i.e., 
there is a positive operator $D_\psi \in \id{N}{}{}{}{}(H_{\mathcal{M}})$ such 
that
\[
\langle T, \psi\rangle = \psi(T) = \tup{tr}(D_\psi T) = 
\langle T, \tup{tr}(D_\psi\,\cdot)\rangle\,\text{ for all }\, T \in \mathcal{M}\,.  
\]
\end{enumerate}
In particular, $\psi$ is a normal state, if and only if $D_\psi \in 
S_{\id{N}{}{}{}{}(H_{\mathcal{M}})} \cap \id{L}{}{}{}{}(H_{\mathcal{M}})^+$, 
where $S_E$ denotes the unit sphere of a normed space $E$. 
\end{theorem}
\begin{remark}\label{rem:structure_of_vector_states}
\eqref{eq:vector_state_omega_x} and the structure of the associated operator 
$D_x$, together with \sref{Theorem}{thm:normal_states_on_vNAs}-(iv) clearly imply 
that for any $x \in H$, the (restricted) vector state $\omega_x\bigl\vert_
{\mathcal{M}}$ is a normal state on the von Neumann algebra ${\mathcal{M}}$. 
Taking into account the obvious fact that for any Hilbert space $H$ over $\C$ 
and a given, fixed unit vector $u_0 \in H$, every $x \in H$ satisfies $x = R_x u_0$, 
where $R_x : = \tp{J_H u_0}{-0.3ex}{\scbox{0.70}{$\overline{H}$}}{-0.3ex}
{\scbox{0.70}{$H$}}{x} \in \id{L}{}{}{}{}(H)$, it follows from 
\cite[Corollary 4.5.4]{KR1983} that the faithful GNS representation of the 
$\textup{C}^\adj$-algebra $\id{L}{}{}{}{}(H)$ itself, obtained from the (normal) 
vector state $\omega_{u_0}$ indeed is given by $Id_{\id{L}{}{}{}{}(H)}$, 
implying that the von Neumann algebra $\id{L}{}{}{}{}(H)$ acts on $H$. So, we 
may consider $(H, Id_{\id{L}{}{}{}{}(H)}, u_0)$ as the GNS triple of 
$\id{L}{}{}{}{}(H)$. A rather surprising role of vector states comes into play, 
when we consider the Hilbert space $(\id{S}{2}{}{}{}\,(H), \langle \cdot, 
\cdot\rangle_{\id{S}{2}{}{}{}\,})$ of all Hilbert-Schmidt operators on $H$ 
(cf. \sref{Remark}{rem:L_H_acting_on_L_S2H} and 
\sref{Corollary}{cor:vector_state_version_of_Gleason}).  
\end{remark}
At this point, the reader should recall the remarkable fact that \textit{any} 
positive linear functional $\psi$ on an \textit{arbitrary} $\tup{C}^\adj$-algebra 
is automatically bounded (cf., e.g., \cite[Chapter 2.1, p. 20]{H2003})! In fact, 
if that $\tup{C}^\adj$-algebra is a unital one, with unit $\bf{\unit}$, then 
$\psi \geq 0$ \textit{if and only if} $\psi$ is bounded and $\Vert \psi\Vert = 
\psi(\unit)$ (due to \cite[Corollary 3.3.4]{M1990}). Consequently, if $\tau$ is 
an arbitrarily given linear functional on a unital $\tup{C}^\adj$-algebra, then 
the following statements are equivalent
\begin{enumerate}
\item $\tau \geq 0$ and $\Vert\tau\Vert = 1$ (i.e., $\tau$ is a state).
\item $\tau \geq 0$ and $\tau(\unit) = 1$.
\item $\tau$ is bounded and $\tau(\unit) = \Vert\tau\Vert = 1$.
\end{enumerate}
\begin{remark}\label{eq:a_further_char_of_normal_states_on_a_vNA}
Consider the Banach space $\E : = \id{N}{}{}{}{}(H_{\mathcal{M}})$. If we recall 
\sref{Corollary}{cor:trace_duality_particular_case} and the construction of the 
isometric isomorphism $\Theta_{\id{N}{}{}{}{}}^{-1} : \id{L}{}{}{}{}
(H_{\mathcal{M}}) \stackrel{\cong}{\longrightarrow} \E^\prime$, it follows that 
the set of all normal bounded linear functionals on the 
von Neumann algebra $\mathcal{M}$ precisely coincides with  
\[
\lb\eta_{\mathcal{M}}(\Theta_{\id{N}{}{}{}{}}^{-1})^\prime j_{\E}\rb(\E) 
\subseteq {\mathcal{M}}^\prime\,, 
\]
where $Q_{\mathcal{M}} : (\id{L}{}{}{}{}(H_{\mathcal{M}}))^\prime 
\stackrel{1}{\twoheadrightarrow} \mathcal{M}^\prime, \xi \mapsto \xi{\vert}_
{\mathcal{M}}$ denotes the canonical surjection from the dual space 
$(\id{L}{}{}{}{}(H_{\mathcal{M}}))^\prime$ onto the dual space $\mathcal{M}^
\prime$, which in general is not injective, yet a contraction (Hahn-Banach). In 
fact, if $R \in \mathcal{M} \subseteq \id{L}{}{}{}{}(H_{\mathcal{M}})$ and 
$D \in \E$ are arbitrarily chosen, it follows from 
\sref{Theorem}{thm:trace_duality_in_the_Hilbert_space_case} that
\[
\langle R, \tup{tr}(D\,\cdot)\rangle = \tup{tr}(D R) = 
\langle D,  \Theta_{\id{N}{}{}{}{}}^{-1}(R)\rangle = 
\langle \Theta_{\id{N}{}{}{}{}}^{-1}(R), j_{\E}(D)\rangle = 
\langle R,  (\Theta_{\id{N}{}{}{}{}}^{-1})^\prime j_{\E}(D)\rangle\,.
\]
\sref{Theorem}{thm:normal_states_on_vNAs} therefore implies that $\psi$ is a 
normal positive linear functional on the von Neumann algebra $\mathcal{M}$ 
if and only if there exists an operator $D_\psi \in \E \cap 
\id{L}{}{}{}{}(H_{\mathcal{M}})^+$, such that $\psi = \tup{tr}(D_\psi\,\cdot) =
\lb Q_{\mathcal{M}}\,(\Theta_{\id{N}{}{}{}{}}^{-1})^\prime j_{\E}\rb(D_\psi) 
\stackrel{\eqref{eq:tr_of_D_cdot}}{=} \Sigma_{\mathcal{M}}^{\,\prime}\,
j_{\mathcal{M}_\ast}(D_\psi + \mathcal{M}^\perp)$. In other words, it follows 
from \sref{Theorem}{thm:normal_states_on_vNAs}-(iv) that the set of all normal 
states on $\mathcal{M}$ is also given by the range of the bounded linear 
contraction $Q_{\mathcal{M}}\,(\Theta_{\id{N}{}{}{}{}}^{-1})^\prime j_{\E} : \E 
\longrightarrow \mathcal{M}^{\,\prime}$. In any case, the following diagram 
commutes:
\vspace{-0.7em}
\begin{center}
\includegraphics[width=6.5cm]{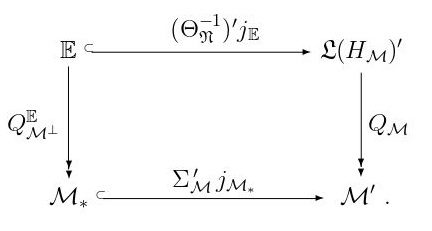} 
\end{center}
\vspace{-0.8em}
Moreover, a standard factorisation implies that $\mathcal{M}_\ast$ coincides 
isometrically with the closure of the range of the operator 
$(\Theta_{\id{N}{}{}{}{}}^{-1})^\prime j_{\E}$. In general, such an operator 
$D_\psi \in \E \cap \id{L}{}{}{}{}(H_{\mathcal{M}})^+$ doesn't have to be an 
element of the von Neumann algebra $\mathcal{M}$, implying that such an 
$D_\psi$ in general is not uniquely defined. However, if $\mathcal{M}$ were 
already the ``whole, large'' von Neumann algebra $\id{L}{}{}{}{}
(H_{\mathcal{M}}) \cong \E^\prime$ itself, we re-recognise the uniqueness of 
$D_\psi$ in this particular case (since $(\id{L}{}{}{}{}(H_{\mathcal{M}}))^\bot 
= \{0\}$, due to Hahn-Banach).   
\end{remark}  
\noindent Note that \sref{Corollary}{cor:rep_of_nuclear_density_operators} 
implies that for any normal state $\psi$, the operator $D_\psi \in 
\id{N}{}{}{}{}(H_{\mathcal{M}}) \cap \id{L}{}{}{}{}(H_{\mathcal{M}})^{+}$ can 
be written as   
\[
D_\psi = \sum\limits_n s_n(D_\psi) (\tp{e_n^\circ}{-0.45ex}{\scbox{0.70}
{$\overline{H}$}}{-0.45ex}{\scbox{0.70}{$H$}}{e_n^\circ})\,,
\]
with some countable $($possibly finite$)$ orthonormal sequence $(e_n^\circ)$ of 
eigenvectors of $D_\psi$, and with $(s_n(D_\psi)) \in l_1$ (the singular values 
of $D_\psi$), satisfying $0 \leq s_n(D_\psi) \leq 1$ for all $n$ and 
$\Vert(s_n(D_\psi))\Vert_{l_1} = \sum\limits_{n} s_n(D_\psi) = 1$. Consequently, 
an application of \eqref{eq:trace_of_a_dyad} and \eqref{eq:dyads_and_composition}, 
together with the continuity of the linear trace functional $\tup{tr} : 
\id{N}{}{}{}{}(H_{\mathcal{M}}) \longrightarrow \C$ implies that
\[
\psi(T) = \tup{tr}(D_\psi T) = 
\sum\limits_n s_n(D_\psi) \langle T e_n^\circ, e_n^\circ \rangle_H = 
\sum\limits_n s_n(D_\psi)\,\omega_n(T)\,\text{ for all }\,T \in \mathcal{M}\,,
\]
with the normal vector state $\mathcal{M} \ni T \mapsto \omega_n(T) : = 
\langle T e_n^\circ, e_n^\circ \rangle_H = 
\tup{tr}((\tp{e_n^\circ}{-0.45ex}{\scbox{0.70}{$\overline{H}$}}
{-0.45ex}{\scbox{0.70}{$H$}}{e_n^\circ})T)$ - even if the Hilbert space $H$ is 
infinite-dimensional and nonseparable.

Regarding \sref{Theorem}{thm:normal_states_on_vNAs}-(iv), note that the Riesz 
representation (of duals of the space of continuous functions on a compact 
Hausdorff space), together with an application of the famous Gel'fand-Na\u{\i}mark 
theorem to an arbitrary \textit{commutative} von Neumann algebra implies that 
the states of an arbitrary commutative von Neumann algebra $\mathcal{M} \cong 
C(\Omega_{\mathcal{M}})$ correspond to ``standard'' 
\textit{probability measures} on some measurable space 
$(\Omega_{\mathcal{M}}, \id{F}{}{}{}{})$, where $\Omega_{\mathcal{M}}$ is a 
compact Hausdorff space (cf., e.g., \cite[Theorem II.2.5 and 
Theorem IX.3.4]{W2011}).

The seminal theorem of A. Gleason, which is of particular importance in the 
foundations and philosophy of quantum mechanics, roughly states that for 
separable Hilbert spaces $H$ over $\C$, being of ``suitably large'' dimension, 
\textit{any} bounded completely additive measure on the projection lattice 
$Proj(\id{L}{}{}{}{}(H))$ actually originates from a \textit{normal} positive 
linear functional on $\id{L}{}{}{}{}(H)$. Consequently, it \textit{implies} the 
Born rule in expectation form (due to 
\sref{Proposition}{prop:weak_star_functionals_on_max_BOI_components} and 
\sref{Theorem}{thm:normal_states_on_vNAs})! 
The assumption of separability was removed by T. Drisch in 1979 (cf. 
\cite[Chapters 3 and 5]{H2003}):
\begin{theorem}[\textbf{Gleason}]
\label{thm:Gleason_Drisch}
Let $H$ be a Hilbert space over $\C$ with $\text{dim}(H) \geq 3$. Then any 
bounded completely additive measure $\mu : Proj(\id{L}{}{}{}{}(H)) 
\longrightarrow \C$ on the projection lattice $Proj(\id{L}{}{}{}{}(H))$ 
\underline{extends uniquely} to a \underline{normal} positive linear functional on 
$\id{L}{}{}{}{}(H)$. In particular, there is a unique positive operator 
$D_\mu \in \id{N}{}{}{}{}(H)$ such that
\[
\mu(E) = \tup{tr}(P D_\mu) \text{ for all } P \in Proj(\id{L}{}{}{}{}(H)).
\]
\end{theorem}
\noindent It is remarkable that the heart of the proof is lying in the 
(unit sphere of the) three-dimensional real Hilbert space $\R^3$. A nontrivial 
generalisation of \sref{Theorem}{thm:Gleason_Drisch}, which sheds light on the 
surprising role of the dimension of the Hilbert space, follows from a deep 
result from S. Dorofeev and A. N. Sherstnev (cf. \cite[Theorem 3.3.5]{H2003}):
\begin{theorem}[\textbf{Dorofeev/Sherstnev}]
If $H$ is infinite-dimensional, then any completely additive measure $\mu$ on 
the projection lattice is bounded.
\end{theorem}
At this juncture, it is important to recall \sref{Remark}{rem:L_H_acting_on_L_S2H}, 
\sref{Theorem}{thm:char_of_trace_class_ops} and 
\sref{Theorem}{thm:normal_states_on_vNAs}. If namely $\mathcal{M}$ is a given 
von Neumann algebra and $\rho$ a normal state on $\mathcal{M}$, then there 
exists an operator $D_\rho \in S_{\id{N}{}{}{}{}(H_{\mathcal{M}})} \cap 
\id{L}{}{}{}{}(H_{\mathcal{M}})^+$, such that for all $T \in \mathcal{M}$, 
\[
\langle T, \rho\rangle = \textup{tr}(D_\rho T)\,.
\]
Hence, $0 \leq X_\rho : = D_\rho^{1/2} \in  K_{\mathcal{M}} : = 
\id{S}{2}{}{}{}\,(H_{\mathcal{M}})$, and
\[
\langle T, \rho\rangle = \textup{tr}(X_\rho T X_\rho) = 
\langle T X_\rho, X_\rho\rangle_{K_{\mathcal{M}}} 
\stackrel{\eqref{eq:pi_L_acting_from_the_left}}{=} 
\langle \pi_L(T)X_\rho, X_\rho\rangle_{K_{\mathcal{M}}} = 
\langle \pi_L(T), \omega_{X_\rho}\rangle\,.
\]
Consequently, it follows that 
\[
\rho = \omega_{X_\rho}\circ \pi_L\!\mid_{\mathcal{M}} = 
(\pi_L\!\mid_{\mathcal{M}})^\prime\omega_{X_\rho}\,\,(!)
\] 
In other words, \textit{any normal state} on $\mathcal{M}$ can be 
\textit{extended to a vector state} on the von Neumann algebra $\id{L}{}{}{}{}
(K_{\mathcal{M}}) \equiv \id{L}{}{}{}{}(\id{S}{2}{}{}{}\,(H_{\mathcal{M}}))$:
\vspace{-2em}
\begin{center}
\includegraphics[width=3.8cm]{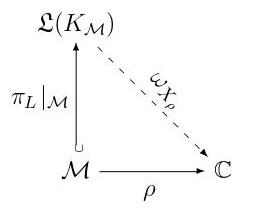} 
\end{center}
\vspace{-1em}
In fact, if we deal with $\mathcal{M} : = \id{L}{}{}{}{}(H)$ itself, we may 
choose $H_{\mathcal{M}} : = H$ (due to 
\sref{Remark}{rem:structure_of_vector_states}), and we have arrived at, 
\begin{corollary}[\textbf{Vector state version of Gleason's theorem}]
\label{cor:vector_state_version_of_Gleason}
Let $H$ be an infinite-dimensional Hilbert space over $\C$. Put $K : = 
\id{S}{2}{}{}{}\,(H)$. Then any completely additive measure $\mu : 
Proj(\id{L}{}{}{}{}(H)) \longrightarrow \C$ on the projection lattice 
$Proj(\id{L}{}{}{}{}(H))$ \underline{extends uniquely} to a vector state 
$\omega_{X_\mu}$ on $\id{L}{}{}{}{}(K)$, where $X_\mu \in S_K \cap 
\id{L}{}{}{}{}^+(H)$: 
\[
\mu(E) = \omega_{X_\mu}\circ\pi_L(P) = \langle PX_\mu, X_\mu\rangle_K 
\text{ for all } P \in Proj(\id{L}{}{}{}{}(H))\,. 
\] 
\end{corollary}
Of course, if the Hilbert space $H$ is finite-dimensional, and if $\textup{dim}
(H) \geq 3$, the vector state version of Gleason's theorem would also hold 
if we assume in addition that the completely additive measure $\mu : 
Proj(\id{L}{}{}{}{}(H)) \longrightarrow \C$ on the projection lattice 
$Proj(\id{L}{}{}{}{}(H))$ is bounded.

One of the most famous applications of Gleason's theorem is the proof of 
the nonexistence of dispersion-free states on the ``Hilbert space logic''. The 
dispersion free state is a state attaining only values 0 and 1. Such a state 
assigns to each projection probability 0 or 1 and hence eliminates the 
probability character of the model. More specifically, Gleason's theorem rules 
out hidden-variable models that are ``noncontextual''. In order to avoid the 
implications of Gleason's theorem, any hidden-variable model for quantum 
mechanics must include hidden variables that are not properties belonging to the 
measured system alone but also depend on the external context in which the 
measurement is made.

Actually, \sref{Theorem}{thm:char_of_trace_class_ops} is a special case of a 
deep result of Grothendieck (cf. 
\cite[Theorem 1.5.1, Theorem 1.7.15]{P1987}).} To this 
end, recall that $T \in \mathfrak{L}(E,F)$ between Banach spaces $E$ and $F$ is 
called absolutely $p$-summing ($1 \leq p < \infty$) if there exists a constant 
$c \geq 0$ such that for all $n \in \N$ and $x_1, \ldots, x_n \in E$
\begin{align}\label{eq:absolutely_p_summing_operator}
\big(\sum_{i=1}^n \Vert Tx_i\Vert^p\big)^{1/p} \leq 
c\,\sup\limits_{a \in B_{E^\prime}}
\big(\sum_{i=1}^n \vert \langle x_i, a\rangle \vert^p\big)^{1/p}.
\end{align}
The $p$-summing norm $\idn{P}{\!p}{}{}{}\,(T)$ is defined as the 
infimum of all constants $c \geq 0$ which satisfy 
\eqref{eq:absolutely_p_summing_operator} (cf., e.g., \cite[Chapter 11]{DF1993}) 
or \cite[Chapter 2]{DJT1995}). In the Hilbert space case, it is known that 
${\mathfrak{S}}_2(H,K)$ coincides with ${\mathfrak{P}}_2(H,K)$, even with equal 
$1$-Banach ideal norms; i.e.,
\begin{align}\label{eq:HS_equals_P2}
\id{S}{2}{}{}{}\,(H, K) \stackrel{1}{=} \id{P}{2}{}{}{}\,(H,K)
\end{align}
for all Hilbert spaces $H$, $K$ (cf., e.g., \cite[Theorem 1.4.5]{P1987} and 
\cite{P2014}).
\begin{theorem}\label{thm:N2_vs_P2_min}
\begin{align}\label{eq:N2_is_P2_min}
\id{P}{2}{\bf{(a)}}{}{} \stackrel{1}{=} \id{N}{2}{}{}{} \stackrel{1}{=} 
\frak{\overline{F}} \circ {\frak{P}}_2 \stackrel{1}{=} \id{P}{2}{\tup{min}}{}{}\,
\stackrel{1}{=} {\frak{P}}_2 \circ \frak{\overline{F}}\,. 
\end{align}
In particular,
\[
\id{N}{2}{}{}{}\,(E, F) \stackrel{1}{\hookrightarrow} \id{P}{2}{}{}{}\,(E, F)
\]
is an isometry. If $E$ and $F$ are arbitrarily given Banach spaces, then every 
$T \in \id{N}{2}{}{}{}\,(E, F)$ admits a factorisation
\vspace{-0.5em}
\begin{center}
\includegraphics[width=6.5cm]{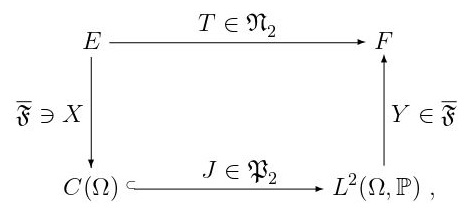} 
\end{center}
\vspace{-0.6em}
for some compact Hausdorff space $\Omega$ and some regular probability measure 
$\P$ on $\Omega$, where $J : C(\Omega) \hookrightarrow L^2(\Omega, \P)$ denotes 
the canonical injection, with $\idn{P}{2}{}{}{}\,(J) = 1$. In this 
case,
\[
\idn{N}{2}{}{}{}\,(T) = \inf(\Vert Y \Vert\,\Vert X \Vert)\,,
\]
where the infimum is taken over all possible factorisations.     
\end{theorem}
\begin{proof}
Our short proof is primarily based on an application of fundamental results of 
Pietsch and highly nontrivial local structural properties of quasi-normed 
operator ideals, known as ``accessibility conditions'', where the latter are 
listed and discussed in depth in \cite{DF1993, Oe2026}. 

Firstly, since $(\id{P}{2}{}{}{}\,, \idn{P}{2}{}{}{}\,)$ is self-adjoint (cf., 
e.g., \cite[Theorem 19.2.14]{P1980}), it follows from 
\cite[Theorem 9.3.1 and Theorem 18.2.5]{P1980} (applied to $r=2=r^\ast$) that 
the (normed\,!) Banach operator ideal $(\id{N}{2}{}{}{}\,, \idn{N}{2}{}{}{}\,)$ satisfies
\[
\id{N}{2}{\tup{max}}{}{} \stackrel{1}{=} \id{N}{2}{\adj\adj}{}{} \stackrel{1}{=} 
\id{P}{2}{\adj}{}{} \stackrel{1}{=} \id{P}{2}{}{}{}
\]
Thus, if we form the minimal kernel on both sides of that isometric equality, 
\cite[Theorem 18.1.4]{P1980} implies that
\[
\id{N}{2}{}{}{} \stackrel{1}{=} \id{N}{2}{\tup{min}}{}{} \stackrel{1}{=} 
(\id{N}{2}{\tup{max}}{}{})^{\tup{min}} \stackrel{1}{=} 
\id{P}{2}{\tup{min}}{}{}.
\]
Since $(\id{P}{2}{}{}{}\,,{\Vert \cdot \Vert}_{\id{P}{2}{}{}{}})$ is (totally) 
accessible (cf. \cite[Theorem 21.5]{DF1993}), we may apply 
\cite[Proposition 25.2]{DF1993}, and it follows that $\frak{\overline{F}} \circ 
{\frak{P}}_2 \stackrel{1}{=} \id{P}{2}{\tup{min}}{}{} \stackrel{1}{=} 
{\frak{P}}_2 \circ \frak{\overline{F}}$. Hence, since $\id{N}{2}{}{}{} 
\stackrel{1}{=} \id{P}{2}{\tup{min}}{}{}$, a resulting application of Pietsch's 
$\id{P}{2}{}{}{}\,$-factorisation theorem (cf., e.g., 
\cite[Theorem 1.5.1]{P1987}), together with \cite[Theorem 1.7.14]{P1987} 
concludes the proof.
\end{proof}
\begin{theorem}[\textbf{Grothendieck, Jarchow/Ott}]
\label{thm:L2_circ_N_equals_P2_circ_P2}
Let $E$ and $F$ be arbitrary Banach spaces and $H$ be an arbitrary 
Hilbert space $($over $\F$$)$. 
\begin{enumerate}
\item A bounded operator $T \in 
\id{L}{}{}{}{}(E, F)$ can be represented as a composition of two absolutely 
$2$-summing operators if and only if there exist a Hilbert space $K$, a nuclear 
operator $S \in \id{N}{}{}{}{}(E, K)$ and a bounded operator 
$R \in \id{L}{}{}{}{}(K, F)$, such that $T = RS$. In this case,
\[
\idn{P}{2}{}{}{}\,\circ \idn{P}{2}{}{}{}\,(T) = 
\inf(\Vert R \Vert\,\idn{N}{}{}{}{}(S))\,,
\]
where the infimum is taken over all possible factorisations. In particular, 
the composition of two absolutely 2-summing operators is nuclear, and
\[
\idn{N}{}{}{}{}(T) \leq \idn{P}{2}{}{}{}\,\circ
\idn{P}{2}{}{}{}\,(T) 
\]
for all $T \in \id{P}{2}{}{}{}\circ\id{P}{2}{}{}{}\,(E,F)$.
\item 
\[
\id{P}{2}{}{}{}\circ\id{P}{2}{}{}{}\,(E, H) \stackrel{1}{=} \id{N}{}{}{}{}(E, H).
\]
\end{enumerate}  
\end{theorem}
\begin{proof}
The statements follow directly from an application of the fundamental
\cite[Theorem 1.7.15]{P1987}.
\end{proof}
In fact, if we apply nontrivial facts about the composition of accessible 
operator ideals (cf. \cite[Chapter 25]{DF1993}), the latter result can be 
improved. To this end, we have to recall the construction of the far-reaching 
maximal Banach operator ideal of $2$-factorable operators $(\id{L}{2}{}{}{}\,,
\idn{L}{2}{}{}{}\,)$. Let $E$ and $F$ be $\F$-Banach spaces 
and $T \in {\mathfrak{L}}(E, F)$. By definition, $T \in \id{L}{2}{}{}{}\,(E, F)$ 
if and only if there exist an Hilbert space $H$ and bounded linear 
operators $R \in \mathfrak{L}(H, F)$, $S \in \mathfrak{L}(E, H)$, such that 
$T = RS$. $T \in \id{L}{2}{}{}{}\,(E,F)$ is said to be \textit{2-factorable} 
(cf., e.g., \cite[Corollary 18.6.2]{DF1993}).
$(\id{L}{2}{}{}{}\,,\idn{L}{2}{}{}{}\,)$ is an injective, maximal Banach 
operator ideal which is totally accessible. The norm is defined as
\[
\idn{L}{2}{}{}{}\,(T) : = \inf(\Vert R\Vert\,\Vert S\Vert),
\]
where the infimum is taken over all possible factorisations.
\begin{corollary}\label{cor:L2_circ_N_equals_N2_circ_N2}
The product ideal $(\id{L}{2}{}{}{}\,\circ\id{N}{}{}{}{}\,,
\idn{L}{2}{}{}{}\,\circ\idn{N}{}{}{}{})$ is minimal, regular and injective, 
totally accessible and $\frac{1}{2}$-normed. Moreover,
\[
\id{L}{2}{}{}{}\,\circ\id{N}{}{}{}{} \stackrel{1}{=} \id{N}{2}{}{}{}\,\circ 
\id{N}{2}{}{}{}\, \stackrel{1}{=} \id{P}{2}{}{}{}\,\circ\id{P}{2}{}{}{}\,\,.
\]
In particular, the composition of two $2$-nuclear operators is nuclear.
\end{corollary}
\begin{proof}
From \sref{Theorem}{thm:L2_circ_N_equals_P2_circ_P2} (respectively 
\cite[Theorem 1.7.15]{P1987}), it follows that 
$\id{L}{2}{}{}{}\,\circ\id{N}{}{}{}{} \stackrel{1}{=} \id{P}{2}{}{}{}\,\circ 
\id{P}{2}{}{}{}\,\,$. Consequently, since both, $(\id{L}{2}{}{}{}\,,
\idn{L}{2}{}{}{}\,)$ and $(\id{P}{2}{}{}{}\,,\idn{P}{2}{}{}{}\,)$ are 
(totally) accessible, and since $(\id{N}{}{}{}{}\,,\idn{N}{}{}{}{}\,)$ is 
minimal, an application of \cite[Proposition 25.2]{DF1993} together with 
\eqref{eq:N2_is_P2_min} implies that
\[
\id{L}{2}{}{}{} \circ \id{N}{}{}{}{} \stackrel{1}{=} \id{L}{2}{}{}{} \circ 
\id{\overline{F}}{}{}{}{} \circ \id{N}{}{}{}{} \circ \id{\overline{F}}{}{}{}{}
\stackrel{1}{=} \id{\overline{F}}{}{}{}{} \circ \id{L}{2}{}{}{} \circ 
\id{N}{}{}{}{} \circ \id{\overline{F}}{}{}{}{} \stackrel{1}{=}
\id{\overline{F}}{}{}{}{} \circ \id{P}{2}{}{}{} \circ 
\id{P}{2}{}{}{} \circ \id{\overline{F}}{}{}{}{} \stackrel{1}{=}
\id{N}{2}{}{}{} \circ \id{N}{2}{}{}{}\,\,.
\]
\cite[Lemma 5.1]{Oe1998} and \cite{Oe2026} conclude the proof. 
\end{proof}
\sref{Corollary}{cor:char_of_tensor_product_of_HS_as_Hilbert_Schmidt}, 
together with \cite[Theorem 5.30]{DJT1995} now closes the loop and implies 
the following - nontrivial - link between absolutely $2$-summing operators, 
$2$-nuclear operators and the tensor product of  Hilbert spaces, which 
should be compared with \eqref{eq:H_otimes_pi_K} and 
\eqref{eq:H_otimes_vareps_K}:
\begin{proposition}\label{prop:induced_tensor_norm_g2}
Let $H$ and $K$ be two  Hilbert spaces. Then
\begin{align}\label{eq:HS_P2_N2_between_Hilbert_Spaces}
{\mathfrak{P}}_2(\overline{H},K) \stackrel{1}{=} {\mathfrak{N}}_2^{}
(\overline{H},K) \cong H \otimes K
\end{align}
holds isometrically and canonically. In particular, the norm on the Hilbert 
space $H \otimes K$ coincides isometrically with the totally accessible tensor 
norm $g_2 = g_2^t = g_2^\adj = g_2^\prime$:
\[
H \otimes K \cong H \widetilde{\otimes}_{g_2} K.
\]
Moreover, also
\begin{align}\label{eq:HS_P2_N2_between_Hilbert_Spaces_II}
{\mathfrak{P}}_2(H_1 \otimes H_2, K_1 \otimes K_2) \cong 
{\mathfrak{P}}_2(H_1, \overline{K}_1) \otimes {\mathfrak{P}}_2(H_2, K_2) \cong 
{\mathfrak{P}}_2(H_1, K_2) \otimes {\mathfrak{P}}_2(H_2, \overline{K}_1)
\end{align}
holds isometrically and canonically for all Hilbert spaces $H_1, H_2, K_1$ 
and $K_2$.
\end{proposition}
\begin{proof}
It's a straightforward exercise to verify that 
\eqref{eq:HS_P2_N2_between_Hilbert_Spaces_II} instantly follows from 
\sref{Theorem}{thm:structure_of_Hilb_space_tp} and 
\sref{Theorem}{thm:associativity_law_of_HS_tensorprod}. We leave the details 
to the reader.
\end{proof}

At this juncture, we would like to shed light on the role of 
Hilbert-Schmidt operators in quantum mechanics. It is well-known that 
composite systems are mathematically described as tensor products of Hilbert 
spaces over $\C$. Hence, they can be equivalently represented as linear operators 
via \eqref{eq:HS_P2_N2_between_Hilbert_Spaces}. By implementing the first 
isometric equality of \eqref{eq:HS_P2_N2_between_Hilbert_Spaces} canonical 
relations between separable and entangled composite quantum states and maps, 
which are known for quantum systems in finite-dimensional Hilbert spaces, can be 
generalised canonically to the more realistic situation of infinite-dimensional 
(and not necessarily separable) Hilbert spaces. Moreover, the use of the first isometric 
equality of \eqref{eq:HS_P2_N2_between_Hilbert_Spaces} also allows a proof of an 
infinite-dimensional version of Choi's theorem (cf. \cite{GKM2007}).

Even algebraic quantum field theory (AQFT), originally coined by Araki, Haag
and Kastler exhibits a highly impressive, surprising and beautiful emergence of 
Banach operator ideals in the sense of Pietsch!

Listing details about AQFT (which is a large research field on its own) would 
be - far - beyond the scope of our contribution. Astute readers will find a 
brief pedagogical outline of the basic formalism of algebraic quantum field 
theory (including the crucial ``split property'') in \cite{FR2020}, 
\cite[Chapter 11.1]{H2003}, in the excellent survey articles 
\cite{BF2023, HM2006, V2025}, and in the second edition of Haag's seminal, 
comprehensive book \cite{H1996}. These sources should be accompanied by the 
standard reference \cite{BR1987} on operator algebras applied to 
mathematical physics. A concise step-by-step description of the fundamental 
Araki-Haag-Kastler axioms, included related figures and notations can be found 
in \cite{L2022}.  

Very roughly formulated, AQFT is a theory in which the mathematical machinery, 
based on the theory of ``rings of operators'', which J. von Neumann established 
in 1935-1940 (partially in collaboration with F. J. Murray), made it possible 
to express in a mathematically rigorous manner the physical intuition about 
field theory formulated by A. Einstein. Firstly, according to Einstein's general 
relativity theory, every observation depends on its ``location'' in 
Minkowski spacetime. 
It is therefore necessary to define observables ``locally'' in relation to 
bounded regions in Minkowski spacetime only. Secondly, according to Einstein's 
principle of causality, the observables associated to causally disjoint regions 
should be independent. The split property could be viewed as a 
sharpening of Einstein's locality principle, describing the possibility for 
experimenters in spacelike separated regions 
in Minkowski spacetime $\M : = (\R^4, g)$ (or more generally in smooth, globally 
hyperbolic spacetime manifolds) to perform \textit{both}, 
preparations and measurements independently. 

At this point, let us recall (e.g., from \cite[Chapter 11.4]{H2003}) that a pair 
$(\mathcal{M}, \mathcal{N})$ of two commuting von Neumann algebras $\mathcal{M}$ 
and $\mathcal{N}$, both acting on some common Hilbert space $H$, is 
\textit{split} if there exists a type I (von Neumann) factor $\mathcal{F}$, such 
that
\begin{align}\label{eq:split_property_of_vNAs}
\mathcal{M} \subseteq \mathcal{F} \subseteq \mathcal{N}^{\,\#}\,.
\end{align}
A very readable introduction to the type classification 
of von Neumann algebras can be found in \cite{HM2006}. At this point we only 
recall that type I factors were completely classified by F. J. Murray and J. 
von Neumann: up to $\adj$-isomorphism, every factor of type I coincides with the 
von Neumann algebra $\id{L}{}{}{}{}(H_0)$, for some complex Hilbert space $H_0$, 
which could be finite-dimensional (cf., e.g., \cite[Kapitel IX.5]{W2011}). 
Observe that the double commutant theorem (applied to the von Neumann algebra 
$\mathcal{F}$) implies that \eqref{eq:split_property_of_vNAs} is equivalent to 
\[
\mathcal{N} \subseteq 
\mathcal{F}^{\,\#}\cong \C\,Id_{H_0} \subseteq 
\mathcal{M}^{\,\#}
\]
(since $\mathcal{N} \subseteq \mathcal{N}^{\,\#\#} \subseteq \mathcal{F}^{\,\#}$ 
and $\mathcal{M} \subseteq \mathcal{M}^{\,\#\#}$).

Unfortunately, since the term ``split'' is not defined univocally 
(cf., e.g., \cite{BF2023, F2016, FR2020}), one has to be particularly careful 
to avoid confusion, when topics in relation to the split property are discussed. 
A brief - yet valuable - summary of the related von Neumann algebra terminology 
is listed in \cite{HR2023} (without proofs, though). Regarding introductory 
details, we recommend to study \cite[Chapter 4]{M1990}.

To the best of our knowledge, H. J. Borchers and R. Schumann were the first 
authors who implemented a famous result of Grothendieck into the framework of 
AQFT in their paper \cite{BS1991} (namely 
\sref{Theorem}{thm:a_corollary_of_Grothendieck}, but without inclusion of the 
conjugate Hilbert space). In the case of charged superselection sectors, where 
no vacuum vector is available, Grothendieck's result enabled them to link the 
split property of pairs of von Neumann algebras, satisfying the 
Araki-Haag-Kastler-axioms 
(explicitly described in \cite[Section 3]{BS1991}), with a suitable family of 
\textit{nuclear operators} $\Theta_{\psi_0, \beta} \in \id{N}{}{}{}{}
(\mathcal{A}(O), H)$ (mapping such a von Neumann algebra ``of 
Araki-Haag-Kastler-type'' $\mathcal{A}(O)$, acting on a Hilbert space $H$, into 
$H$). The physical context, underlying the extension of the Araki-Haag-Kastler-
axioms by that ``nuclearity condition'' is discussed in \cite[Section 1]{BS1991} 
(and the related references therein).

More precisely, Borchers and Schumann assume that for any given open, 
bounded double cone $O \in \M$ (in Minkowski spacetime), there is a cyclic and 
separating vector $\psi_0 \in H$ with bounded energy support, where $H$ is the 
underlying Hilbert space of the von Neumann algebra $\mathcal{A}(O)$ and $\H$ is 
the related self-adjoint Hamiltonian, such that the following conditions are 
satisfied: there are constants $\beta_0 > 0, c > 0$ and $n \in \N$ (which only 
depend on the choice of $O \in \M$), such that for all $\beta \in (0, \beta_0]$ 
\begin{align}\label{eq:nuclearity_condition_I}
\Theta_{\psi_0, \beta} \in \id{N}{}{}{}{}(\mathcal{A}(O), H)
\end{align}
and 
\begin{align}\label{eq:nuclearity_condition_II}
\idn{N}{}{}{}{}(\Theta_{\psi_0, \beta}) \leq \exp(c\,\beta^{-n}),
\end{align}
where
\[
\mathcal{A}(O) \ni A \mapsto \Theta_{\psi_0, \beta}(A) : = 
\exp(-\beta\,\H) A\exp(\beta\,\H)\psi_0\,.
\]
Regarding statement (ii) in the following seminal, purely mathematical result 
(without any relation to applications in physics), a direct comparison 
indicates that Schumann's definition of ``strong statistical independence of a 
pair of mutually commuting von Neumann algebras'' in \cite{S1996} in fact 
coincides with the definition of ``$W^\adj$-independence in the spatial product 
sense of a pair of mutually commuting von Neumann algebras'', introduced by 
M. Redei and J. S. Summers in \cite{RS2010} (cf. also \cite{BW1986, DL1983, F2016, 
HM2006, S1994, S1996}). 
\begin{theorem}\label{thm:split_char}
Let $\mathcal{M}$ and $\mathcal{N}$ be commuting von Neumann algebras, acting 
on a complex Hilbert space $H$. Then the following statements are equivalent:  
\begin{enumerate}
\item The pair $(\mathcal{M}, \mathcal{N})$ is split.
\item The von Neumann algebras $\mathcal{M} \vee \mathcal{N}$ $($i.e., the von 
Neumann algebra generated by sums and products of elements in $\mathcal{M}$ and 
$\mathcal{N}$$)$ and $\mathcal{M}\,\overline{\otimes}\,\mathcal{N}$ are 
spatial-unitarily equivalent; i.e., there is a unitary operator 
$U : H \stackrel{\cong}{\longrightarrow} H \otimes H$, such that
\[
U(AB)U^{\adj} = A \otimes B\,\text{ for all } (A, B) \in \mathcal{M} \times 
\mathcal{N}
\]
and
\[
U(\mathcal{M} \vee \mathcal{N})U^{\adj} = \mathcal{M}\,\overline{\otimes}\,
\mathcal{N}\,. 
\]
\end{enumerate}
\end{theorem}
Consequently, within the axiomatic framework of AQFT in the sense of 
Araki, Haag and Kastler, \cite[Theorem 4.1]{BS1991}, together with 
\sref{Theorem}{thm:split_char} and \cite[Section 7]{FR2020}, immediately leads 
to  
\begin{theorem}
\label{thm:nuclearity_implies_split_property}
Let $O_1 \in \M$ and $O_2 \in \M$ be open, bounded double cones in Minkowski 
spacetime, such that their closures $\overline{O_1}$ and $\overline{O_2}$ are 
compact. Assume that $O_1$ and $O_2$ are spacelike separated. If the nuclearity 
conditions \eqref{eq:nuclearity_condition_I} and 
\eqref{eq:nuclearity_condition_II} are satisfied for $\mathcal{A}(O_2)$, then 
the pair $(\mathcal{A}(O_1), \mathcal{A}(O_2))$ is split.
\end{theorem}
A natural question instantly arises: is it possible to reverse the 
latter statement? In other words, if we assume the split property, could then 
a ``nuclearity condition'' be derived? To tackle this problem, also Schumann 
made use of the modular structure, arising from Tomita-Takesaki modular theory 
(as was previously accomplished by Buchholz, D'Antoni and Longo - \cite{BDL1990, 
BF2023}). A quick introduction to Tomita-Takesaki modular theory can be found in 
\cite{HM2006, S2005, V2025} (without proofs, though). A comprehensive review, built on 
the polar decomposition of closed, densely defined, unbounded operators, 
including complete proofs, is contained in \cite[Chapter 6.1 and Chapter 9.2]
{KR1986}. Particularly, \cite[Theorem 7.20]{W1980} and 
\cite[Proposition 5.3.18]{P1989} establish an extension of the polar 
decomposition of bounded operators to the unbounded operator case in a concise 
and elegant manner. A strongly simplified approach to the Tomita-Takesaki theory 
is revealed in \cite[Chapter 18.3]{P2018} (built on the approach of 
\cite{RvD1977}), including the proof of Theorem 8.3.14 therein. 
\begin{theorem}
\label{thm:Tomita_Takesaki_1967} 
Let $v_0$ be a cyclic and separating vector for a von Neumann algebra 
$\mathcal{M} \subseteq \id{L}{}{}{}{}(H)$ on a complex Hilbert space $H$. Then 
there are a positive, self-adjoint $($in general unbounded$)$, densely defined 
and invertible linear operator $\Delta : D(\Delta) (\subseteq H) 
\stackrel{\approx}{\longrightarrow} R(\Delta) = D(\Delta)$, $($the so called 
modular operator associated with the pair $(\mathcal{M}, v_0)$$)$, and a 
semilinear isometric isomorphism $J: H \stackrel{\cong}{\longrightarrow} H$ with 
$J^2 = Id_H$ $($the so called modular conjugation$)$, such that
\[
J\,\mathcal{M}\,J = \mathcal{M^{\,\#}}\,\text{ and }\,
\Delta^{it}\,\mathcal{M}\,\Delta^{-it} = \mathcal{M}
\]
for every $t \in \R$ $($where $\Delta^{\pm it} \equiv e^{\pm it\log(\Delta)}$ - 
Borel-function calculus$)$. Moreover, $\Delta^{-1/2} = J\Delta^{1/2}J$, 
$Jv_0 = \Delta v_0 = v_0$ and $\mathcal{M}v_0 \subseteq 
D(\Delta^{1/2})$, with
\[
J \Delta^{1/2}A v_0 = A^\ast v_0\,\text{ for all }\,A \in \mathcal{M}\,.
\]
In particular, $\mathcal{M^\#}v_0 \subseteq D(\Delta^{-1/2})$, with
\[
J \Delta^{-1/2}B v_0 = B^\ast v_0\,\text{ for all }\,B \in \mathcal{M^\#}\,.
\]
\end{theorem}
Actually - taking \sref{Theorem}{thm:split_char} 
properly into account - Schumann constructed a related \textit{absolutely 
$2$-summing operator}, \textit{mapping a von Neumann algebra into a Hilbert 
space} (cf. \cite[Theorem 6.7]{S1996}). However, his short proof seems to be 
incomplete and is built on a possibly empty assumption (cf. 
\sref{Remark}{rem:flaw}). Nevertheless, for the convenience of the reader, we 
provide a detailed and complete proof of Schumann's result. \textit{If} that 
operator could be even represented as a composition of two absolutely 
$2$-summing operators, it would be already a nuclear one (due to 
\sref{Theorem}{thm:L2_circ_N_equals_P2_circ_P2}). Concretely formulated, we have:
\begin{theorem}\label{thm:split_property_and_2_summing_modular_op}
Let $\mathcal{M} \subseteq \id{L}{}{}{}{}(H)$ and $\mathcal{N} \subseteq 
\id{L}{}{}{}{}(H)$ be commuting von Neumann algebras, acting on a complex 
Hilbert space $H$. Assume that the pair $(\mathcal{M}, \mathcal{N})$ is split. Let 
$U \in \id{L}{}{}{}{}(H, H \otimes H)$ be the corresponding unitary operator. 
Let $u, v \in S_H$ be unit vectors in $H$, such that $w_0 : = U^\adj
(u \otimes v) \in H$ is a cyclic and separating vector for the von Neumann 
algebra $\mathcal{N}^{\#}$. Let $\Delta$ be the modular operator, associated 
with the pair $(\mathcal{N}^{\#}, w_0)$. Then the linear operator 
$\Theta_{w_0} \in \id{L}{}{}{}{}(\mathcal{M}, H)$, defined as
\[
\Theta_{w_0} : \mathcal{M} \longrightarrow H, A \mapsto \Delta^{1/4} A w_0\,,
\]
is absolutely $2$-summing, and
\[
\idn{P}{\!2}{}{}{}\,(\Theta_{w_0}) = \Vert\Theta_{w_0}\Vert = 1\,.
\] 
\end{theorem}
\begin{proof}
An application of \sref{Theorem}{thm:Tomita_Takesaki_1967} to the von Neumann 
algebra $\mathcal{N}^{\#}$ leads to the existence of the modular operator 
$\Delta$, associated with the pair $(\mathcal{N}^{\#}, w_0)$ and the corresponding 
modular conjugation $J$. Thus, $\mathcal{M} \ni T \mapsto \langle T, a_1\rangle : = 
\langle Tu, u \rangle_H$ and $\mathcal{M} \ni T \mapsto \langle T, a_2\rangle : = 
\overline{\langle JTJ v, v \rangle_H}$ are well-defined bounded linear 
functionals on the von Neumann algebra $\mathcal{M}$, satisfying 
$\Vert a_1\Vert = \Vert a_2\Vert = 1$ (since $J^2 = Id_H \in \mathcal{M}$). 
Since the pair $(\mathcal{M}, \mathcal{N})$ is split (by assumption), it 
follows from \sref{Theorem}{thm:split_char}-(ii) and the construction of $w_0 \equiv 
U^{\adj}(u \otimes v)$ that
\begin{align}\label{eq:split_property_applied}
\langle AB w_0, w_0\rangle_H = 
\big\langle (A \otimes B)(u \otimes v), u \otimes v\big\rangle_{H \otimes H} =
\langle A u, u \rangle_H\,\langle B v, v \rangle_H 
\end{align}
for all $(A, B) \in \mathcal{M} \times \mathcal{N}$. 
\sref{Theorem}{thm:Tomita_Takesaki_1967} therefore implies that 
\eqref{eq:split_property_applied} in particular is satisfied for the 
operators $A \in \mathcal{M}$ and $B : = JAJ \in J \mathcal{M} J \subseteq 
J \mathcal{N}^{\#} J \stackrel{(!)}{=} \mathcal{N}^{\#\#} = \mathcal{N}$, where 
$A \in \mathcal{M}$ is arbitrarily chosen. Hence, 
\begin{align}
\begin{split}
\Vert \Theta_{w_0}(A)\Vert_H^2 & = 
\langle \Delta^{1/4} A w_0, \Delta^{1/4} A w_0\rangle_H =
\langle A w_0, \Delta^{1/2} A w_0\rangle_H = 
\langle J(J A w_0), J (A^\ast w_0)\rangle_H\\ 
&= \langle A^\ast w_0, J A w_0 \rangle_H = \langle w_0, A(J A J)w_0 \rangle_H =
\overline{\langle A (J A J) w_0, w_0\rangle_H}\\
&= \langle A (J A J) w_0, w_0\rangle_H 
\stackrel{\eqref{eq:split_property_applied}}{=} 
\langle A, a_1 \rangle\,\overline{\langle JAJ, a_2 \rangle}
\end{split} 
\end{align}
(since $J w_0 = w_0$). Thus, $\Vert\Theta_{w_0}\Vert = 1$, and an application 
of the inequality of arithmetic and geometric means further implies that 
\[
\Vert\Theta_{w_0}(A)\Vert_H^2 = \vert \langle A, a_1 \rangle\vert\,
\vert \langle JAJ, a_2 \rangle \vert \leq \frac{1}{2}
(\vert \langle A, a_1 \rangle\vert^2 + \vert\langle JAJ, a_2\rangle \vert^2)\,.
\]
for all $A \in \mathcal{M}$. Consequently, if $n \in \N$ and $A_1, \ldots, A_n \in 
\mathcal{M}$, it follows that
\[
\sum\limits_{i=1}^n \Vert \Theta_{w_0}(A_i)\Vert_H^2 \leq \frac{1}{2}
\lb \sum\limits_{i=1}^n \vert \langle A_i, a_1 \rangle\vert^2 + 
\sum\limits_{i=1}^n \vert \langle A_i, a_2 \rangle\vert^2\rb \leq 
\sup\limits_{a \in B_{\mathcal{M}^\prime}}
\big(\sum_{i=1}^n \vert \langle A_i, a\rangle \vert^2\big)\,, 
\]
implying that $\Theta_{w_0} \in \id{P}{\!2}{}{}{}\,(\mathcal{M}, H)$, and 
$1 = \Vert\Theta_{w_0}\Vert \leq \idn{P}{\!2}{}{}{}\,(\Theta_{w_0}) 
\stackrel{\eqref{eq:absolutely_p_summing_operator}}{\leq} 1$.
\end{proof}
\begin{remark}\label{rem:flaw}
Unfortunately, there seems to be a flaw regarding Schumann's assumptions, 
listed in \sref{Theorem}{thm:split_property_and_2_summing_modular_op}. 
It is not clear to us namely whether there is always a pair $(u, v) \in S_H 
\times S_H$, such that the vector $U^\adj(u \otimes v) \in H$ is a cyclic and 
separating vector for the von Neumann algebra $\mathcal{N}^{\#}$. This 
assumption is crucial. It namely implies the presence of the corresponding 
- decisive - modular conjugation $J$, as the proof of 
\sref{Theorem}{thm:split_property_and_2_summing_modular_op} clearly shows\,!

However, observe that if $\varphi \in (\mathcal{N}^{\#})^\prime$ is a 
faithful normal state on $\mathcal{N}^{\#}$ with faithful GNS representation 
$(\pi_\varphi, H_\varphi, w_\varphi)$, then $\mathcal{N}^{\#}$ can be 
isometrically identified with the von Neumann algebra $\pi_\varphi
\big(\mathcal{N}^{\#}\big)$ (cf. \cite[Theorem 7.5.3]{KR1986}). 
Consequently, $w_\varphi$ is a cyclic and separating vector for 
$\mathcal{N}^{\#} \cong \pi_\varphi\big(\mathcal{N}^{\#}\big)$.
\end{remark}
% ...
\section{General probabilistic theories and the Banach operator ideal 
$(\id{P}{1}{}{}{}\,, \idn{P}{1}{}{}{}\,)$}
\label{sec:GPTs_and_P1}
Just prior to the end of our work, we want to point out a further nontrivial 
application of Banach operator ideals, yet without diving into technical details 
here. In relation to the investigation of compatibility of dichotomic 
measurements in general - finite-dimensional - probabilistic theories (GPTs), 
the operator ideal of absolutely $1$-summing operators plays a key role; very 
similarly to its emergence in context of the famous Grothendieck inequality and 
its connection to Euclidean norms of real- and complex Gaussian random vectors 
and Gaussian hypergeometric functions (cf. \cite[Proposition 9.5]{BJN2022} 
and \cite[Remark 4.2]{Oe2024})!

Any GPT - including finite-dimensional quantum theory as a particular case 
(\cite[Example 4.2]{BJN2022}) - is built on basic operational notions of states 
(or preparation procedures) and effects (or dichotomic measurements) of the 
theory, which are identified with certain elements in a dual pairing of ordered 
finite-dimensional vector spaces, $(V, V^+)$ (``state space'') and $(A, A^+) : = 
(V^\prime, (V^+)^\prime)$ (``effect space''). 

More specifically, from a bird's eye view, 
a GPT is a particular case of a triple $(V, C, u) \equiv ((V, C), u)$, where 
$(V, C)$ is a real ordered topological vector space (which could be 
infinite-dimensional), with \textit{closed} generating cone $C \subseteq V$; 
i.e., $C$ is assumed to be a closed and generating wedge, which satisfies $C 
\cap (-C) = \{0\}$ (cf. \cite[Definition 1.1 and Definition 1.2]{AT2007}), and 
$u \in C^\prime$ is assumed to be an existing order unit of $C^\prime$, where 
$C^\prime : = \{a \in V^\prime : \langle x, a\rangle \geq 0\,\text{ for all }\, 
x \in C\}$ denotes the dual wedge; i.e., $u \in C^\prime$ satisfies
\[
V^\prime = (V^\prime)_u : = \{a \in V^\prime : \exists\,\lambda > 0\,\text{ such 
that }\, a \in [-\lambda u, \lambda u]\} = \bigcup\limits_{n=1}^\infty n[-u,u]\,.
\]
In geometric terms, the dual wedge consists of all linear functionals in 
$V^\prime$ that support the wedge $C$ at the origin. The order relation on $V$ 
is given as $y \geq x :\Leftrightarrow y - x \in C$, whence $C = V^+ : = 
\{x \in V : x \geq 0\}$. If $(V, V^+, u)$ is a GPT, it is assumed by definition 
that $V$ (and hence $V^\prime$) is finite-dimensional (cf. \cite{BJN2022}). In 
this case, an order unit of $C^\prime$ always exists and coincides with an 
interior point of $C^\prime$.
Even in the finite-dimensional case, the assumption of the closedness of the 
cone $C$ is a nontrivial task which is deeply linked with Farkas' duality lemma 
in linear programming and hence cannot be ignored (cf., e.g., 
\cite[Appendix]{OeO2009}). Moreover, lexicographical cones in a vector space of 
dimension $n \geq 2$ are never closed. Since by assumption $C = V^+$ is a 
generating wedge, the dual wedge $C^\prime$ is an Archimedean 
$\text{weak}^\adj$-closed cone of $V^\prime$.
Consequently, the Minkowski functional
\[
\Vert \cdot \Vert_u : V^\prime \longrightarrow [0, \infty),
\]  
defined via $\Vert a \Vert_u : = \inf\{\rho > 0 : a \in \rho\,[-u, u]\}$ 
is a norm on $V^\prime = (V^\prime)_u$, and the ($\Vert \cdot \Vert_u$-closed) 
unit ball of the normed space $(V^\prime, \Vert \cdot \Vert_u)$ coincides with 
the order interval $[-u, u]$; i.e., 
\[
\{a \in V^\prime : \Vert a \Vert_u \leq 1\} = [-u,u].
\]
The states of the physical system, modelled by $(V, V^+, u)$, are represented 
by vectors in the closed convex set 
\[
\Omega_u : = \{\omega \in V^+ : \langle \omega, u \rangle 
= 1\} = V^+ \cap u^{-1}(\{1\}).
\]
Since $u$ is an order unit of $C^\prime$, it clearly follows that the 
dual cone $C^\prime$ is generating and that $u$ is $C$-strictly 
positive; i.e., $\langle x, u\rangle > 0$ for all $x \in V^+\setminus\{0\}$. 
Hence, $\Omega_u$ is a base for the cone $V^+$. 
Let 
\begin{align*}
U \equiv U_{\Omega_u} :&= \text{cx}(\Omega_u \cap (-\Omega_u))\\ 
&= \{\lambda \omega_1 -(1-\lambda)\omega_2 : 0 \leq \lambda \leq 1
\text{ and } \omega_1, \omega_2 \in \Omega_u\}.
\end{align*}
be the convex hull of $\Omega_u \cap (-\Omega_u)$.
Then the Minkowski functional of the convex, circled, and absorbing set $U$ 
\[
V \ni x \mapsto p_U(x) : = \inf\{\rho > 0 : x \in \rho\,U\}
\]
is a semi-norm. If $V$ is a Riesz space, then the base $\Omega_u$ is 
linearly compact and $\Vert \boldsymbol{\cdot} \Vert_U : = p_U$ defines a norm 
on $V$.
In particular, $(V, V^+, \Vert\boldsymbol{\cdot} \Vert_U)$ is a base norm space 
(with base $\{x \in V^+ : \langle x, u \rangle = 1\}$). $\Omega_u$ is 
$\Vert\boldsymbol{\cdot} \Vert_U$-closed, if and only if $V^+$ is 
$\Vert\boldsymbol{\cdot} \Vert_U$-closed. 
Consequently, {\textit{if} $(V, V^+, \Vert\boldsymbol{\cdot} \Vert_U)$ 
is even a Banach lattice, and if the cone $V^+$ is 
$\Vert\boldsymbol{\cdot}\Vert_U$-closed}, we may apply a result of T. And\^{o} 
(cf. \cite[Theorem 2.42]{AT2007}), implying that the dual cone $C^\prime = 
(V^+)^\prime$ is normal. Moreover, 
it can be shown that in fact $(V^\prime, \Vert \cdot\Vert_u) = 
(V, \Vert\cdot\Vert_U)^\prime$. 
It follows that
\begin{align*}
\Vert x \Vert_{U} &= \inf\{r > 0 : x \in r U\} = 
\sup\{\vert\langle x, a\rangle\vert\ : a \in [-u, u]\}\\
&= \inf\{\langle x_+ + x_-, u\rangle : x = x_+ - x_-\,,\, x_\pm \geq 0\}.  
\end{align*}
This approach {also can be adopted if it is assumed that
$(V, V^+, \Vert\boldsymbol{\cdot} \Vert_U)$ is a reflexive Banach space}.

Conversely, 
every base norm space $(V, V^+, \Vert\boldsymbol{\cdot} \Vert)$ with closed, 
generating cone $V^+$ and base $\{x \in V^+ : \langle x, \psi \rangle = 1\}$ 
(for some strictly positive functional $\psi \in V^\prime$) induces the 
triple $(V, V^+, \psi)$, which is by definition a GPT if $V$ is assumed to be 
finite-dimensional. \textit{If} $V$ is finite-dimensional, then both, 
$\Omega_u $ and $U = \text{cx}(\Omega_u \cap (-\Omega_u))$ are 
compact subsets of $V$.

In particular, the triple
\begin{align}\label{eq:GPT_modelling_a_fin_dim_quantum_system}
(\H_n, \H_n^+, \textup{tr}) 
\end{align}
is a GPT, which models a finite-dimensional quantum system with $n$ levels, 
where $\H_n$ denotes the set of all Hermitian $n \times n$-matrices and $\H_n^+ 
\equiv \M_n(\C)^{+}$ the cone of all positive-semidefinite matrices. The base 
is given by state space $\{D \in \H_n^+ : \langle D, \textup{tr} \rangle\ = 
\textup{tr}(D) = 1\}$. In fact, since for any Hermitian matrix $A = (a_{kl}) \in 
\H_n$, 
\[
A = \sum_{k=1}^n a_{kk}E_{kk} + \sum_{k<l}^n 
\Re(a_{kl})(E_{kl}+E_{lk}) + \sum_{k<l}^n \Im(a_{kl})
(i\,(E_{kl}-E_{lk})),
\] 
where $E_{lk} : = e_l e_k^\top$, we reobtain the well-known fact that $\H_n$ is 
an $n^2$-dimensional real vector space (since $n + 2(\frac{n^2-n}{2}) = n^2$). 

Elements of the order interval $[0, u] = \{f \in V^\prime: 0 \leq f \leq u\}$ 
are called \textit{effects}. Based on the view of \cite{BJN2022} that the GPT 
formalism allows probabilistic predictions of the outcomes of measurements with 
a finite number of outcomes $1, 2, \ldots, k$, performed on a certain state, a 
single measurement is modelled as an element $(f_1, \ldots, f_k) \in [0,u]^k$, 
such that $\sum_{j=1}^k f_j = u$, and $\langle\omega, f_j \rangle \in [0,1]$ 
represents the probability that the measurement gives the outcome $j$, if the 
system is in the state $\omega \in \Omega_u$ (in accordance with the Born rule). 
If $k=2$, the measurement is a binary one, leading to either ``yes'' or ``no'' 
outcomes. Obviously, any single measurement $f \equiv (f_1, \ldots, f_k) \in 
[0,u]^k$ induces the linear mapping $\Psi_f : V \longrightarrow \R^k$, defined 
as $\Psi_f(v) : = (\langle v, f_1\rangle, \ldots,\langle v, f_k\rangle)^\top$ 
which satisfies $\Psi_f(\Omega_u) \subseteq \Delta_k$, where $\Delta_k : = 
\{ p \in [0,1]^k : \sum\limits_{j=1}^k p_i = 1\}$ denotes the probability 
simplex (actually without having to make an explicit use of the 
extremal points of the latter one).

In relation to the GPT \eqref{eq:GPT_modelling_a_fin_dim_quantum_system}, the 
self-duality of the PSD cone, 
together with the representation of the smallest and largest eigenvalue of a 
Hermitian matrix (via Courant-Fischer, respectively Raleigh-Ritz) implies that 
effects with respect to the 
GPT \eqref{eq:GPT_modelling_a_fin_dim_quantum_system} actually can be 
isometrically identified with positive operator-valued measures (POVMs).

A series of $g$ measurements $f^{(1)}, \ldots, f^{(g)}$, $f^{(i)} \equiv 
(f^{(i)}_1, \ldots, f^{(i)}_{k_i}) \in [0, u]^{k_i}$ ($i \in [g]$) are said to 
be \textit{compatible}, if all $f^{(i)}$ are marginals of a single joint 
measurement; i.e., for all $i \in [g]$ and $j \in [k_i]$, $f^{(i)}_{j} \in 
[0,u]$ can be represented as
\[
f^{(i)}_{j} = \sum\limits_{\widehat{m_i} \in \A_i} h_{m_1, \ldots, m_{i-1}, j, m_{i+1}, 
\ldots, m_g}
\]
for some series of joint measurements $\big\{h_m : m \in \prod\limits_{\nu=1}^g 
[k_\nu]\big\} \subseteq [0,u]$, where $\widehat{m_i} : = (m_1, \ldots, m_{i-1},\\ 
m_{i+1}, \ldots, m_g) \in \A_i : = \prod\limits_{\nu \in [g]\setminus\{i\}} 
[k_\nu]$ (cf. \cite[Section 4.4]{BJN2022}). Compatibility of dichotomic 
measurements can thus be understood as the existence of a series of joint 
measurements $\{h_m : m \in \{1,2\}^g\}$ for a given series of $g$ single 
measurements $f^{(1)}, \ldots, f^{(g)}$, such that all $f^{(i)}$s are marginals, 
where the outcome set of each single measurement $f^{(i)}$ is a binary set 
$\{f^{(i)}_1, f^{(i)}_2\}$.

Any noncompatible $g$-tuple of measurements $(f^{(1)}, \ldots, f^{(g)})$ can 
be made compatible if one adds ``enough random white noise''. The amount of noise 
needed for this is the basis for the commonly used definition of a so called 
``compatibility region'', described by the largest probability parameter $s \in 
[0,1]$ of the added ``white noise''. The numerical value (for given $g \in \N$), 
such that this maximal ``compatibility probability'' $s \in [0,1]$ can be 
applied to \textit{all} collections of $g$-tuples of single measurements with 
outcomes in the set $\prod_{i=1}^g \{f^{(i)}_1,f^{(i)}_2\}$, implies the 
existence of the number $\gamma(g; V, V^+, u) \in [0,1]$; the so called 
\textit{compatibility degree} of the GPT $(V, V^+, u)$. All details of its 
construction is given in \cite[Section 4]{BJN2022}.

A partial and surprising result of \cite[Proposition 9.5]{BJN2022}, which 
is built on quotients of norms of tensor products of finite-dimensional 
Banach spaces (or equivalently on quotients of norms of the minimal Banach 
operator ideals of nuclear and approximable operators), states that the set
$\{\gamma(g; V, V^+, u) : g \in \N\} \subseteq [0,1]$ is also uniformly bounded 
from below: 
\[
1 \geq \gamma(g; V, V^+, u) \geq \frac{1}{\idn{P}{1}{}{}{}\,(Id_V)} > 0\,
\text{ for all } g \in \N\,.
\]
\begin{ack*}
I want to express my gratitude to Dr Robin Schumann for forwarding his 
highly fascinating dissertation \cite{S1994} in the past to me. I also would 
like to thank the anonymous non-virtual assistant for a very helpful electronic 
correspondence, which in particular has expanded the list of references to 
include further significant sources (unknown to me before) including 
\cite{G2025}.
\end{ack*}
\noindent\dci{We declare that we do have no known competing financial interests 
or personal relationships that could have appeared to influence the work 
presented in this article.}\\[0.5em]
\noindent\aiuse{We declare that we have not used AI-assisted technologies in 
creating this article.}


\begin{thebibliography}{10}

\bibitem{AT2007}
C. D. Aliprantis and R. Tourky.
\newblock\textit{Cones and duality}.
\newblock Graduate Studies in Mathematics 84. Providence, RI: American 
Mathematical Society (2007).

\bibitem{AS2017}
G. Aubrun and S. J. Szarek. 
\newblock\textit{Alice and Bob meet Banach. The interface of asymptotic geometric analysis 
and quantum information theory}. 
\newblock Mathematical Surveys and Monographs 223. Providence, RI: American Mathematical 
Society. xxi (2017).

\bibitem{B2013}
S. K. Berberian.
\newblock Tensor product of Hilbert spaces. 
\url{https://web.ma.utexas.edu/mp_arc/c/14/14-2.pdf}. 
\newblock\textit{Unpublished} (2013).

\bibitem{B2006}
B. Blackadar
\newblock\textit{Operator algebras. Theory of $C^\adj$-algebras and von 
Neumann algebras}.
\newblock Encyclopaedia of Mathematical Sciences 122. Operator Algebras and 
nonCommutative Geometry III. Springer, Berlin, Heidelberg, New York (2006).
\newblock Revised and corrected version (2017), available via  
\url{https://bruceblackadar.com/Mathematics/Cycr.pdf}.

\bibitem{BJN2022}
A. Bluhm, A. Jen\v{c}ov\'{a}, and I. Nechita.
\newblock Incompatibility in general probabilistic theories, generalized 
spectrahedra, and tensor norms.
\newblock\textit{Commun. Math. Phys. 393, No. 3, 1125-1198} (2022).

\bibitem{GV2023}
M. Van Den Bossche and P. Grangier. 
\newblock Contextual unification of classical and quantum physics. 
\newblock\textit{Found. Phys. 53, No. {\bf{2}}, Paper No. 45, 24 p.} (2023). 

\bibitem{BS1991}
H. J. Borchers and R. Schumann. 
\newblock A nuclearity condition for charged states. 
\newblock\textit{Lett. Math. Phys. 23, No. \textbf{1}, 65-77} (1991).

\bibitem{BR1987}
{
O. Bratteli and D. W. Robinson.
\newblock\textit{Operator Algebras and Quantum Statistical Mechanics 1. 
$\tup{C}^\adj$- and $\tup{W}^\adj$-Algebras, Symmetry Groups, Decomposition of 
states. 2nd ed}.
Texts and Monographs in Physics. Springer, Berlin-Heidelberg-New York (1987).
}

\bibitem{BDF1987}
{
D. Buchholz, C. D'Antoni, and K. Fredenhagen.
\newblock The universal structure of local algebras. 
\newblock\textit{Comm. Math. Phys. 111, 123-135} (1987).
}

\bibitem{BDL1990}
{
D. Buchholz, C. D'Antoni, and R. Longo.
\newblock Nuclear maps and modular structures II: Applications to quantum field theory. 
\newblock\textit{Commun. Math. Phys. 129, 115-138} (1990).
}

\bibitem{BF2023}
D. Buchholz and K. Fredenhagen.
\newblock Algebraic quantum field theory: objectives, methods, and results.
\newblock\textit{\url{https://arxiv.org/abs/2305.12923}} (2023).

\bibitem{BW1986}
D. Buchholz and E. H. Wichmann.
\newblock Causal Independence and the Energy-Level Density of States in Local 
Quantum Field Theory.
\newblock\textit{Commun. Math. Phys. 106, 321-344 } (1986).

\bibitem{CY1961}
P. Civin and B. Yood.
\newblock The second conjugate space of a Banach algebra as an algebra. 
\newblock\textit{Pac. J. Math. 11, 847-870} (1961).

\bibitem{C2000}
J. B. Conway.
\newblock\textit{A course in operator theory}.
Graduate Studies in Mathematics 21. Providence, RI: American Mathematical 
Society (2000)

\bibitem{DL1983}
{
C. D'Antoni and R. Longo.
\newblock Interpolation by Type I Factors and the Flip Automorphism.
\newblock\textit{J. Funct. Anal. 51, 361-371} (1983).
}

\bibitem{DF1993}
A. Defant and K. Floret.
\newblock\textit{Tensor norms and operator ideals}.
\newblock North-Holland Mathematics Studies 176. North-Holland, Amsterdam (1993).

\bibitem{DFS2008}
{
J. Diestel, J. Fourie, and J. Swart.
\newblock\textit{The metric theory of tensor products. Grothendieck's 
r\'{e}sum\'{e} revisited}.
\newblock Providence, RI: American Mathematical Society (2008).
}

\bibitem{DJP2001}
{
J. Diestel, H. Jarchow, and A. Pietsch.
\newblock Operator ideals.
\newblock\textit{Handbook of the geometry of Banach spaces, Volume 1, 
Elsevier, Amsterdam, 437-496} (2001).
} 

\bibitem{DJT1995}
J. Diestel, H. Jarchow, and A. Tonge.
\newblock\textit{Absolutely Summing Operators}. 
\newblock Cambridge University Press, Cambridge (1995).




\bibitem{E2020}
J. Earman.
\newblock Quantum Physics in Non-Separable Hilbert Spaces.
\newblock\textit{\url{https://philsci-archive.pitt.edu/18363/}} (2020).

\bibitem{F2016}
C. J. Fewster. 
\newblock The split property for quantum field theories in flat and curved 
spacetimes. 
\newblock\textit{Abh. Math. Semin. Univ. Hambg. 86, No. \textbf{2}, 153-175} (2016).

\bibitem{FR2020}
C. J. Fewster and K. Rejzner.
\newblock Algebraic quantum field theory. An introduction.
\newblock\textit{Finster, Felix (ed.) et al., Progress and visions in quantum 
theory in view of gravity: bridging foundations of physics and mathematics. 
Selected talks presented at the seventh international conference, Leipzig, 
Germany, October 1-5, 2018. Cham: Birkh\"{a}user}, 1-61 (2020).


\bibitem{G2025}
M. Gallego. 
\newblock\textit{Quantum theory at the macroscopic scale}.
\newblock Ph.D. thesis, Faculty of Physics, University of Vienna, Austria 
(2025).

\bibitem{GPP2014}
S. R. Garcia, E. Prodan, and M. Putinar.
\newblock Mathematical and physical aspects of complex symmetric operators.
\newblock\textit{J. Phys. A, Math. Theor. 47, No. \textbf{35}, Article ID 353001}, 
54 p. (2014).

\bibitem{G2022}
P. Garrett.
\newblock Discrete Fubini-Tonelli.
\newblock\textit{\url{https://www-users.cse.umn.edu/~garrett/m/real/notes_2022-23/03a_discrete_Fubini-Tonelli.pdf}} (2022).


\bibitem{GKM2007}
J. Grabowski, M. Ku\'{s}, and G. Marmo.
\newblock On the relation between states and maps in infinite dimensions.
\newblock\textit{Open Syst. Inf. Dyn. 14, No. \textbf{4}, 355-370} (2007). 

\bibitem{G1967}
B. Gramsch. 
\newblock Eine Idealstruktur Banachscher Operatoralgebren (in German). 
\newblock\textit{J. Reine Angew. Math. 225, 97-115} (1967).

\bibitem{G1953}
A. Grothendieck. 
\newblock R\'{e}sum\'{e} de la Th\'{e}orie M\'{e}trique des Produits Tensoriels 
Topologiques.  
\newblock{\it Bol. Mat. Sao Paulo}, No. {\bf{8}}, 1-79 (1953/1956).\newline
Reprinted in \newblock{\it Resen. Inst. Mat. Estat. Univ. Sao Paulo 2}, No. 
\textbf{4}, 401-480 (1996). 


\bibitem{H1996}
R. Haag.
\newblock\textit{Local Quantum Physics. Fields, Particles, Algebras. 2nd., rev. 
and enlarged ed}.
Texts and Monographs in Physics. Springer, Berlin-Heidelberg-New York. (1996).

\bibitem{HM2006}
H. Halvorson and M. M\"{u}ger.
\newblock Algebraic Quantum Field Theory.
\newblock\textit{\url{https://arxiv.org/abs/math-ph/0602036}} (2006).

\bibitem{H2003}
J. Hamhalter. 
\newblock\textit{Quantum Measure Theory}. 
Fundamental Theories of Physics, Vol. 134. Kluwer Academic Publishers, 
Dordrecht (2003)

\bibitem{H2011}
C. Heil. 
\newblock\textit{A Basis Theory Primer: Expanded Edition}.
\newblock Springer, New York (2011).

\bibitem{HR2023}
{
S. Hollands and A. Ranallo.
\newblock Channel Divergences and Complexity in Algebraic QFT.
\newblock\textit{Commun. Math. Phys. 404, 927-962} (2023).
}

\bibitem{J1981}
H. Jarchow.
\newblock\textit{Locally convex spaces}.
\newblock Mathematische Leitf\"{a}den. Stuttgart: B. G. Teubner (1981).

\bibitem{JO1982}
H. Jarchow and R. Ott.
\newblock On trace ideals. 
\newblock\textit{Math. Nachr. 108, 23-37} (1982).

\bibitem{KR1983}
R. V. Kadison and J. R. Ringrose.
\newblock\textit{Fundamentals of the theory of operator algebras. Vol. 1: 
Elementary theory}. 
\newblock Pure and Applied Mathematics, 100. New York-London etc.: Academic 
Press. XV (1983).

\bibitem{KR1986}
R. V. Kadison and J. R. Ringrose.
\newblock\textit{Fundamentals of the theory of operator algebras. Vol. 2: 
Advanced theory}. 
\newblock Pure and Applied Mathematics, 100. New York-London etc.: Academic 
Press. XV (1986).

\bibitem{K1995}
N. J. Kalton.
\newblock An elementary example of a Banach space not isomorphic to its complex 
conjugate.
\newblock\textit{Can. Math. Bull. 38, No. {\textbf{2}}, 218-222} (1995).

\bibitem{LNRR1981}
C. Laurie, E. Nordgren, H. Radjavi, and P. Rosenthal. 
\newblock On triangularization of algebras of operators. 
\newblock\textit{J. Reine Angew. Math. 327, 143-155} (1981).

\bibitem{L2022}
Sascha Lill. 
\newblock\textit{Time Dynamics in Quantum
Field Theory Systems}.
\newblock Ph.D. thesis, Mathematisch-Naturwissenschaftliche Fakult\"{a}t der 
Eberhard-Karls-Universit\"{a}t T\"{u}bingen (2022).

\bibitem{LP1968}
{
J. Lindenstrauss and A. Pe{\l}czy\'{n}ski.
\newblock Absolutely summing operators in $L_p$-spaces and their applications.
\newblock\textit{Stud. Math. 29, 275-326} (1968). 
}

\bibitem{MV1997}
{
R. Meise and D. Vogt.
\newblock\textit{Introduction to Functional Analysis. Transl. from the German 
by M. S. Ramanujan}.
Oxford Graduate Texts in Mathematics. Clarendon Press. Oxford (1997).
}

\bibitem{M1990}
{
G. J. Murphy.
\newblock\textit{$C^\adj$-algebras and operator theory}.
Academic Press, Inc., Boston, MA etc. (1990)
}

\bibitem{NC2010}
M. A. Nielsen and I. L. Chuang.
\newblock\textit{Quantum computation and quantum information - 10th Anniversary 
Edition}. 
\newblock Cambridge University Press, Cambridge (2010).

\bibitem{Oe1998}  
F. Oertel.
\newblock Local properties of accessible injective operator ideals. 
\newblock\textit{Czech. Math. J. 48, No. \textbf{1}, 119-133} (1998).

\bibitem{Oe2003}  
F. Oertel.
\newblock On normed products of operator ideals which contain $\mathfrak{L}_2$ 
as a factor. 
\newblock\textit{Arch. Math. 80, 61-70} (2003).

\bibitem{Oe2006}
F. Oertel.
\newblock On random measures, unordered sums and discontinuities of the first 
kind.
\newblock\textit{\url{https://arxiv.org/abs/math/0609395}} (2006).

\bibitem{OeO2009}
F. Oertel and M. P. Owen.
\newblock Geometry of polar wedges in Riesz spaces and super-replication 
prices in incomplete financial markets.
\newblock\textit{Positivity 13, No. \textbf{1}, 201-224} (2009).

\bibitem{Oe2024}
F. Oertel.
\newblock\textit{Upper bounds for Grothendieck constants, quantum correlation matrices 
and CCP functions}.
\newblock Lecture Notes in Mathematics 2349. Springer, Cham (2024).

\bibitem{Oe2026}
F. Oertel.
\newblock\textit{Local structures in quasi-normed operator ideals and
trace duality: a unifying framework}.
\newblock Work in Progress.



\bibitem{P2020}
M. Pankov.
\newblock\textit{Wigner-type theorems for Hilbert Grassmannians}. 
\newblock London Mathematical Society Lecture Note Series 460. Cambridge 
University Press, Cambridge (2020).


\bibitem{P2015}
P. Pajot.
\newblock La revanche d'un th\'{e}or\`{e}me oubli\'{e} (in French).
\newblock\textit{\hspace{-0.2cm}
\url{https://www.larecherche.fr/la-revanche-dun-th\'{e}or\`{e}me-oubli\'{e}}
} (2015).

\bibitem{P1989}
G. K. Pedersen.
\newblock\textit{Analysis now}. 
\newblock Graduate Texts in Mathematics 118. Springer, New York (1989).

\bibitem{P2018}
G. K. Pedersen.
\newblock\textit{$C^\adj$-Algebras and Their Automorphism Groups. Edited 
by S. Eilers and D. Olesen. 2nd edition}. 
\newblock Pure and Applied Mathematics. Elsevier/Academic Press, Amsterdam (2018).

\bibitem{P1980}
A. Pietsch.
\newblock\textit{Operator ideals}. 
\newblock North-Holland Mathematical Library 20. North-Holland, Amsterdam (1980).

\bibitem{P1981}
A. Pietsch.
\newblock Operator Ideals with a Trace.
\newblock\textit{Math. Nachr. 100, 61-91} (1981).

\bibitem{P1987}
A. Pietsch.
\newblock\textit{Eigenvalues and s-numbers}.
\newblock Cambridge Studies in Advanced mathematics, 13. Cambridge University 
Press, Cambridge (1987).

\bibitem{P2014}
A. Pietsch.
\newblock Traces of operators and their history. 
\newblock\textit{Acta Comment. Univ. Tartu. Math. 18, No. 1, 51-64} (2014).

\bibitem{RS2010}
M. R\'{e}dei and S. J. Summers. 
\newblock When Are Quantum Systems Operationally Independent?
\newblock\textit{Int. J. Theor. Phys. 49, 3250-3261} (2010).

\bibitem{RvD1977}
M. A. Rieffel and A. van Daele.
\newblock A bounded operator approach to Tomita-Takesaki theory. 
\newblock\textit{Pacific J. Math. 69, 187-221} (1977).

\bibitem{R2022}
B. W. Roberts.
\newblock\textit{Reversing the arrow of time}.
\newblock Cambridge University Press, Cambridge - open access (2022).

\bibitem{R2020}
R. R\"{u}{\ss}mann. 
\newblock\textit{Tensor Product of Hilbert Spaces}.
\newblock M.Sc. thesis, Faculty of Mathematics, TU Kaiserslautern, Germany 
(2020).

\bibitem{R2002}
{
R. A. Ryan. 
\newblock\textit{Introduction to tensor products of Banach spaces}.
\newblock Springer Monographs in Mathematics. Springer, London (2002).
}

\bibitem{S1994}
R. Schumann. 
\newblock Operatorenideale und die statistische Unabh\"{a}ngigkeit in der 
Quantenfeldtheorie (in German). 
\newblock Dissertation, Georg-August-Universit\"{a}t zu G\"{o}ttingen (1994).

\bibitem{S1996}
R. Schumann. 
\newblock Operator ideals and the statistical independence in quantum field 
theory. 
\newblock\textit{Lett. Math. Phys. 37, No. \textbf{3}, 249-271} (1996).

\bibitem{SP2026}
L. K. Singh and A. Peperko. 
\newblock Sherman-Takeda type theorems for locally $\tup{C}^\adj$-algebras.
\newblock\textit{\url{https://arxiv.org/abs/2601.00717}} (2026).

\bibitem{S2005}
S. J. Summers. 
\newblock Tomita-Takesaki Modular Theory. 
\newblock\textit{\url{https://arxiv.org/abs/math-ph/0511034}} (2005).

\bibitem{S2016}
V. S. Sunder.
\newblock\textit{Operators on Hilbert space}.
\newblock Texts and Readings in Mathematics 71. Springer Science $+$ Business 
Media Singapore 2016 and Hindustan Book Agency (2016).

\bibitem{T2025}
M. Thill.
\newblock Introduction to Normed $\adj$-Algebras and their Representations, 7th 
ed.
\newblock\textit{\url{https://arxiv.org/abs/0807.4242}} (2025)

\bibitem{T1967}
F. Tr\`{e}ves.
\newblock\textit{Topological vector spaces, distributions and kernels}.
\newblock Pure and Applied Mathematics (Academic Press) 25. New York-London: 
Academic Press (1967).

\bibitem{V2025}
R. Verch.
\newblock Lecture Notes on Operator Algebras and Quantum Field Theory. EMS-IAMP 
Spring School ``Symmetries and Measurement in Quantum Field Theory - April 
7-11, 2025''.
\newblock\textit{\url{https://arxiv.org/abs/2507.00900}} (2025).

\bibitem{W2008}
G. Warner.
\newblock Positivity. 
\url{http://www.math.washington.edu/~warner/Positivity_Warner.pdf}. 
\newblock\textit{Unpublished} (2008).

\bibitem{W2010}
G. Warner.
\newblock $\tup{C}^\adj$-algebras. 
\url{http://www.math.washington.edu/~warner/C-star.pdf}. 
\newblock\textit{Unpublished} (2010).

\bibitem{W1980}
J. Weidmann.
\newblock\textit{Linear operators in Hilbert spaces. Transl. by Joseph Sz\"{u}cs}.
\newblock Graduate Texts in Mathematics, Vol. 68. Springer, New York (1980).

\bibitem{W2011}
D. Werner.
\newblock\textit{Functional analysis. 7th revised edition} (in German).
\newblock Springer, Berlin (2011).

\bibitem{W1987}
R. Werner. 
\newblock Local preparability of states and the split property in quantum field 
theory. 
\newblock\textit{Lett Math Phys 13, 325-329} (1987). 

\bibitem{Z2008}
C. Zwarich.
\newblock\textit{Von Neumann Algebras for Abstract Harmonic Analysis}.
\newblock M.Sc. thesis, University of Waterloo, Canada (2008).
\newblock Available via 
\url{https://uwspace.uwaterloo.ca/bitstreams/4b497e8c-779f-4c75-b6cf-b397c5f5c315/download}.

\end{thebibliography}
\end{document}